%% file: arxiv.tex
\documentclass[runningheads]{llncs}
\usepackage[T1]{fontenc}
\usepackage{graphicx}
\usepackage{url}
\usepackage{graphicx}
\usepackage{todonotes}
\usepackage{amsmath}
\usepackage{float}
\usepackage{tcolorbox}
\usepackage{xcolor}
\usepackage[n,advantage,operators,sets,adversary,landau,notions,logic,ff,mm,primitives,events,complexity,asymptotics,keys,probability]{cryptocode}

\usepackage{hyperref}
\usepackage{color}

\newcommand{\Comment}[1]{\textcolor{gray}{ \textsf{// #1}}}

\newcommand{\gray}[1]{\textcolor{gray}{#1}}
\newcommand{\white}[1]{\textcolor{white}{#1}}

\newcommand{\hidden}[1]{\textcolor{gray}{#1}}
\newcommand{\intent}[1]{\textcolor{gray}{#1}}

\newcommand{\alisanotes}[1]{{#1}}
\newcommand{\alisa}[1]{{#1}}
\newcommand{\lizanotes}[1]{{#1}}
\newcommand{\liza}[1]{{#1}}
\newcommand{\lizanew}[1]{{#1}}

\title{PRETTINESS - Privacy pResErving aTTrIbute maNagEment SyStem}

\newcommand{\env}{\mathcal{Z}}
\newcommand{\values}[1]{\mathcal{V}_{{#1}}}

\newcommand{\pin}{\mathsf{PIN}}
\newcommand{\uid}{\mathsf{uID}}

\newcommand{\aid}{\mathsf{aID}}

\newcommand{\aids}{\overrightarrow{\aid}}
\newcommand{\aidsd}{\aids_{\textsf{dscl}}}
\newcommand{\iid}{\mathsf{IID}}
\newcommand{\cid}{\mathsf{cID}}
\newcommand{\cids}{\overrightarrow{\cid}}

\newcommand{\rtag}{\mathsf{tag}}
\newcommand{\req}{\mathsf{req}}
\newcommand{\cred}{\mathsf{cred}}
\newcommand{\creds}{\overrightarrow{\cred}}
\newcommand{\pres}{\mathsf{pres}}
\newcommand{\pcred}{\pres}
\newcommand{\ok}{\mathsf{ok}}

\newcommand{\rev}{\mathsf{rev}}

\newcommand{\atr}{\mathsf{atr}}
\newcommand{\atrs}{\overrightarrow{\atr}}
\newcommand{\atrsd}{\atrs_{\textsf{dscl}}}
\newcommand{\ver}{\mathsf{Ver}}
\newcommand{\present}{\mathsf{Present}}

\newcommand{\id}{\mathsf{id}}

\newcommand{\store}{\mathsf{ST}}
\newcommand{\sid}{\mathsf{sid}}
\newcommand{\idfun}{\mathcal{F}}
\newcommand{\ifun}{\mathcal{F}_{\textsf{ideal}}}
\newcommand{\funs}{\mathcal{F}_{\textsf{sign}}}

\newcommand{\funsign}{\mathcal{F}_{\textsf{tsign}}}
\newcommand{\fundec}{\mathcal{F}_{\textsf{tdecrypt}}}
\newcommand{\funrev}{\mathcal{F}_{\textsf{rev}}}

\newcommand{\tiss}{\mathsf{T}_{\textsf{issue}}}

\newcommand{\funrevoke}{\funrev}

\newcommand{\system}{\pi_{\textsf{PRETTINESS}}}

\newcommand{\zenv}{\mathcal{Z}}
\newcommand{\sidset}{\mathcal{S}_{ID}}
\newcommand{\sids}{S_{ID}}

\newcommand{\tsign}{\mathsf{T}_{\textsf{sign}}}
\newcommand{\tenc}{\mathsf{T}_{\textsf{enc}}}

\newcommand{\tisign}{\mathsf{T}_{\textsf{Isign}}}

\newcommand{\ctr}{\textsf{ctr}}

\newcommand{\user}[1]{\mathcal{U}_{{#1}}}
\newcommand{\revoker}[1]{\mathcal{R}_{{#1}}}

\newcommand{\revoked}{X^{-}}

\newcommand{\leak}[2]{\mathcal{L}^{{#2}}_{{#1}}}
\newcommand{\tokens}{\mathcal{X}}
\newcommand{\parties}{\mathcal{P}}
\newcommand{\simm}{\mathcal{S}}
\newcommand{\corrupted}{\mathcal{C}}
\newcommand{\issuersim}{\textit{Issuer}}
\newcommand{\usersim}{\textit{User}}

\newcommand{\send}{\longrightarrow}

\newcommand{\pirkeys}{\mathcal{X}_{\mathsf{id}}}

\newcommand{\db}[1]{\mathsf{DB}_\mathrm{{#1}}}

\newcommand{\verdict}{\mathsf{blame}}

\newcommand{\mal}[1]{\overline{{#1}}}

\newcommand{\presentedv}{\mathsf{pres}(V)}
\newcommand{\revokedu}{\mathsf{notif}(U)}
\newcommand{\revsids}[1]{\mathsf{rev\_sids}({{#1}})}

\newcommand{\ch}{\mathsf{ch}}

\newcommand{\issuereq}{\rtag_{issue}}
\newcommand{\getcidsreq}{\rtag_{cids}}
\newcommand{\revokereq}{\rtag_{rev}}
\newcommand{\presentreq}{\rtag_{pres}}
\newcommand{\verifyreq}{\rtag_{verify}}

\newcommand{\hcred}{h_{\cred}}

\newcommand{\brokencons}{\mathsf{brokenCons}}

\newcommand{\symenc}[2]{\mathsf{SE}({#1},{#2})}
\newcommand{\symdec}[2]{\mathsf{SD}({#1},{#2})}
\newcommand{\foutput}[2]{(\mathsf{Output},{{#1}},{{#2}})}
\newcommand{\ciphertxtset}{\mathsf{C}}
\newcommand{\plaintxtset}{\mathsf{M}}

\begin{document}
\title{PRETTINESS - Privacy pResErving aTTrIbute maNagEment SyStem}
%
%
\author{Jelizaveta Vakarjuk\inst{1,2} \and Alisa Pankova\inst{1}\\
\email{\{jelizaveta.vakarjuk, alisa.pankova\}@cyber.ee}}
\institute{Cybernetica AS, Mäealuse 2/1, 12618 Tallinn, Estonia \and
Tallinn University of Technology, Akadeemia tee 15a, 12618, Tallinn, Estonia}
\authorrunning{Vakarjuk and Pankova}
%
%

%
\maketitle              

\begin{abstract}
 Recent advancements in the European regulatory frameworks (eIDAS2, GDPR) as well as increasing awareness of user privacy needs, have fostered a growing demand for identity management solutions that prioritise privacy. Attribute-based credential systems offer a solution, allowing users to prove ownership of their attributes to relying parties while selectively sharing specific attributes. While many systems have been proposed, and some systems have practical implementations, these solutions often store attributes locally on user devices.

In this work, we propose an attribute management system, where the user attributes are stored in encrypted form by the attribute management service (AMS). AMS helps to offload the storage from the user's mobile device and manage all the operations with credentials, including revocation. Our system is built by carefully balancing between the centralised and decentralised identity system approaches. The system relies on AMS, at the same time offering privacy against AMS, and ensuring integrity even if the AMS is corrupted. We prove security of the proposed system in the universal composability model. 
\end{abstract}





\section{Introduction}

With the changes in the regulatory framework of eIDAS2, there is an increasing demand in the research of identity management solutions with privacy guarantees. eIDAS2 defines the new European Digital Identity (EUDI) Wallet concept that aims to provide a user with the way to manage and present their credentials to the relying parties (RP), while trying to ensure interoperability of solutions between the EU member states. The regulation~\cite{regulation} lists security and privacy requirements for wallets, mentioning \emph{selective disclosure of data}, \emph{generating pseudonyms and storing them encrypted}, \emph{accessing a log of all transactions} (Article 5a, §4)
as well as \emph{not allowing any party, after the issuance of the attestation of attributes, to obtain data that allows transactions or user behaviour to be tracked, linked or correlated unless explicitly authorised by the user} and \emph{enabling privacy preserving techniques which ensure unlinkability} (Article 5a, §16). The currently proposed solution is defined in The Architecture and Reference Framework (ARF, \lizanew{version 2.7.3})~\cite{wallet-arf}. \lizanew{Additionally, the ARF defines wallet functionalities, including ability of creating backup of a list of user's attributes, attestations, and configurations as well as revocation by the issuer.} In our view, the currently proposed solution defined in The Architecture and Reference Framework (ARF), is lacking some of the functionality and security properties. Similar concerns were raised in the Cryptographers' Feedback on the EU Digital Identity’s ARF (\lizanew{version 1.4.0})~\cite{cryptographers}. However, we look at it from the different angle, trying to balance between the practical usability of the system and providing reasonable security and privacy guarantees. 

According to the system description, the user is responsible for managing and storing the attributes themselves on their (mobile) device. \alisa{However, to ensure a safe backup of credential data,}
it is beneficial to have an Attribute Management Service (AMS) that would be responsible for storing user's credentials. \alisa{Such a service would also} help with the process of issuing new credentials and revoking old ones. \lizanew{In this work, we introduce PRETTINESS -- Privacy pResErving aTTrIbute maNagEment
SyStem.} In design of \lizanew{Prettiness}, we are trying to find the balance between the centralised and decentralised identity management systems. We ensure that the user is in charge of their credentials, allowing selective disclosure and revocation by the user. At the same time, a centralised AMS server helps user with the storage and management of those credentials. \lizanew{Additionally, while details of logging user transactions are omitted from the description of our protocol, this functionality can be added to the system by saving user's signatures that are created to authorise each transaction with their credential.} \lizanew{One more detail that we omit from the current definition of the protocol is the use of pseudonyms. However, we note that this functionality can be added to the system through the user identifiers $\uid$ being used as pseudonyms.}

In this work, we use revocation functionality that works through AMS and allows RP to verify the validity of presented credential without the need to request any additional information from the issuer, \alisa{which is important to ensure that presentations cannot be tracked by the issuer}. On the other hand, the AMS does not learn which credential is being revoked \alisa{ or presented}. We allow two types of credential revocation: by the user and by the issuer, where revocation by the user does not require communicating with the credential issuer.

\alisa{In our system, the user should be able to decrypt stored credentials and generate signatures.} We ensure the protection of the cryptographic keys that need to be stored on the user's mobile device by using server-supported threshold cryptography for signing and decryption, \alisa{sharing the secret key between the user device and a supporting server. The same AMS can take the role of supporting server, storing the second key share and assisting in blind server-supported decryption and signing.}

We describe the resulting system in terms of an ideal functionality, construct a protocol that implements that functionality, and provide a security proof in the universal composability~\cite{UC} model.


This paper is organized as follows. The background and related work is given in Sec.~\ref{sec:background}. The building blocks that will be used by the proposed systems are described in Sec.~\ref{sec:bb}. The system itself is presented in Sec.~\ref{sec:main}. The efficiency of the system is evaluated in Sec.~\ref{sec:efficiency}. The results are concluded in Sec.~\ref{sec:conclusion}.

\section{Background}\label{sec:background}

\subsection{Related work}
There are several well-known attribute-based credential (ABC) schemes that maintain privacy of users, allowing users to prove possession of credentials without disclosing any additional information. Among these schemes are Idemix~\cite{idemix} and U-Prove~\cite{u-prove}. Idemix is provably unlinkable and unforgeable, and its security relies on the strong RSA problem. There are several real-world implementations of Idemix as PrimeLife\footnote{\url{http://primelife.ercim.eu/about/factsheet}, last visited December 2025}, and Yivi app with IRMA backend system\footnote{\url{https://docs.yivi.app/what-is-yivi/}, last visited December 2025}. U-Prove is more efficient compared to Idemix, as it can be implemented using elliptic curves, but it does not provide unlinkability. There is a number of other attribute-based systems that utilise bilinear pairings~\cite{atr-with-pairings,atr-with-pairings-2,atr-with-pairings-3}, where performing operations with pairings is more expensive than performing exponentiations on an elliptic curve.

Additionally, there is a line of works focusing on the cloud-based ABC systems, where a dedicated cloud service helps to present the users’ credentials to the relying party. This approach is useful when there are efficiency constraints on the user's device. In~\cite{towards-cloud} authors have proposed a scheme, where the computations have been outsourced to a semi-trusted cloud, and the main underlying cryptographic construction is proxy re-encryption. Subsequent work~\cite{breaking-fixing} focus on improvements of the proposed scheme and mitigating possible attacks.  

\lizanew{The Request For Comments (RFC)~\cite{sd-jwt} defines} a mechanism that allows the user to selectively disclose attributes from the JSON Web Token (SD-JWT) for the presentation. The idea behind the solution is to create hash commitments to the attribute values that can be later opened by the user by providing the attribute value itself with the salt used to create commitments.
\liza{Additionally, the RFC defines Key Binding JWT, where the user's signature over SD-JWT is 
\alisa{included.}
This signature servers as proof of possession of the user's private key to the verifier of the credential presentation. More precisely, the user is required to sign an issuer-signed JWT and corresponding disclosures selected for the presentation to the verifier. The verifier uses the user's public key contained in SD-JWT to verify this signature.}
This format is referenced in the EUDI Wallet ARF~\cite{wallet-arf} as one of a few suitable standardised formats for releasing electronic attestations of attributes. In this work, we follow this approach as it describes an option of selectively disclosing attributes. 

Another approach referenced in the EUDI Wallet ARF is mobile driver's license (mDL)~\cite{mDL,mDL-whitepaper}. It is a digital representation of physical driving license, stored in the user's mobile device. The format of mobile credential allows to use it for various types of user attributes, not only as a driving license. Similarly to JWT, it uses hash commitments with random salt to allow selective disclosure of attributes. A particular attribute management system CACS (Centralized Attribute Collection Service)~\cite{cacs} supports interaction of users, issuers, and relying parties, following mDL standards.

In this work, we follow the main idea of~\cite{cacs}, letting a centralized system help to manage user's credentials. However, while~\cite{cacs} assumes a trusted server that is allowed to learn the user's data, in this work we are protecting the user's data from the server, improving the way credentials are stored and handled by the server. We also provide an integration of revocation mechanism, which has been missing in~\cite{cacs}. \liza{The usage of centralised server is motivated by the following reasons:
\begin{itemize}
    \item EUDIW ARF defines the need in backup of a list of the user's attributes, attestations, and configurations. Thus, the attributes should be stored in a backup server by default.
    \item \alisa{The server can assist in managing revocations.}
    \item \alisa{We need a centralised server for server-supported decryption and signing.} This approach ensures that:
        \begin{itemize}
        \item \lizanew{The adversary cannot decrypt attributes even if the adversary copies the mobile device's memory in between decryption sessions;} 
        \item \lizanew{If the memory of the mobile device gets leaked during the decryption session, a special mechanism of the server-supported signature schemes (clone detection) allows to decrypt attributes for at most one presentation before the server stops to communicate.}
        \end{itemize}
\end{itemize}}

\subsection{System requirements}
The requirements that we pose on our system are mainly motivated by~\cite{reqs} and~\cite{cryptographers}. In~\cite{cryptographers} (2024), Carsten Baum et al. have given a feedback to the European Digital Identity ARF regulation, pointing out important properties that a system should achieve by law. The survey by Maximilian Richter et al.~\cite{reqs} (2023) is more technical, and it lists some properties of a system that are classified as "Essential" and "Important", and references to possible solutions that allow to achieve these properties.

Let us first go through the \emph{essential} properties of~\cite{reqs}. We mark with $+$ properties that we address in \textsf{Prettiness} system, and with $-$ that we do not.
\begin{itemize}
\item[+] \emph{Completeness.} A valid credential will be accepted.
\item[+] \emph{Soundness.} An invalid credential will be rejected.
\item[+] \emph{Unforgeability.} A valid credential can only be created by the corresponding issuer. This can be achieved by letting the issuer sign a credential.
\item[+] \emph{Freshness of the presentation.} Every presentation has to be created anew for every verification. This can be achieved by letting the verifier produce a challenge that the user should respond to (e.g. in the proof of key ownership).
\item[+] \emph{Impossibility to duplicate credentials.} No one (except the issuer) may create a valid copy of a credential. This includes the ability of one user to present the same credential from multiple devices. 
In this work, we propose to use a threshold key signature scheme with a clone detection mechanism~\cite{splitkey} which gives a designated semi-trusted server the ability to detect that the signing protocol was initiated from different devices.

We note that the problem of \emph{non-transferability} of the credentials is deeper, as we eventually cannot prevent a user from sharing their device with a friend. Using all-or-nothing approach~\cite{idemix} or one-time credentials that are not easy to obtain~\cite{clonewars} would be possible, but we do not discuss such possibilities in this work, and leave these as optional additional mechanisms.

\item[+] \emph{Revocation} is an essential feature that is difficult to implement efficiently while preserving privacy. An alternative approach to a separate revocation functionality is to set short validity time for the credentials, requiring credentials to be frequently re-generated. However, we believe that introducing separate revocation functionality is essential for the usability of the attribute management system in practice. \alisa{While the user should be responsible for managing and storing their
attributes, ARF also requires that the issuer should be able to revoke their issued credential at any time.} 
Assume that the user has a driver licence issued and stored in the system. After an event requiring the police to revoke driver's licence occurs, police should be able to revoke the user's driver's licence \alisa{by  contacting the issuer.}
\end{itemize}

Let us now go through non-essential, but \emph{important} properties of ~\cite{reqs}.
\begin{itemize}
\item[+] \emph{Selective disclosure}. The user should be able to present only certain attributes of a credential. In the survey of~\cite{ramic2024selective}, the methods for this are classified to hash-based, signature-based, and zero-knowledge based. In this work, we use a kind of hash-based disclosure. 

\item[-] \emph{Unlinkability of presentations}. The verifier should not be able to correlate several presentations of the same credential.
In this work, we are considering the setting where the relying party \emph{should} learn the user's identity, so we are not aiming for this property. \alisa{This is in accordance with (Article 5a, §16), as in this case the user \emph{explicitly} authorises linkability. We see our system in the context of legal transactions by officially recognised issuers, like presenting driving/fishing licence, or getting access to a building. We do not aim to use it in anonymous transactions (like shop purchases).}

\item[-] \emph{Repudiation}. RP should not be able to prove the authenticity of verified statements to a third party. In this work, we are not aiming for this property. \alisa{In contrary, we believe that we may need non-repudiation for later accountability of transactions (although we leave it out of scope of this paper).}
\end{itemize}

\subsection{Participants and Threat model}\label{sec:threatmodel}
In this section, we discuss the participants and the threat model for our system, including constraints on what a corrupted party should be allowed to see or do, motivated by the list of requirements on credential management systems~\cite{reqs}.

\underline{User} is the party to whom the attributes (and, hence, credentials) belong. A user \alisanotes{can use their mobile device} to get credentials from an issuer, present them to a relying party (RP), and revoke credentials previously issued to them. User can be \emph{dynamically} corrupted, and \emph{corruption level} $c_{\mathsf{U}}$ may increase. \lizanew{There are the following corruption models for the user:
\begin{itemize}
    \item $c_{\mathsf{U}} = 0$ (no corruption)
    \item Starting from $c_{\mathsf{U}} = 0$ (no corruption), the corruption level may rise to $c_{\mathsf{U}} = 1$ (partial corruption, after the key generation). This can advance further, moving to $c_{\mathsf{U}} = 2$ (partial corruption \emph{and} learning user's $\pin$)
    \item $c_{\mathsf{U}} = 3$ (full corruption, before key generation)
\end{itemize}}
The motivation behind introducing $c_{\mathsf{U}} = 1$ is that we want to ensure privacy of the user data even if its mobile device memory gets leaked \alisa{(e.g. the device gets malware)}, but the secret key
\lizanew{stored encrypted under the user's $\pin$} is not leaked. For $c_{\mathsf{U}} = 2$, \lizanew{the attacker learns all the device memory, including the the secret key}. This can happen if adversary guessed correct user's $\pin$. The final case, $c_{\mathsf{U}} = 3$ means that adversary has fully corrupted the user's device, including access to memory and ability to issue requests. In this case, the attacker learns everything, and there are no privacy guarantees for the user.

\lizanew{The adversary can corrupt the user before key generation, reaching level $c_{\mathsf{U}} = 3$. Otherwise, $c_{\mathsf{U}} = 0$ is assumed during the setup phase, and the adversary may increase the corruption levels later.} \lizanew{The adversary can clone device memory, reaching level $c_{\mathsf{U}} = 1$ (the main focus of security analysis).} \alisa{If the adversary guesses $\pin$, it proceeds to $c_{\mathsf{U}} = 2$, and there are no privacy guarantees for the user anymore.}

\lizanew{The reason for analysing our system in such corruption model is that static corruption of the user (from the beginning of the protocol) does not allow us to rely on specific properties introduced by server-supported signing protocol. To make use of the user's key protection mechanisms, we need to assume that the user has not been corrupted during the key generation and adversary can copy device's memory.}

\underline{Attribute management service (AMS)} is responsible for helping the user to manage storage and interactions with their credentials. AMS server can be \emph{statically} corrupted. In particular, we assume a \emph{covertly}~\cite{covert} (explained in more details in Sec.~\ref{sec:covert}) corrupted server, which may deviate from the prescribed routines, but will not do it if the probability of being caught is sufficiently large. Such an assumption is reasonable for AMS established by an institution which cares for its reputation.

AMS server is allowed to see credential issuance requests submitted by the user, including the issuer identifier (as AMS is responsible for establishing connection). It is also allowed to see the \emph{identifiers} of requested attribute, as these likely depend on the issuer identity anyway. However, AMS server should not learn the \emph{values} of attributes that are a part of issued credential.

If AMS has not seen the credential yet (e.g. from an issuer or a RP), it should not be able to learn when that credential has been revoked. Also, it should not be able to track which credentials have been presented, and to which RP \liza{(unless RP is also corrupted)}.

\underline{Issuer} is responsible for issuing credentials to the user, signed with the issuer's private key. Additionally, the issuer has authority to submit credential revocation requests to the AMS for the law enforcement purposes. Issuer can be \emph{statically} corrupted. The issuer should not be able to track presentations \liza{(unless colluded with RP)} and revocations by the user \liza{(unless colluded with AMS)} of its credentials. 

\underline{Relying Party (RP)} is a party requesting to see certain user’s attributes to authenticate user or to give access to certain services. RP can be \emph{statically} corrupted. From presentation, RP should not be able to learn anything except the attributes than the user has explicitly shown to it. However, we do not aim for unlinkability of presentations, as this can be tracked through the issuer's signature. \alisa{Differently from the anonymous credentials of Idemix~\cite{idemix}, we are assuming standard signatures like RSA and ECDSA (for compatibility with current standards), which cannot be updated for every presentation without the issuer. Nevertheless, we see the possibility of unlinkability using proofs of ECDSA signatures in zero-knowledge as in~\cite{frigo2024anonymous}. This addition is out of scope of this paper.} \lizanew{Similarly to the EUDIW defined in ARF~\cite{wallet-arf}, our solutions supports methods proposed in ARF (section 7.4.3.5.2) for dealing with RP linkability. For the \emph{once-only attestations}, user would need to run issuing protocol several times to obtain batch of the credentials (for the same attributes). Then each credential is used only once. For the \emph{limited-time attestations}, issuer will need to include additional entry to the credential, defining the validity period of credential. However, even with this method we still allow to use long-term attestations with revocation. For the \emph{rotating-batch attestations}, user would need to run issuing protocol several times to obtain batch of the credentials. During the presentation, user would randomly choose which credential from the batch to present. This method does not restrict credential to be used only once. In this case, random selection of credential for presentation and resetting the batch happens in the environment. Finally, \emph{Per-Relying Party attestations} method can be achieved through parametrising $\ifun$ with RP. Meaning that the user will use different instance of $\ifun$ for each RP. }

\subsection{\alisa{Covert Adversary}}\label{sec:covert}

In this work, we will make some assumption on the corruption of the AMS server. Since AMS is assumed to be a reputable entity, one option would be to assume a honest-but-curious server, which will not deviate from prescribed protocols. However, in practice it could be a too strong assumption, as even a reputable entity could be capable of silent cheating, which would not be detected, but could potentially harm user privacy. To ensure that such silent attacks are not possible, we make use of the notion of a \emph{covert} adversary, proposed by Aumann and Lindell~\cite{covert}. Intuitively, a covertly corrupted party will not cheat if it will be caught with a certain non-negligible probability.

In this work, we strengthen this assumption. We assume that the AMS will not cheat if it will be \emph{possible} to prove the fact of cheating to a third party. In particular, while AMS will be technically able to conceal the fact of revocation from the verifier, the revoker and the verifier together will be able to detect the problem and prove to a third party (e.g. external judge) that AMS has cheated.



\subsection{Universal Composability}

Universal composability (UC) model~\cite{UC} follows the real-ideal world paradigm. In the real protocol execution, honest parties run the protocol $\pi$ over a network controlled by an adversary $\adv$, who can also corrupt some parties by sending a special corruption message directly to the party. In the ideal world, honest parties are viewed as \emph{dummy parties} who forward their inputs directly to the ideal functionality $\mathcal{F}$ that internally executes the main task of the protocol, allowing the adversary to send certain corruption messages to a special interface of $\mathcal{F}$, whose effect on the outcome of $\mathcal{F}$ is specified explicitly. \alisa{In hybrid model, ideal functionalities can in turn be used by parties of the real protocol, and the adversary may communicate with them directly via adversary interface.}

There is an environment $\env$, representing the whole world outside of the protocol, which provides inputs to the parties and receives corresponding outputs. The protocol $\pi$ securely realises ideal functionality $\mathcal{F}$ if for every real adversary $\adv$ attacking the protocol there exists an ideal adversary (simulator) $\mathcal{S}$, such that no environment $\mathcal{Z}$ can distinguish whether it is interacting with the protocol $\pi$ and real adversary $\adv$ or with ideal functionality $\mathcal{F}$ and simulator $\simm$. Intuitively, any attack on the environment that is possible in the real world, is also possible in the ideal world.

\section{Building Blocks}\label{sec:bb}


\subsection{Signing functionality}\label{sec:bb:tsign}
For the digital signatures that are needed to sign issued attributes, we propose to use any of the standardised signature schemes like ECDSA, RSA. For the security analysis, we use digital signature ideal functionality $\funs$ proposed by Canetti \liza{~\cite{UC-cert}, depicted in Figure~\ref{fig:ideal-func-sign}. In contrast to the regular signing functionalities, it does not have the key generation interface. This functionality binds produced signatures to the identity of the signer (issuer). Intuitively, in the real world, it represents sending signatures with \emph{certificates} that bind the verification process to the identity of signer, not the public key.} \alisa{In this functionality, the signature $\sigma$ is generated by the adversary. For signatures generated outside of the ideal functionality by a corrupted signer, there are no verification guarantees, and the adversary may decide the outcome of verification, with the only requirement that the outcome of verification should be consistent with previous verifications.}

\begin{figure}[ht!]
    \centering
    \begin{tcolorbox}[colback=white,arc=0.3mm, boxrule=0.3mm]
    \underline{\textsf{Signature Generation}.} Upon receiving $(\textsf{sign}, \sid, m)$ from $S$, send $(\textsf{sign}, \sid, m)$ to the $\adv$. Upon receiving $(\textsf{signature}, \sid, m, \sigma)$ from the $\adv$, verify that no entry $(m, \sigma, 0)$ is recorded. If it is, then output an error message to $S$ and halt. Else, output $(\textsf{signature}, \sid, m, \sigma)$ to $S$, and record the entry $(m, \sigma, 1)$.

    \underline{\textsf{Signature Verification}.} Upon receiving a value $(\textsf{verify}, \sid, m, \sigma)$ from some party $P$, send $(\textsf{verify}, \sid, m, \sigma)$ to the $\adv$. Upon receiving $(\textsf{verified}, \sid, m, \phi)$ from the $\adv$, do:
    \begin{itemize}
        \item If $(m, \sigma, 1)$ is recorded then set $f = 1$.
        \item Else, if the signer is not corrupted, and no entry $(m, \sigma', 1)$ for any $\sigma'$ is recorded, then set f = 0 and record the entry $(m, \sigma, 0)$.
        \item Else, if there is an entry $(m, \sigma, f')$ recorded, then set $f= f'$.
        \item Else, set $f= \phi$, and record the entry $(m, \sigma', \phi)$.
    \end{itemize}
    Output $(\textsf{verified}, \sid, m, f)$ to $P$.
    \end{tcolorbox}
    \caption{The signing functionality $\funs$}
    \label{fig:ideal-func-sign}
\end{figure}

Moreover, we use \liza{blind} threshold signatures on the user's side to protect the storage of private key material. We use functionality $\funsign$ described in~\cite{blind-splitkey} for the two-party threshold \liza{blind} signature scheme with clone detection (functionality is presented in Figures~\ref{fig:fideal-sign-1}, ~\ref{fig:fideal-sign-2} in App.~\ref{app:functionalities}). Since the scheme assumes the use-case where one share of the private key is stored on the user's mobile device (client), client's share of the private key in storage is protected with the knowledge factor ($\pin$ code) required from the user for each signing session.
\alisa{The other share is stored on the supporting server side, and we assume that AMS takes this role.}
\liza{The usage of \emph{blind} threshold signature ensures that, for the signatures created for the verifiable presentation, supporting server does not learn which credential is being presented. In this scheme, the message is blinded only from the supporting server. We note that the blindness requirement is only important for the presentation phase. For the signatures that are used to initialise other requests, AMS server is still expected to see and verify resulting (unblinded) signature.}




The functionality $\funsign$ lists the following routines -- Key generation, Signing, Verification, Corruption. \lizanew{Additionally, $\funsign$ has some triggers that are used internally by functionality and not accessible to the environment.}

\underline{Key generation} is initiated by both client and server. $\funsign$ outputs
\lizanew{public key $\pk$}, which is allowed to be chosen by adversary.  

\underline{Signing} is initiated by 
\alisanotes{both client and server.} \liza{Message to be signed is input by the client, and the server does not know the message.}
$\funsign$ verifies correctness of $\pin$ code provided by the user and calls clone detection mechanism. If both verification checks are successful, $\funsign$ outputs signature $\sigma$ (\liza{only to the client}), that is allowed to be generated by adversary.

\underline{Verification} can be queried by any party to verify the correctness of the signature. \alisa{Verification procedure is similar to the one presented in regular signing functionality $\funs$. A corrupted client cannot produce more valid signatures than the number of signing sessions approved by the server.}

\underline{Corruption} allows adversary to corrupt parties throughout the protocol. \cite{blind-splitkey} defines several corruption levels for the client and server components. Corruption levels for the sever are trivial consisting of just two options $\mathsf{c}_\mathrm{S} = 0$  (uncorrupted) or $\mathsf{c}_\mathrm{S} = 1$  (corrupted -- adversary is being in full control of the server). Corruption of the client consists of more levels: 
\begin{itemize}
    \item $\mathsf{c}_\mathrm{P} = 0$ --uncorrupted
    \item $\mathsf{c}_\mathrm{P} = 1$ --adversary learned phone’s memory \emph{between} the signing sessions (encrypted memory). \lizanew{It means that all the sensitive values that are recomputed for signing are deleted from memory, and only the values in encrypted form or less sensitive values can be cloned by adversary.} 
    \item $\mathsf{c}_\mathrm{P} = 2$ --adversary learned phone’s memory \emph{during} a signing session (unencrypted memory). \lizanew{It means that adversary can access the values that are recomputed for signing, including private key share.}
    \item $\mathsf{c}_\mathrm{P} = 3$ --adversary has full control over the phone. \lizanew{It means that adversary learns all the updated values in the memory and can initiate interactions.}
\end{itemize} 
$\funsign$ is internally keeping track of the corruption levels of parties, and the source of the last signing query to support additional security properties (clone detection). 

\lizanew{\begin{table}[ht!]
    \centering
    \begin{tabular}{|p{1.7cm}|p{3cm}|p{4.4cm}|p{3.1cm}|} \hline
        \textbf{corruption level} & \textbf{explanation} & \textbf{consequences} & \textbf{how to achieve?} \\ \hline
         $\mathsf{c}_\mathrm{P} = 0$ & not corrupted & nothing is leaked to adversary & - \\ \hline 
         $\mathsf{c}_\mathrm{P} = 1$ & memory between signing sessions (encrypted memory) is leaked & adversary learns values stored in device's memory, including \emph{encrypted} share of the private key & clone memory of the device after key generation, in between signing sessions \\ \hline 
         $\mathsf{c}_\mathrm{P} = 2$ & memory during the signing session (unencrypted memory) is leaked & adversary learns values stored in device's memory as well as values generated during signing session, including \emph{unencrypted} share of the private key & clone memory of the device during signing process \textbf{OR} guess the user's $\pin$ after learning encrypted memory  \\ \hline 
         $\mathsf{c}_\mathrm{P} = 3$ & device is fully corrupted & adversary gains full access to the device, including learning its memory and being able to initialise some commands & get full access to the device, including during the key generation phase\\ \hline 
    \end{tabular}
    \caption{Corruption levels}
    \label{tab:corruption-levels}
\end{table}}

\lizanew{Since in \textsf{Prettiness}, we only assume static corruptions after the key generation phase, we do not need to use all the previously defined corruption levels for client. When client gets corrupted after key generation, the corruption level is set to $\mathsf{c}_\mathrm{P} = 1$. During the interaction with protocol adversary may guess client's $\pin$, which will increase corruption level to $\mathsf{c}_\mathrm{P} = 2$. Since we do not allow adaptive corruptions, adversary cannot get to level $\mathsf{c}_\mathrm{P} = 3$, which would mean gaining full control over the user's device. However, adversary can corrupt user to level $\mathsf{c}_\mathrm{P} = 3$ before the key generation, but in this case we cannot provide any security guarantees for the user.}




\subsection{Threshold decryption}
For \lizanew{decryption of the data stored at the AMS server},
we propose to use privacy-preserving server-supported decryption from~\cite{threshold-decryption}. The scheme ensures privacy against the server, meaning that when the client initiates decryption session, it sends a blinded ciphertext so the server is not able to learn which ciphertext exactly is being decrypted. We use functionality $\fundec$ defined in~\cite{threshold-decryption} (presented in Figures~\ref{fig:fideal-dec-1} and Figure~\ref{fig:fideal-dec-2} in App.~\ref{app:functionalities}).

\underline{Key generation} is initiated by the client and the server, its functionality is similar to functionalities for public-key encryption, \lizanew{defined by Canetti et al.~\cite{canetti-pke}}, where the public key $\pk$ is chosen by the adversary.

\underline{Encryption} can be queried by any party, providing a plaintext message $m$ as input. The ciphertext $c$ is created by adversary, who does not receive plaintext $m$ as an input. Pair $(m,c)$ is saved to the internal table $\mathsf{T}$.

\underline{Decryption} functionality has more nuanced details, but the main idea for $\fundec$ is to find the plaintext $m$ from a ciphertext $c$ that was created using the key $\pk$. This functionality is initiated by both parties, client and server with the client being the one who provides the ciphertext $c$ as input. A corrupted client cannot decrypt more ciphertexts than the number of decryptions approved by the server.

Additionally, we introduce 
a new interface to the functionality $\fundec$ -- verification of decryption. After receiving signed encrypted credential and corresponding disclosure, RP can use it to verify that the disclosure indeed corresponds to the ciphertext from signed credential. 

\underline{Verify decryption} requires inputs from the user and some other verifier party (in our case, the RP). Given user consent, $\fundec$ tells to the verifier whether the pair $(m,c)$ provided as input has been saved in table $\mathsf{T}$ of previously encrypted ciphertexts.

We propose to extend threshold decryption protocol from~\cite{threshold-decryption} with this interface as follows. In the protocol of~\cite{threshold-decryption} that UC-realises $\fundec$, the ciphertext is of the form $c = (\symenc{k}{m}, H'(k,m))$ for some hash function (random oracle) $H'$ where $\symenc{k}{m}$ is symmetric encryption of the message $m$ with the key $k$, the second component $H'(k,m)$ makes the decryption unique, and the single-use key $k$ is derived through KEM (key exchange mechanism) between the client and the server.
\begin{enumerate}
    \item Together with the ciphertext $c$, the user reveals the message $m$ and the single-use KEM key $k$ to decrypt $c$. 
    \item The verifier parses $c$ as $(c_1,c_2)$ and performs the following checks:    \begin{itemize}
        \item check that $m = \symdec{k}{c_1}$ (symmetric decryption of $c_1$ with $k$).
        \item check that $H''(k,m')=c_2$ (\alisanotes{note that this check does not allow the user to reveal a different $k$ and make $c$ open to a different message}).
    \end{itemize}
\end{enumerate}

\subsection{Revocation mechanism}\label{sec:bb:revocation}

For revocation of credentials, we use functionality of~\cite{funrevoke}. This functionality allows to extend the credential with a special \emph{revocation token} in such a way that only a designated revoking party will be able to revoke that token. A credential that is linked to a revoked token is also considered revoked. We will assign the revoker role to \emph{both} the user and the issuer, creating \emph{two} distinct revocation tokens for one credential. This ensures that the issuer will not learn that the user has revoked the credential issued by that issuer. We use functionality $\funrevoke$ defined in~\cite{funrevoke} (presented in Figures~\ref{fig:funrevoke:1} and~\ref{fig:funrevoke:2} in App.~\ref{app:functionalities}).

In this functionality, we have the server $S$ (in practice, the AMS server) who maintains the set $\revoked$ of revoked tokens, the revoker $R$ (the issuer or the user) who gets the issued token and may revoke it later, the user $U$ who has the right to query token revocation status, and the verifier $V$ (in practice the RP) who learns $x$ and whether $x \in \revoked$. 


\underline{Initialisation.} The adversary chooses the universal set $\tokens$ of tokens. 
The size of $\tokens$ is given as a parameter to $\funrevoke$, so it is not allowed to be too small.

\underline{Issuing.} $\funrevoke$ samples a fresh token $x \sample \tokens$ that it outputs to the user and to the revoker. In our case, the revoker party can be the user or the issuer. The issuer may then include $x$ into a credential as a special attribute (this will happen in the embedding protocol, outside of $\funrevoke$). 

\underline{Revocation.} A party $R$ may revoke $x$ if $\funrevoke$ has previously declared $R$ as revoker of $x$. A revoked token is added to the set $\revoked$ of revoked tokens.

\underline{User notification.} The user $U$ may ask from $\funrevoke$ whether a token, that belongs to the user $U$, has already been revoked. Such a notification would be called just before verification, to ensure freshness of revocation status.


\underline{Verifier notification.} 
$\funrevoke$ remembers the current state of revocation database for future reference, without leaking anything to $V$. Such a notification could be triggered just before verification, but having it as a separate routine can be more practical.

\underline{Verification.} The user presents a token $x$. The verifier learns $x$ and whether $x$ has been revoked at the moment when $V$ made its last \emph{verifier notification} request.

\underline{Cheating detection.} The functionality $\funrevoke$ proposes interfaces for cheating detection, which allow the revoker and the verifier together to figure out whether any revocation by that revoker has been silenced from the verifier. This interface is important, assuming that the covertly corrupted server will not cheat if it will be accountable for cheating. \alisa{We formalise this assumption in Definition~\ref{def:covert:formal}.}

\begin{definition}[covertly corrupted AMS]\label{def:covert:formal}
We say that AMS is \emph{covertly} corrupted if the adversary will not interact with the protocol in a way that allows an honest revoker $R$ and a honest verifier $V$ to make cheating detection interface of $\funrevoke$ report AMS as cheater. 
\end{definition}

\section{Attribute management system}\label{sec:main}
\label{sec:system}

\subsection{Ideal functionality}\label{sec:props}
We start with defining system functional and security requirements in the form of ideal functionality. The ideal functionality $\ifun$ is inspired by the previous works~\cite{delegatable,bind} with some adjustments to fit design principles and requirements,
following the model defined in Section~\ref{sec:threatmodel}. 
We can view $\ifun$ as an ideal credential storage, which stores issued credentials, and also allows to revoke them, and present to third parties. Let us use $\store$ to denote the storage, which in practice will be realised by a single AMS server. Ideal functionality provides the following interfaces -- \textsf{Initialisation}, \textsf{Issue}, \textsf{Present}, \textsf{Revoke}, \textsf{Verify}, \textsf{Corrupt}. \alisa{Every call to such an interface includes a session identifier $\sid$, which is needed to separate calls from each other. It does not contain any private data by itself, but leaking $\sid$ to $\adv$ allows it to match the session against a certain call from the environment.}


\underline{\textsf{Initialisation}} is used to initialise tables $\textsf{T}_a, \textsf{T}_r, \textsf{T}_p, \textsf{T}_v$,  for storing issued attributes, revocation information, credential presentations and verification information correspondingly. Additionally, $\mathrm{c}_{\mathsf{U}} = 0$ gets initialised, storing information about the corruption level of user $\mathsf{U}$, as in Section~\ref{sec:threatmodel}. 
Next, we let adversary define $\present$ and $\ver$ functions. $\present$ generates credential presentation for the presentation phase, and $\ver$ is the corresponding verification function. \alisa{We need these functions to be able to output a particular presentation to $\env$ to be shown later to a verifier, who can verify it using $\ifun$.} 

\underline{\textsf{Silent request}} is initiated by adversary at any point of time trying to query \textsf{Issue}, \textsf{Revoke} or \textsf{Present} interfaces on behalf of corrupted initiator (user or issuer). For a user, $\ifun$ firstly calls the authentication sub-routine which verifies user corruption level. If \emph{user is not fully corrupted ($\mathrm{c}_{\mathsf{U}} < 2$), adversary is not allowed to make those requests} to the $\ifun$. If the user is fully corrupted ($\mathrm{c}_{\mathsf{U}} \geq 2$), $\ifun$ proceeds with the requested interface using input provided by the adversary. Finally, $\ifun$ gives the user's output to the adversary. This way, even though the adversary will see all the inputs and outputs of partially corrupted users, it cannot initiate any routines on their behalf.
 
\underline{\textsf{Issue}} is initiated by the user query containing information about the requested attribute(s). If nobody is corrupted, adversary just gets notification that the communication is happening, and is able to interrupt the communication. We are not hiding issuer's identity and the issued $\aids$, from a corrupted storage ($\store$). This is a functional requirement: to make querying by credential identifier $\cid$ usable in practice, we need that it would be self-explanatory, \alisa{and we also need these values for the transaction log.} Actually, since in our model, $\store$ knows the issuer identifier, it will likely anyway be able to guess which kinds of attributes have been requested by the user.

The attribute request is sent to the issuer and, if not interrupted by the adversary, the issued credential consisting of attributes $\atrs$ is recorded to the table $\textsf{T}_a$. Every record in table $\textsf{T}_a$ serves as an issued credential. While credential identifier $\cid$ is chosen by the adversary, $\ifun$ checks that it is unique for an honest user \alisa{and an honest storage}. After attributes are received from the issuer, \emph{corrupted storage learns only $\cid$, but not the values of attributes}. 

\underline{\textsf{Revoke}} is initialised either by the user or issuer query containing identifier $\cid$ of the credential to be revoked. 
$\ifun$ verifies that a credential with $\cid$ exists. 
The adversary learns $\cid$ provided by an honest $P$ if either (1) the user of $\cid$ is corrupted and hence is able to trace revocation, or (2) $\store$ and some RP to which $\cid$ has been presented before, are both corrupted (in that case, adversary already learned the link between the credential and corresponding revocation token). Otherwise, $\ifun$ just notifies adversary that the revocation process has been initiated. \alisa{If the revoking party is an issuer, $\ifun$ also leaks to corrupted storage the identity of the user, owner of the revoked credential. In practice, this is needed purely for logging purposes, so that AMS could show to the user only those revocation logs which concern that user.}

After receiving approval from the adversary, $\ifun$ records that $\cid$ has been revoked by user or issuer in table $\textsf{T}_r$ and confirms successful revocation to the revoker.  

\underline{\textsf{Present}} is initiated by the user query containing information about the credential $\cid$ and attribute(s) $\aidsd$ that need to be presented. First of all, $\ifun$ verifies whether the requested credential is in the table $\textsf{T}_a$. 
$\ifun$ sends only the \alisa{user identifier} and the number of attributes to the adversary, 
even if AMS and the issuer are corrupted. 
If presentation is approved by the adversary, $\ifun$ verifies that credential has not been revoked, by consulting table $\textsf{T}_r$. 
If the check passes, $\ifun$ generates presentation for the requested credential using $\present$ function.
The obtained credential presentation is output to the user. 

\underline{\textsf{Verify}} is initiated by the user and RP. If the AMS server (with or without the issuer of presented credential) is corrupted, adversary does not learn neither RP identifier nor credential identifier. It helps to capture the fact that \emph{AMS server does not learn RP identifier or which credential is being presented by the user or which attributes are being disclosed during the presentation}. $\ifun$ verifies that presented credential has not been revoked, and has been recorded by the issuer. Note that, as far as no revocations took place between a presentation and a revocation, both checks always verify for the presentations generated through the $\ifun$. 

\underline{\textsf{Corrupt}} allows adversary to submit corruption requests identifying which party is being corrupted. User corruption is handled as described in Sec.~\ref{sec:threatmodel}.

\begin{figure}[htbp]
\begin{tcolorbox}[colback=white,arc=0.3mm, boxrule=0.3mm]
\underline{\textsf{Initialisation}.} On input $(\textsf{setup}, \sid)$ from $\adv$:

\begin{itemize}
    \item Create empty tables $\textsf{T}_a, \textsf{T}_r, \textsf{T}_p, \textsf{T}_v$, set flag $\mathsf{c}_{\mathsf{U}} = 0$.
    \item Send $(\textsf{setup-ok}, \sid)$ to $\adv$ and wait for  $(\textsf{setup-ok}, \sid, \ver, \present)$ from $\adv$.
\end{itemize}

\underline{\alisanotes{\textsf{Corrupted inputs}}.} At any time when $(\textsf{routine},\sid,x)$ is called from $\mathcal{Z}$ by a party $P \in \corrupted$, send $(\textsf{input},P,\textsf{routine},\sid,x)$ to $\adv$. Wait for $(\textsf{input},P,\textsf{routine},\sid,x')$ from $\adv$ and run \underline{\textsf{routine}} with input $x'$ instead of $x$.

\underline{\textsf{Silent request}.} At any time, on input $(\textsf{invoke},\textsf{routine},\sid,P,\textsf{inputs})$ from $\adv$:

\begin{itemize}
\item \alisanotes{If $\store \in \corrupted$, or $P = \mathsf{I}$, authentication is not required.}
\item Otherwise $\textsf{authenticate}(\sid,\textsf{U})$, proceed if successful.    
\end{itemize}

Wait for  $(\textsf{input},P,\textsf{routine},\sid,x)$ from $\adv$ and follow \underline{\textsf{routine}} interface $x$ as input of party $P$. Return output of \underline{\textsf{routine}} to $\adv$.

\underline{\textsf{Issue}.} On input $(\textsf{issue-req},\sid,\aids,\textsf{I})$ from \textsf{U}: 

\begin{itemize}
    \item If 
    $\textsf{I} \in \mathcal{C}$ or $\store \in \mathcal{C}$ send $(\textsf{issue-req},\sid,\alisa{\textsf{U},\textsf{I}},\aids)$ to $\adv$.\Comment{$\textsf{I}$ and $\store$ learn $\aids$}
    
    Otherwise, send $(\textsf{issue-req},\sid,\alisa{\textsf{U},\textsf{I}},|\aids|)$ to $\adv$.
    
    \item 
    Wait for $(\textsf{issue-ok},\sid,\alisanotes{\cid})$ from $\adv$. \alisa{Check that $\cid$ is unique.} 
    
    Send $(\textsf{issue-req},\sid,\alisa{\textsf{U}},\alisanotes{\cid},\aids)$ to \textsf{I}.
    \item Upon receiving $(\textsf{issue-atr},\sid,\atrs)$ from \textsf{I}:
    \begin{itemize}
    \item If 
    $\mathsf{U} \in \corrupted$ send  message $(\textsf{issue-atr},\sid,\atrs)$ to $\adv$.\Comment{$\textsf{U}$ learns $\atrs$}
    \item Otherwise, send message $(\textsf{issue-atr},\sid)$ to $\adv$.
    \end{itemize}
    
    \item 
    Wait for $(\textsf{atr-ok},\sid)$ response from $\adv$.
    
    \item Record $(\mathsf{U},\alisanotes{\mathsf{I}},\cid,\aids,\atrs)$ to table $\textsf{T}_a$. Send $(\textsf{issue-resp},\sid,\cid)$ to \textsf{U} \alisanotes{and \textsf{I}}. 
\end{itemize}

\underline{\alisanotes{\textsf{Get all credential identifiers}}.} On input $(\textsf{getcids-req},\sid)$ from $\mathsf{U}$:
\begin{itemize}
\item Send $(\textsf{getcids-req},\sid,\alisa{\mathsf{U}})$ to $\adv$. Wait for $(\textsf{getcids-ok},\sid)$ from $\adv$
\item Let $\cids$ be all identifiers stored in $\mathsf{T}_a$. Output $(\textsf{getcids-resp},\sid,\cids)$ to $\mathsf{U}$.
\end{itemize}

\underline{\textsf{Revoke}.} On input $(\textsf{revoke-req},\sid,\cid)$ from $P$:

\begin{itemize}

    \item Find an entry $(\mathsf{U},\alisanotes{\mathsf{I}},\cid,\_,\_) \in \textsf{T}_a$, s.t. $P \in \set{\mathsf{I},\mathsf{U}}$. 
    
    \alisa{If there is no such entry, send $(\textsf{revoke-resp},\sid,\bot)$ to the party $P$.} 
    
    \item Produce leakage to the adversary:
    \begin{itemize}
    \item If $\store \in \mathcal{C}$ and $(\textsf{RP},\cid) \in \mathsf{T}_p$ for $\textsf{RP} \in \mathcal{C}$,
    send $(\textsf{revoke-req},\sid,P,\cid)$ to $\adv$.

    \Comment{credential has been presented to the corrupted RP before}
    
    
    
    \item If $\textsf{U} \in \mathcal{C}$ send $(\textsf{revoke-req},\sid,P,\cid)$ to $\adv$. 
    
        \Comment{corrupted user learns if their credential is revoked}

    \item Otherwise, send $(\textsf{revoke-req},\sid,P,\alisa{\textsf{U}})$ to $\adv$.
    \end{itemize}

    \item 
    Wait for $(\textsf{revoke-ok},\sid)$ response from $\adv$. 

    \item If there is no entry $(\mathsf{U},\cid,P)$ in table $\textsf{T}_r$, add $(\mathsf{U},\cid,P)$ to the table $\textsf{T}_r$.
    
    \item Send $(\textsf{revoke-resp},\sid,\ok)$ to the party $P$.
\end{itemize}

\end{tcolorbox}
\caption{\alisanotes{Ideal functionality $\ifun$ (Part 1)}}
\label{fig:ideal-func-1}
\end{figure}


\begin{figure}[htbp]
\begin{tcolorbox}[colback=white,arc=0.3mm, boxrule=0.3mm]
\underline{\textsf{Present}.} On input $(\textsf{present-req},\sid,\cid,\aidsd)$ from \textsf{U}:

\Comment{values in $\aidsd$ are identifying attributes that will be disclosed to the RP}

\begin{itemize}

    \item Find an entry $\alisa{\cred=}(\mathsf{U},\alisa{\mathsf{I}},\cid,\aids,\alisa{\atrs}) \in \textsf{T}_a$ with $\aidsd \subseteq \aids$. 
    
    \alisa{If there is no such entry, send $(\textsf{present-resp},\sid,\bot)$ to \textsf{U}.}
    
    \item 
    Send $(\textsf{present-req},\sid,\alisa{\textsf{U}},|\aidsd|)$ to $\adv$. Wait for $(\textsf{present-ok},\sid)$ from $\adv$. 
        
        
    \item If $(\alisa{\mathsf{U}},\cid,\_) \in \textsf{T}_r$, send $(\textsf{present-resp},\sid,\bot)$ to \textsf{U}.
    
    \item 
    Generate $\pres = \present(\cred,\aidsd)$.    
    
    If $\ver(\cred,\pres) = 1$, add $(\cred,\pres,\lizanew{\aidsd},1)$ to $\textsf{T}_v$. 
    
    Otherwise, send $(\textsf{present-resp},\sid,\bot)$ to \textsf{U} 

    

    \item Send $(\textsf{present-resp},\sid,\pres)$ to \textsf{U}.
    
\end{itemize}


\underline{\textsf{Verify}.} On input $(\textsf{verify-req},\sid,\pres,\liza{\textsf{RP}})$ from \textsf{U}:

\begin{itemize}
    \item \alisa{If $\mathsf{U} \notin \corrupted$, find $(\cred=(\mathsf{U},\_,\cid,\_,\_),\pres,\aidsd,1) \in \textsf{T}_v$.} Add $(\mathsf{RP},\cid)$ to $\mathsf{T}_p$.

    \Comment{honest users only input presentations previously produced by $\ifun$}
    
    \alisa{If $\mathsf{U} \in \corrupted$, find $\cred=(\_,\_,\cid,\_,\_) \in \textsf{T}_a$ s.t. $\ver(\cred,\pres)=1$ and $(\_,\cid,\_) \notin \mathsf{T}_r$.} 

    \Comment{presentation of non-issued attributes cannot be forged}
    
    \Comment{revoked credentials cannot be presented}
    
    In both cases, if there is no such entry, send $(\textsf{verify-resp}, \sid, \bot)$ to $\mathsf{U}$ \alisa{and $\mathsf{RP}$}.

    \item Produce leakage to the adversary.
    \begin{itemize}
    \item If $\textsf{RP} \in \corrupted\ \text{and}\ \store \in \corrupted$, send $(\textsf{verify-req}, \sid, \mathsf{U}, \cid, \pres,\alisa{\aidsd},\sids)$ to $\adv$,\\ where
    $\sids \gets \set{\sid\ |\ (\textsf{revoke-req},\sid,\mathsf{U},\cid)\text{was input to $\ifun$}}$.
    
    \item Otherwise, if $\textsf{RP} \in \corrupted $ 
    , send $(\textsf{verify-req}, \sid, \mathsf{U},\alisanotes{\cid},\pres,\alisa{\aidsd})$ to $\adv$.
    
    \item Otherwise, send $(\textsf{verify-req}, \sid, \mathsf{U})$ to $\adv$.
    \end{itemize}
    
    \item  Wait for $(\textsf{verify-ok},\sid,\liza{\ch})$ from $\adv$. \alisa{If $\textsf{RP} \notin \corrupted$, check $\ch \in \mathcal{CH}$}. 
    
    
    
        

    \item \liza{Set $\pres' \gets (\pres, \ch)$. Set $\mathcal{CH} \gets \mathcal{CH} \setminus \{\ch\} $.}
    \item Send $(\textsf{verify-resp},\sid,\pres',v)$ to \textsf{RP} and to \textsf{U}. 
\end{itemize}

\underline{\textsf{Corrupt.}} 

\begin{itemize}
    \item     \Comment{corruption in the beginning of the protocol} 
    
    On input $(\textsf{corrupt-party}, P, \sid)$ from $\adv$ where $P \in \{ \textsf{I, RP}, \store, \textsf{U} \}$,
    
    set $P \in \mathcal{C}$ (if $P = \textsf{U}$, set $\mathsf{c}_{\mathsf{U}} = 3$) and send $(\textsf{corrupted},P,\sid)$ to $\adv$.
    \item  \Comment{corruption after initialisation for user}  
    
    On input $(\textsf{corrupt-party}, \mathsf{U}, \sid)$ from $\adv$:
    \begin{itemize}
        \item if $\mathsf{c}_{\mathsf{U}} = 0$, set $\mathsf{c}_{\mathsf{U}} = 1$ \alisanotes{and $\mathsf{U} \in \mathcal{C}$}; 
        \item if $\mathsf{c}_{\mathsf{U}} = 1$, set \alisanotes{$\mathsf{c}_{\mathsf{U}} = 2$ and send $(\textsf{full-leakage},\sid,\mathsf{U},(\textsf{T}^{\mathsf{U}}_a, \textsf{T}^{\mathsf{U}}_r))$ 
        to $\adv$, where $\mathsf{T}^{\mathsf{U}}$ is projection of table $\mathsf{T}$ to the user $\mathsf{U}$.} 
    \end{itemize}

\end{itemize}

\textsf{authenticate}($\sid,\mathsf{U}$):

\begin{itemize}
    \item Send $(\textsf{auth},\sid)$ to $\adv$ and wait for $(\textsf{auth-ok},\sid)$.
    \item If $\mathsf{c}_{\mathsf{U}} \in \{ 0,1\}$, ignore this query.  
    
    \Comment{with partial corruption, no queries allowed}     
    \item If \alisanotes{$\mathsf{c}_{\mathsf{U}} \in \{2,3\}$}, proceed with the protocol.
\end{itemize}

\end{tcolorbox}
\caption{Ideal functionality $\ifun$ (Part 2)}
\label{fig:ideal-func-2}
\end{figure}

\paragraph{Properties of $\ifun$.}\label{sec:props} 
In Table~\ref{tab:leakage-system}, we summarize which values are leaked to the adversary depending on the corruption and the queried interface of $\ifun$. \liza{During the issuing, supporting server learns attribute values only if it colludes with issuer. Otherwise, it learns only credential identifier and attribute identifiers. During the revocation, AMS server learns identifier of revoked credential only if it colludes with the RP to which this credential has been presented before. During the presentation, AMS server learns disclosed attributes only if it colludes with RP. The only coalition that can learn all the attributes from presented credential consists of both the issuer and the RP, since the issuer has all the information available about the user's attributes, \alisa{and the RP sees when the credential is being presented}. \alisa{Only AMS and the RP together may trace revocation status of presented credentials.}}



We claim that functionality $\ifun$ satisfies the following properties, formally stated and proved in below.

\begin{table}[t]
    \centering
    \begin{tabular}{l|p{0.7cm}|p{0.7cm}|p{0.7cm}|p{1.3cm}|p{2.0cm}|p{2.0cm}|p{0.8cm}}
         & \multicolumn{3}{c|}{\textbf{Issue}}
         & \textbf{Revoke}
         & \multicolumn{3}{c}{\textbf{Present}}\\ \hline 
         \textbf{corrupted parties} & $\cid$ 
         & $\aids$ 
         & $\atrs$ 
         & $\cid$ 
         & $\aids$ 
         & $\atrs$ 
         & $\cid$ 
         \\ \hline \hline
        issuer & \intent{yes} & \intent{yes} & \intent{yes} & no & no & no & no \\ \hline
         AMS & \intent{yes} & \intent{yes} & no & no & only $|\aidsd|$  & no & no \\ \hline
        RP & no & no & no &  no & \intent{only $\aidsd$} & \intent{only $\atrsd$} & yes \\ \hline
        issuer + RP & \intent{yes} & \intent{yes} & \intent{yes} & no & yes$^*$ (via $\cid$) & yes$^*$ (via $\cid$) & yes \\ \hline
         AMS + issuer & \intent{yes} & \intent{yes} & \intent{yes} & no & only $|\aidsd|$ & no & no \\ \hline
        AMS + RP & \intent{yes} & \intent{yes} & no & yes$^*$ & yes$^*$ (via $\cid$) & only $\atrsd$ & yes \\ \hline
        AMS + issuer + RP & \intent{yes} & \intent{yes} & \intent{yes} & yes & yes & yes & yes \\ \hline
    \end{tabular}
    \caption{User data disclosed to corrupted parties, where $\cid$ and $\aids$ are credential and attribute identifiers, $\atrs$ are attribute values, $\aidsd$ are the attributes used in a presentation, and $\atrsd$ is the disclosure of attribute values coming from that presentation. Disclosures marked by $^*$ are possible only for the given coalition of corrupted parties, and are \emph{not} derived from smaller coalitions. Highlighted disclosures (\intent{yes}) are intentional and needed for functionality.}
    \label{tab:leakage-system}
\end{table}

\begin{itemize}
\item The user and the AMS server cannot convince RP of possession of attributes that have not been issued to the user \alisa{using issuing interface of $\ifun$}. 
\item The user cannot convince RP of possession of the credential that has been revoked by the user themselves or by the issuer of that credential. 
\item An adversary who has only partially corrupted the user cannot make issuing, presentation, or revocation requests. 
\end{itemize}

The first two items are covered by Theorem~\ref{thm:two}. The third item is covered by Theorem~\ref{thm:three}.





\begin{theorem}
\label{thm:two} An adversary $\adv$ and an environment $\mathcal{Z}$ running in parallel with $\ifun$, are not able to make a relying party accept an output \alisa{$(\mathsf{verify\text{-}resp},\sid,(\pcred,\ch),1)$} if there is no $\cred=(\_,\_,\cid,\_,\_) \in \mathsf{T}_a$ that would satisfy $\ver(\cred,\pres)$ and $(\_,\cid,\_) \notin \mathsf{T}_r$.
\end{theorem}

\begin{proof}
Let us assume an adversary $\adv$ produced $\pres$, such that a relying party received an output \alisa{$(\textsf{verify-resp},\sid,(\alisanotes{\pres},\ch),\aidsd,1)$}, but there is no $\cred=(\_,\_,\cid,\_,\_) \in \mathsf{T}_a$ s.t. $\ver(\cred,\pres)=1$ and $(\_,\cid,\_) \notin \mathsf{T}_r$.

\alisa{If $U \in \corrupted$, then the checks $\cred=(\_,\_,\cid,\_,\_) \in \mathsf{T}_a$, $\ver(\cred,\pres)=1$, and $(\_,\cid,\_) \notin \mathsf{T}_r$ are performed explicitly.}

\alisa{If $U \notin \corrupted$,} then $\pres$ was produced as a response to the presentation command by $\ifun$. This implies that functionality has verified that there exists $\cred \in \mathsf{T}_a$ such that $(\_,\cid,\_) \notin \mathsf{T}_r$ before creating $\pres$. Passing subsequent verification check implies that $\ver(\cred,\pres)=1$. Altogether, this contradicts the assumption.



\end{proof}

\begin{theorem}
\label{thm:three} An adversary $\adv$ and an environment $\mathcal{Z}$ running in parallel with $\ifun$, corrupting user for at most $\mathsf{c}_{\mathsf{U}} = 1$, are not able to issue requests (issuing, user revocation, presentation) to the ideal functionality.
\end{theorem}

\begin{proof}
    Assume that $\adv$, who partially corrupted user ( $\mathsf{c}_{\mathsf{U}} = 1$) makes request on behalf of the user to the issuing, user revocation or presentation  interface. Each of the interfaces firstly calls \textsf{authenticate} sub-routine with the user identifier $\mathsf{U}$ as input. To proceed with the interface instructions, it should hold that either the user is fully corrupted by the adversary ($c_{\mathsf{U}} = 2$), or the interface has been called by the honest user, not the adversary, which contradicts the assumption. 
\end{proof}

\subsection{The protocol set}
In this section, we present a set of protocols of the proposed attribute management system. 
The AMS server uses 
$\funrevoke$ to manage revocations. 
It also serves as the supporting server for threshold decryption ($\fundec$) and signing ($\funsign$).

\alisa{The notation used in this section is given in Table~\ref{tab:notation}. The real protocol routines will need to internally store some additional values, which are not a part of input/output interface. 
For clarity, we \hidden{highlight} such values.}

\begin{table}[ht!]
    \centering
    \begin{tabular}{c | p{7cm}}
        \textbf{Notation} & \textbf{Explanation} \\ \hline \hline
         $\aid$ & attribute identifier \\ \hline
         $\aids = (\aid_1, \dots ,\aid_n)$ & list of attribute identifiers \\ \hline
         $\atr$ & attribute value \\ \hline
         $c_{\atr}$ & encrypted attribute value \\ \hline
         $\atrs = (\atr_1, \dots ,\atr_n)$ & list of attribute values \\ \hline
         $\cid$ & credential identifier \\ \hline
         $\cids = (\cid_1,\ldots,\cid_n)$ & list of credential identifiers \\ \hline
         $\cred = \langle (\aid_1, c_{\atr_1}), \dots (\aid_n, c_{\atr_n}) \rangle$ & credential \\ \hline
         $\pcred$ & credential presentation \\ \hline
         $\uid$ & user identifier \\ \hline
         $\iid$ & issuer identifier \\ \hline
         $\idfun(cmd,\sid,x)$ & run the subroutine $cmd$ of $\idfun$ with $\sid$ and input $x$ \\ \hline
    \end{tabular}
    \caption{Notation for Prettiness system}
    \label{tab:notation}
\end{table}

\subsubsection{Initialisation}
Parties initialise building block functionalities \lizanew{as well as perform generation of keypairs as depicted on Figure~\ref{fig:prettiness-init}.} There is an instance of \lizanew{$\funs^{\iid}$} for an issuer identified by $\iid$, and of \lizanew{$\funsign^{\uid}$, $\fundec^\uid$} for a user identified by $\uid$. \alisa{Technically, the issuers inform to the verifiers their public keys, and each user generates keys for signing and decryption, shared with the supporting server, getting a $\pin$ that it can use to access the signing functionality later.} Revocation functionality $\funrev$ is initialised by AMS for all users altogether. 

\begin{figure}[htbp]
\centering
\begin{pchstack}
\begin{pcvstack}
\pseudocode[head={$\textsf{Initialise}_{\mathsf{U}}$}]{
\Comment{get input from environment}\\
\env \send: (\textsf{initialise-req}, \sid);\\
\Comment{generate signing key pair}\\
\pk \gets \funsign^{\uid}(\textsf{keygen},\sid,\hidden{L,T_0,\pin}) \\
\Comment{generate decryption key pair}\\
\pk_{\textsf{c}} \gets \fundec^{\uid}(\textsf{KeyGen},\sid) \\
}
\end{pcvstack}
\pchspace
\begin{pcvstack}
\pseudocode[head={$\textsf{Initialise}_{\store}$}]{
\Comment{generate signing key pair}\\
\pk \gets \funsign^{\uid}(\textsf{keygen},\sid,\hidden{T_0}) \\
\Comment{generate decryption key pair} \\
\pk_{\textsf{s}} \gets \fundec^{\uid}(\textsf{KeyGen},\sid) \\
}
\end{pcvstack}
\end{pchstack}
\caption{Initialisation protocol (User, $\store$)}
\label{fig:prettiness-init}
\end{figure}

\paragraph{User authentication.} After the initialisation, each query to the AMS server by the user starts from authentication. 
For this, user creates a threshold signature $\sigma_{\req}$, using \alisa{$\funsign^{\uid}$}.
Threshold signature generation using $\funsign^{\uid}$ already includes user authentication to the server. \alisa{Technically, going into details of the protocol implementing $\funsign$, the user proves ownership of their secret key share, as well as a special randomness, which serves as a challenge for the next authentication}. Obtained signature serves as a consent from the user to perform certain actions with their data. \lizanew{The usage of threshold signatures instead of regular signatures is important to ensure user's key protection. Since the user's private key is stored on the mobile device, usage of threshold signatures adds level of protection to the private keys. Since the signature also servers as a consent to perform certain actions with the user's credentials, this signature should be verifiable for audit purposes, hence the security of the corresponding private key is crucial.}




\subsubsection{Issuing}


The protocol is presented in Figure~\ref{fig:prettiness-issue}. Upon successful user authentication, AMS server forwards attribute request to the issuer. \alisa{The request contains a fresh identifier $id_\cid$ for the credential. To avoid repetition attacks by a corrupted AMS, the issuer keeps all received signatures and tracks their repetitions.}

\alisa{AMS provides to the user a link that they can use to directly communicate with the issuer $\iid$ (i.e. links $\iid$ to a communication channel of $I$), which is needed to generate revocation tokens using $\funrevoke$. There will be two such tokens, one for the user, and one for the issuer. First, the user generates a token $\rev_U$ whose revoker will be the user. Then, the user and the issuer jointly generate a token $\rev_I$, whose revoker will be the issuer.}

\alisa{For each requested attribute $\aid_i$, the issuer comes up with a value $\atr_i$.} The issuer creates the credential by encrypting each $\atr_i$ under the user's public key of $\fundec^{\uid}$, including revocation token\alisa{s $\rev_I$ and $\rev_U$}. Next, the issuer signs the credential using $\funs^{\iid}$ and forwards it to the AMS server. AMS server verifies the issuer signature and forwards issued credential to the user for approval. User decrypts attributes from the credential, verifies that the attributes are not malformed \alisa{(i.e. that the value $\atr_i$ belongs to the set of valid attribute values $\values{\aid_i}$)}, that revocation token data is correct, and the \liza{correctness of} issuer signature. If all the verification steps succeed, AMS server saves the issued credential with identifier $\cid = \alisa{(id_{\cid},\textsf{summary}(\uid,\aids,\iid))}$, where  $id_{\cid}$ ensures uniqueness of $\cid$, and function $\textsf{summary}$ generates an intuitive short description of the credential. \alisa{The user includes credential data into a small checksum digest that it maintains locally.}

\begin{figure}[htbp]
\centering
\begin{pchstack}
\begin{pcvstack}
\pseudocode[head={$\textsf{Issue}_{\mathsf{U}}$}]{
\Comment{get input from environment}\\
\alisa{\env \send: (\textsf{issue-req}, \sid, \aids, \iid);}\\
\alisa{id_{\cid} \sample \pirkeys};\\
\text{\Comment{establish connection with the server}} \\
\alisa{\send \store: \issuereq, \sid, \uid;}\\
\text{\Comment{generate request message}} \\
m = (\issuereq,\alisa{id_{\cid}},\uid,\aids,\iid); \\
\text{\Comment{generate threshold signature}} \\
\sigma_{\req} \gets \funsign^{\uid}(\textsf{sign},\sid,m,\hidden{\pin}); \\
\alisa{\longrightarrow \store : \liza{\sigma_{\req},m};}\white{\aids}\\
\\
\Comment{Get connected to the issuer $I$}\\
\alisa{\store \longrightarrow: I, \iid;}\\
\Comment{locally generate a revocation token for the user}\\
\alisa{\rev_U \gets \funrev(\textsf{issue},U,(\sid,0));\white{\aids}} \\
\alisa{\longrightarrow I: \rev_U;}\\
\Comment{jointly generate revocation token for the issuer}\\
\alisa{\rev_I \gets \funrev(\textsf{issue},\alisa{I},(\sid,1));} \\
\Comment{get encrypted credential from AMS}\\
\store \longrightarrow: (\cred,\sigma_{\cred}); \\
\alisanotes{\ok \gets \funs^{\iid}(\textsf{verify},\sid,\cred,\sigma_{\cred})};\\
\text{\Comment{decrypt and verify content}} \\
\alisa{(\rev'_I,\rev'_U)} \gets \fundec^{\uid}(\textsf{Decrypt},\sid,c_{\rev}); \\
\pcassert \alisa{(\rev_I,\rev_U) = (\rev'_I,\rev'_U)};\\
\pcfor \aid_i \in \aids \pcdo: \\
\quad \atr_i \gets \fundec^{\uid}(\textsf{Decrypt},(\sid,i),c_{\atr_i}); \\
\quad \pcassert \atr_i \alisa{\in \values{\aid_i}}; \\
\liza{\pcassert \atr_0 = \pk_\uid;} \\
\alisa{\cid \gets (\id_{\cid},\textsf{summary}(\uid,\aids,\iid));}\\
\alisa{\hidden{\hcred} \gets H(\hidden{\hcred}, (\cid,\cred,\sigma_\cred))};\\
\alisa{\send \env: (\textsf{issue-resp}, \sid, \cid);}
}
\end{pcvstack}
\pchspace
\begin{pcvstack}
\pseudocode[head={$\textsf{Issue}_{\store}$}]{
\\
\white{\aids}\\
\\ \Comment{wait for user connection}\\
\alisa{U \send: \issuereq, \sid, \uid;}\\
\\
\white{\aids}
\\
\text{\Comment{generate and verify threshold signature}} \\
\liza{\ok} \gets \funsign^{\uid}(\textsf{sign},\sid); \\
\alisa{\textsf{U} \longrightarrow: \liza{\sigma_{\req},m=(\issuereq,id_{\cid},\uid,\aids,\iid)};}\\
\liza{\ok \gets \funsign^{\uid}(\textsf{verify},\sid,m,\sigma_{\req});} \\
\text{\Comment{authenticate $I$ as $\iid$ to $U$}} \\
\alisa{\longrightarrow \textsf{U}: I, \iid;}\\
\text{\Comment{authenticate $U$ as $\uid$ to $I$}} \\
\alisa{\longrightarrow \textsf{I}: U, \alisa{id_{\cid}},\uid,\aids, \sigma_\req;} \\
\Comment{get encrypted credential from $I$}\\
\textsf{I} \longrightarrow : (\cred,\sigma_{\cred}); \\
\alisanotes{\ok \gets \funs^{\iid}(\textsf{verify},\sid,\cred,\sigma_{\cred})};\\
\Comment{send encrypted credential to $U$}\\
\longrightarrow \textsf{U} : (\cred,\sigma_{\cred}); \\
\\
\text{\Comment{participate in threshold decryption}} \\
\ok \gets \fundec^{\uid}(\textsf{Decrypt},\sid);\white{\rev'} \\
\\
\alisanotes{\pcfor \aid_i \in \aids \pcdo:}\\
\alisanotes{\quad ok \gets \fundec^{\uid}(\textsf{Decrypt},(\sid,i))}; \\
\\ \\
\alisa{\cid \gets (\id_{\cid},\textsf{summary}(\uid,\aids,\iid))};\\
\alisa{\db{S}[\uid,\cid] \gets (\cred,\sigma_{\cred});}\\
}
\end{pcvstack}
\end{pchstack}
\caption{Issuing protocol (User, $\store$)}
\label{fig:prettiness-issue}
\end{figure}

\begin{figure}[htbp]
\centering
\begin{pchstack}
\begin{pcvstack}
\pseudocode[head={$\textsf{Issue}_{\mathsf{I}}$}]{
\text{\Comment{receive an issue request from the server}} \\
\store \longrightarrow: \alisa{\sid, U, id_{\cid}},\uid,\aids, \sigma_\req; \\ 
\text{\Comment{verify the request}} \\
ok \gets \alisanotes{\funsign^{\uid}(\textsf{verify}, \sid, m=(\issuereq,\alisa{id_{\cid}},\uid,\aids,\iid), \sigma_{\req})}\\
\text{\Comment{verify freshness of the credential identifier}} \\
\cid \gets \alisa{(id_{\cid},\textsf{summary}(\uid,\aids,\iid))};\
\alisanotes{\pcassert \cid \notin \cids};\\
\alisanotes{\cids \gets \cids \cup \cid}\\
\text{\Comment{generate revocation tokens for the user and the issuer}} \\
\alisa{U \send: \rev_U;} \\
\alisa{\rev_I} \gets \funrev(\textsf{issue},\alisa{U},(\sid,1));\\
\Comment{request attribute values from the environment}\\
\alisa{\send \env: (\textsf{issue-req}, \sid, \uid, \cid, \aids);}\\
\alisa{\env \send: (\textsf{issue-atr}, \sid, \atrs);}\\
\Comment{verify attribute values obtained from the environment}\\
\alisa{\pcassert(\atr_1 \in \values{\aid_1},\ldots,\atr_n \in \values{\aid_n});}\\ 
\text{\Comment{encrypt revocation tokens and the user public key}} \\
c_{\rev} \gets \fundec^{\uid}(\textsf{Encrypt},\sid,\alisa{(\rev_I,\rev_U)}); \\
\liza{\atr_0 \gets \pk_\uid}; \\
\liza{c_{\atr_0} \gets \fundec^{\uid}(\textsf{Encrypt},\sid,\atr_0)}; \\
\Comment{encrypt attribute values}\\
\pcfor \aid_i \in \aids \pcdo\\
\quad c_{\atr_i} \gets \fundec^{\uid}(\textsf{Encrypt},(\sid,i),\atr_i) \\
\Comment{compose the credential}\\
\cred = \langle \liza{(\aid_{0}, c_{\atr_0})}, (\aid_1, c_{\atr_1}), \dots (\aid_n, c_{\atr_n}), (\aid_{\rev}, c_{\rev}) \rangle \\
\sigma_{\cred} \gets \funs^{\iid}(\textsf{sign},\sid,\cred); \\
\longrightarrow \store : (\cred,\sigma_{\cred});\\
\Comment{remember revocation token}\\
\alisa{\db{I}[\cid] = \rev_I;}\\
\alisa{\send \env: (\textsf{issue-resp}, \sid, \cid);}\\
}
\end{pcvstack}
\end{pchstack}    
\caption{Issuing protocol (Issuer)}
\label{fig:prettiness-issue-2}
\end{figure}

\subsubsection{Get all $\cid$-s} The protocol of Figure~\ref{fig:prettiness-getcids} allows to fetch all the previously issued credential identifiers $\cids$, as well as all the stored credentials $\creds$. While user has already received all this data during the issuing phase, the server storage can be used as a backup. Verification of the received data against the previously stored checksum ensures integrity against corrupted AMS. While issuer's signature already protects integrity of a credential, the checksum prevents AMS from being silent about the existence of some credential. At this step, the user only downloads \emph{encrypted} credentials, and will need access to the decryption functionality to get their contents.

Since we want the server to serve as a backup, we need to be sure that user has all relevant data to further \emph{use} the credentials. Besides raw credential data, the user also needs the  following protocol-specific data:
\begin{itemize}
\item The ciphertexts of $\fundec^{\uid}$ should be decryptable, and the signatures of $\funsign^{\uid}$ generatable. We assume that the user has not lost their secret keys corresponding to these functionalities.
\item To keep integrity against the corrupted server, the user needs to keep the small checksum digest. As this digest is not sensitive (it it is computed from encrypted credentials), it can be copied to any other device.
\item The revocation protocol of~\cite{funrevoke} that realises functionality $\funrevoke$ assumes that the user knows a small secret $r$ that corresponds to the token $x$, which is needed for both revocation and presentation. Technically, $r$ can also be stored on the server in form of a ciphertext $c_r$, which can (but does not have to) be included into the credential. This additional overhead has been included into the benchmarks of Sec.~\ref{sec:efficiency}.

\item The user, acting as a revoker, should be able to prove that AMS has ignored their revocation. \alisa{In the revocation protocol of~\cite{funrevoke},} the user should be able to present a confirmation signature of AMS. Such signatures cannot be stored by AMS. Here we assume that the number of such emergency-revoked credentials is not large for a single user, so they can afford to store them elsewhere.
\end{itemize}

\begin{figure}[htbp]
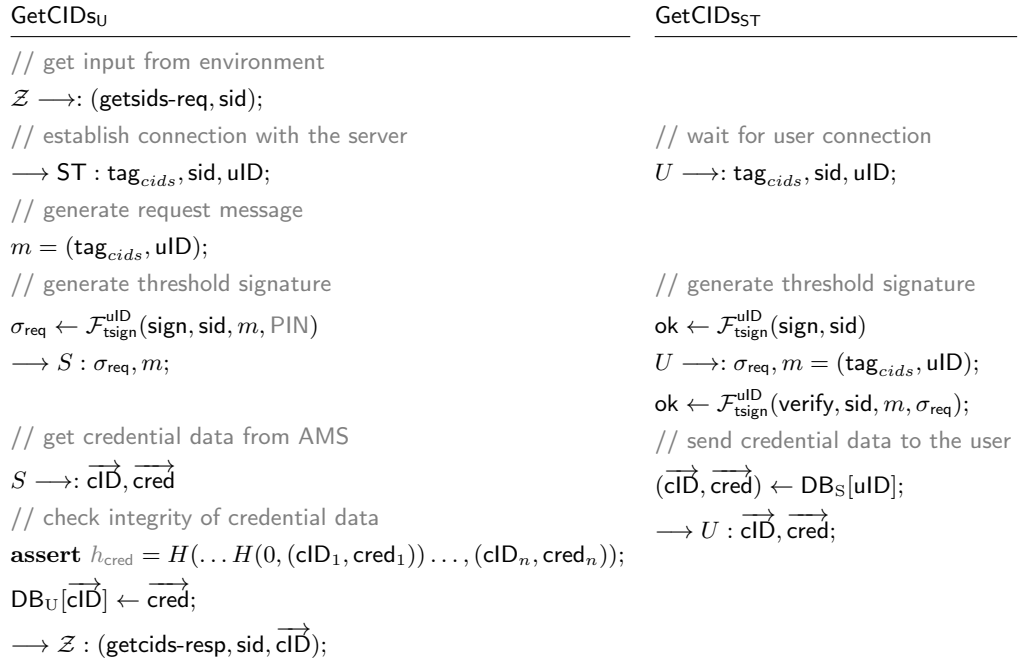

\centering
\begin{pchstack}
\begin{pcvstack}
\pseudocode[head={$\textsf{GetCIDs}_{\mathsf{U}}$}]{
\Comment{get input from environment}\\
\alisa{\env \send: (\textsf{getsids-req},\sid);}\\
\text{\Comment{establish connection with the server}} \\
\alisa{\send \store: \getcidsreq, \sid, \uid;}\\
\text{\Comment{generate request message}} \\
m = (\alisa{\getcidsreq},\uid);\\
\text{\Comment{generate threshold signature}} \\
\sigma_{\req} \gets \funsign^{\uid}(\textsf{sign},\sid,m,\hidden{\pin}) \\
\liza{\send S: \sigma_{\req},m;}\\
\\
\text{\Comment{get credential data from AMS}} \\
\alisa{S \send : \cids,\creds}\\
\text{\Comment{check integrity of credential data}} \\
\alisa{\pcassert \hidden{h_\cred} = H(\ldots H(0,(\cid_1,\cred_1))\ldots,(\cid_n,\cred_n));}\\
\alisa{\db{U}[\cids] \gets \creds;}\\
\alisa{\send \env: (\textsf{getcids-resp},\sid,\cids);}
}
\end{pcvstack}
\pchspace
\begin{pcvstack}
\pseudocode[head=$\textsf{GetCIDs}_{\store}$]{
\\ \\
\Comment{wait for user connection}\\
\alisa{U \send: \getcidsreq, \sid, \uid;}\\
\\ \\
\text{\Comment{generate threshold signature}} \\
\liza{\ok}
\gets \funsign^{\uid}(\textsf{sign},\sid) \\
\liza{U \send: \sigma_{\req},m =(\alisa{\getcidsreq},\uid);}\\
\liza{\ok \gets \funsign^{\uid}(\textsf{verify},\sid,m,\sigma_{\req});} \\
\text{\Comment{send credential data to the user}} \\
\alisa{(\cids,\creds) \gets \db{S}[\uid];}\\
\alisa{\send U: \cids,\creds;}\\
\\
\\
}    
\end{pcvstack}
\end{pchstack}
\caption{Protocol for downloading \alisa{credential backup from the server}}
\label{fig:prettiness-getcids}
\end{figure}

\subsubsection{Revocation}

Revocation can be initiated either by the user (Figure~\ref{fig:prettiness-revoke-u}) who wants to revoke their credential, or by the issuer (Figure~\ref{fig:prettiness-revoke-i}). Issuer revocation is initiated by the issuer's signed revocation request, which confirms that it is an authorised revocation. 
Both parties revoke $\cid$ by revoking the corresponding revocation token ($\rev_U$ or $\rev_I$) using $\funrevoke$.

Revocation procedure may result in a double revocation of the same credential, by the user and by the issuer. During verification, as RP does the check for $\rev_I$ and $\rev_U$ separately, this could leak to RP the fact that some credential has been rejected twice. \alisa{Moreover, if the token has been revoked once, RP could learn whether it has been revoked by the user or by the issuer.} As far as the AMS is honest, since the user will call revocation notification just before the final verification, this leakage is only possible if the server collaborates with a corrupted RP.


\begin{figure}[htbp]
\centering
\begin{pchstack}
\begin{pcvstack}
\pseudocode[head={$\textsf{Revoke}_{\mathsf{U}}$}]{
\Comment{get input from environment}\\
\alisa{\env \send: (\textsf{revoke-req},\sid,\cid);}\\
\text{\Comment{select credential to revoke}} \\
\alisa{(\cred, \sigma_{\cred}) \gets \db{U}[\cid];} \\
\text{\Comment{establish connection with the server}} \\
\alisa{\send \store: \revokereq, \sid, \uid;}\\
\text{\Comment{generate request message}} \\
\alisa{c_\cid \gets \fundec^{\uid}(\textsf{Encrypt},sid,\cid);}\\
m = (\alisa{\revokereq},\uid,\alisa{c_\cid});\\
\text{\Comment{generate threshold signature}} \\
\sigma_{\req} \gets \funsign^{\uid}(\textsf{sign},\sid,m,\hidden{\pin}) \\
\liza{\send S: \sigma_{\req},m;}\\
\\
\text{\Comment{decrypt revocation token data}}\\
\alisa{(\rev_I,\rev_U) \gets \fundec^{\uid}(\textsf{Decrypt},\sid,c_{\rev})}; \\
\text{\Comment{revoke the selected token}} \\
\ok \gets \funrev(\textsf{revoke},\sid,\alisa{\rev_U});\\
\alisa{\send \env: (\textsf{revoke-resp},\sid,\ok);}
}
\end{pcvstack}
\pchspace
\begin{pcvstack}
\pseudocode[head={$\textsf{Revoke}_{\store}$}]{
 \\ \\
 \\ \\
\Comment{wait for user connection}\\
\alisa{U \send: \revokereq, \sid, \uid;}\\
\\ \\ \\
\text{\Comment{generate and verify threshold signature}} \\
\liza{\ok} \gets \funsign^{\uid}(\textsf{\sign},\sid) \\
\liza{U \send: \sigma_{\req},m=(\revokereq,\uid,c_\cid);}\\
\liza{\ok \gets \funsign^{\uid}(\textsf{verify},\sid,m,\sigma_{\req});} \\
\text{\Comment{participate in threshold decryption}} \\
\ok \gets \fundec^{\uid}(\textsf{Decrypt},\sid); \\
\text{\Comment{confirm revocation}} \\
\ok \gets \funrev (\textsf{revoke},\sid);\\
}    
\end{pcvstack}
\end{pchstack}
\caption{Revocation protocol initiated by user}
\label{fig:prettiness-revoke-u}
\end{figure}

\begin{figure}[htbp]
\centering
\begin{pchstack}
\begin{pcvstack}
\pseudocode[head={$\textsf{Revoke}_{\mathsf{I}}$}]{
\Comment{get input from environment}\\
\alisa{\env \send: (\textsf{revoke-req},\sid,\cid);}\\
\text{\Comment{select credential to revoke}} \\
\alisa{\rev_I \gets \db{I}[\cid];}\\
\text{\Comment{establish connection with the server}} \\
\alisa{\send \store: \revokereq, \sid, \iid;}\\
\text{\Comment{extract $\uid$ from $\cid$}} \\
\alisa{\uid \gets \cid;}\\
\text{\Comment{generate request message}} \\
\alisa{c_\cid \gets \fundec^{\uid}(\textsf{Encrypt},sid,\cid);}\\
\alisa{m = (\revokereq,\iid,c_\cid)};\\
\sigma_{\req} \gets \funs^{\iid}(\textsf{sign},\sid,m); \\
\longrightarrow \store: (m,\sigma_{\req}); \\
\\
\text{\Comment{revoke the selected token}} \\
\ok \gets \funrev(\textsf{revoke},\sid,\alisa{\rev_I});\\
\alisa{\send \env: (\textsf{revoke-resp},\sid,\ok);}
}
\end{pcvstack}
\pchspace
\begin{pcvstack}
\pseudocode[head=$\textsf{Revoke}_{\store}$]{
 \\ \\
 \\ \\
\Comment{wait for issuer connection}\\
\alisa{I \send: \revokereq, \sid, \iid;}\\
\\ \\ \\ \\ \\ \\
\textsf{I} \longrightarrow: (m\alisa{=(\revokereq,\iid,c_\cid)},\sigma_{\req}) \\
\ok \gets \funs^{\iid}(\textsf{verify},\sid,m,\sigma_{\req}) \\
\text{\Comment{confirm revocation}} \\
\ok \gets \funrev(\textsf{revoke},\sid);\\
}
\end{pcvstack}
\end{pchstack}
\caption{Revocation protocol initiated by issuer}
\label{fig:prettiness-revoke-i}
\end{figure}

\subsubsection{Presentation}
The protocol of Figure~\ref{fig:prettiness-present} allows the user to generate presentation of credential with identifier $\cid$ and selectively disclosed attributes $\aidsd$. 
\alisa{First of all, the user looks for a} credential $\cred$ identified by $\cid$. User proceeds with authentication, followed by decrypting revocation token pair $\rev=(\rev_I,\rev_U)$ and the chosen attributes indexed by the identifiers $\aidsd$ using threshold decryption of $\fundec$. The user calls $\funrev$ to verify that revocation tokens $\rev_I$ and $\rev_U$ have not been revoked. The attributes are decrypted even if the credential is revoked, to prevent leaking the fact of revocation to the server. The obtained credential presentation $\pres$ contains \alisa{only the disclosed attributes}. All auxiliary data that will be needed to verify its validity will be provided in the verification phase. 


\begin{figure}[htbp]
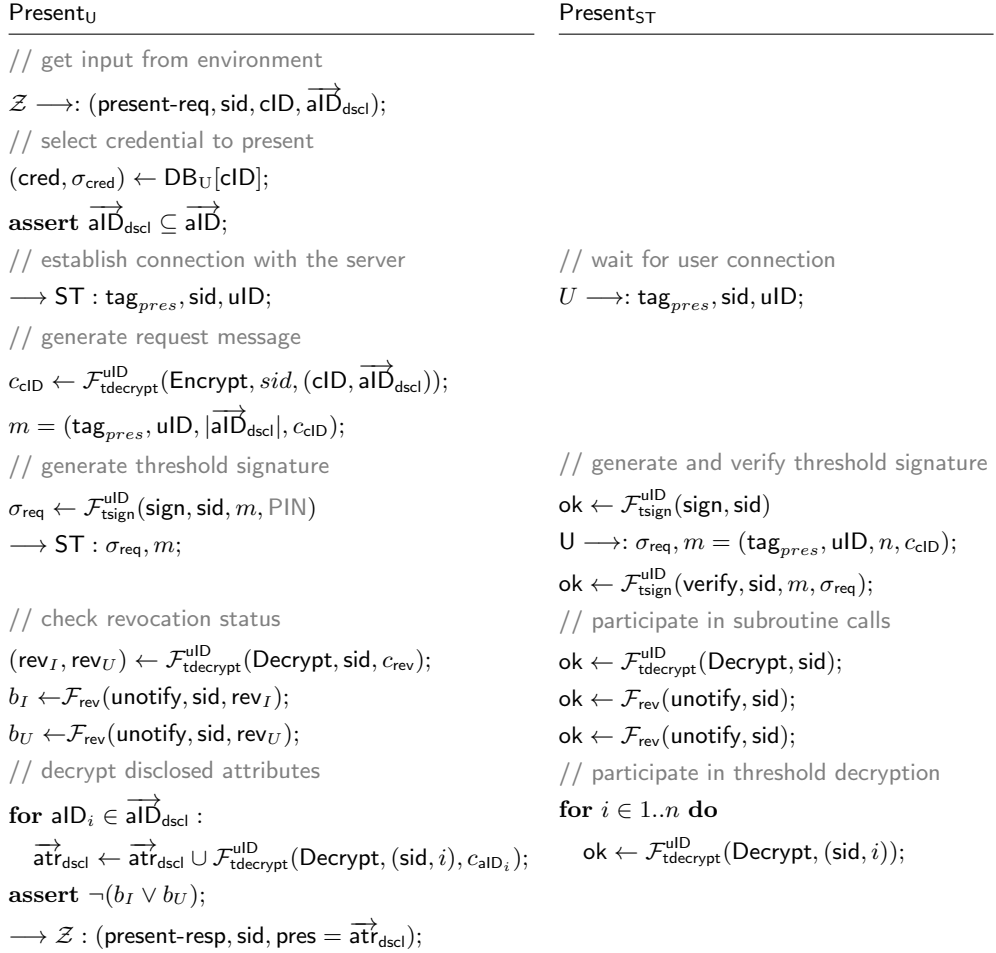

\centering
\begin{pchstack}
\begin{pcvstack}
\pseudocode[head={$\textsf{Present}_{\mathsf{U}}$}]{
\Comment{get input from environment}\\
\alisa{\env \send: (\textsf{present-req},\sid,\alisanotes{\cid},\aidsd);}\\
\text{\Comment{select credential to present}} \\
\alisa{(\cred,\sigma_\cred) \gets \db{U}[\cid];}\\
\alisa{\pcassert \aidsd \subseteq \aids;}\\
\text{\Comment{establish connection with the server}} \\
\alisa{\send \store: \presentreq, \sid, \uid;}\\
\text{\Comment{generate request message}} \\
\alisa{c_\cid \gets \fundec^{\uid}(\textsf{Encrypt},sid,(\cid,\aidsd));}\\
\alisa{m=(\presentreq,\uid,|\aidsd|,c_\cid)};\\
\text{\Comment{generate threshold signature}} \\
\sigma_{\req} \gets \funsign^{\uid}(\textsf{sign},\sid,m,\hidden{\pin}) \\
\liza{\send \mathsf{\store}: \sigma_{\req},m;} \\
\\
\Comment{check revocation status}\\
\alisa{(\rev_I,\rev_U)} \gets \fundec^{\uid}(\textsf{Decrypt},\sid,c_{\rev}); \\
\alisa{b_I \gets } \funrev(\textsf{unotify},\sid,\alisa{\rev_I});\\
\alisa{b_U \gets } \funrev(\textsf{unotify},\sid,\alisa{\rev_U});\\
\Comment{decrypt disclosed attributes}\\
\pcfor \aid_i \in \aidsd: \\
\quad \atrsd \gets \atrsd \cup \fundec^{\uid}(\textsf{Decrypt},(\sid,i),c_{\aid_i}); \\
\alisa{\pcassert \neg(b_I \vee b_U);}\\
\alisa{\send \env: (\textsf{present-resp},\sid,\alisa{\pres=\atrsd});}
}
\end{pcvstack}
\pchspace
\begin{pcvstack}
\pseudocode[head=$\textsf{Present}_{\store}$]{
 \\ \white{\aids}\\
 \\ \\ \white{\aids}\\
\Comment{wait for user connection}\\
\alisa{U \send: \presentreq, \sid, \uid;}\\
\\ \white{\aids}\\ \white{\aids}\\
\text{\Comment{generate and verify threshold signature}} \\
\liza{\ok} \gets \funsign^{\uid}(\textsf{sign},\sid) \\
\liza{\mathsf{U} \send: \sigma_{\req},m=(\presentreq,\uid,n,c_\cid);} \\
\liza{\ok \gets \funsign^{\uid}(\textsf{verify},\sid,m,\sigma_{\req});} \\
\text{\Comment{participate in subroutine calls}} \\
\ok \gets \fundec^{\uid}(\textsf{Decrypt},\sid); \\
\ok \gets \funrev(\textsf{unotify},\sid); \\
\ok \gets \funrev(\textsf{unotify},\sid); \\
\text{\Comment{participate in threshold decryption}} \\
\alisanotes{\pcfor i \in 1..n \pcdo} \\
\alisanotes{\quad \ok \gets \fundec^{\uid}(\textsf{Decrypt},(\sid,i))};\\
\\
}    
\end{pcvstack}
\end{pchstack}
\caption{Presentation with selective disclosure}
\label{fig:prettiness-present}
\end{figure}

\subsubsection{Verification by relying party} The protocol is presented on Figure~\ref{fig:prettiness-verify}. Before verification, RP receives updated state of the revocation functionality using notification interface of $\funrev$. If it is done immediately before the presentation, we cannot guarantee the unlinkability of presentation to the RP by AMS server, so it should be run independently. \liza{RP generates fresh challenge and sends it to the user for signature. This step is required to guarantee freshness of the presentation. User generates signature using $\funsign$, and blindness property ensures that the AMS server does not learn anything about credential presentation,} \alisa{except the fact that the user has presented some credential to some party.} \lizanew{We include $\cred$ and $\sigma_{\cred}$ to the message that gets signed, to achieve structure more similar to SD-JWT presentation protocol.} RP receives credential presentation from the user and verifies \liza{the user} and the issuer signatures. It uses verification interface of $\funrevoke$ to ensure that the credential is not revoked. Then it verifies whether the selectively disclosed attributes are indeed contained in the signed credential.


\alisa{Using threshold signing for the challenge means that the supporting server should be online upon verification. If we want to make verification possible offline, we will need to use ordinary signature like $\funs$, which means weaker guarantees for the RP if the user device gets stolen together with the secret key, or cloned. If server-supported signature is used, then the actual user may contact the server to stop generating signatures in case of stealth, and clone detection protects against silent cloning of the user device.}

Selective disclosure is achieved by decrypting only the values of selected $\aidsd$. For each such $\aid$, user calls verification interface of $\fundec$ to prove to RP that the presented plaintext value indeed corresponds to the ciphertext from the signed credential $\cred$. 
Figure~\ref{fig:comparison} presents comparison of the proposed solution with the selective disclosure for JWT specification with the trivial solution of protecting credentials in the storage. The main difference is that in the case of SD-JWT for the presentation, the user needs to decrypt all the attributes. Therefore, if the device memory leaks to the adversary at this point, all the attribute plaintext values will be disclosed. In our solution, only the attributes that are disclosed for the specific presentation are decrypted and can potentially leak. 

\begin{figure}
    \centering
    \includegraphics[width=\linewidth]{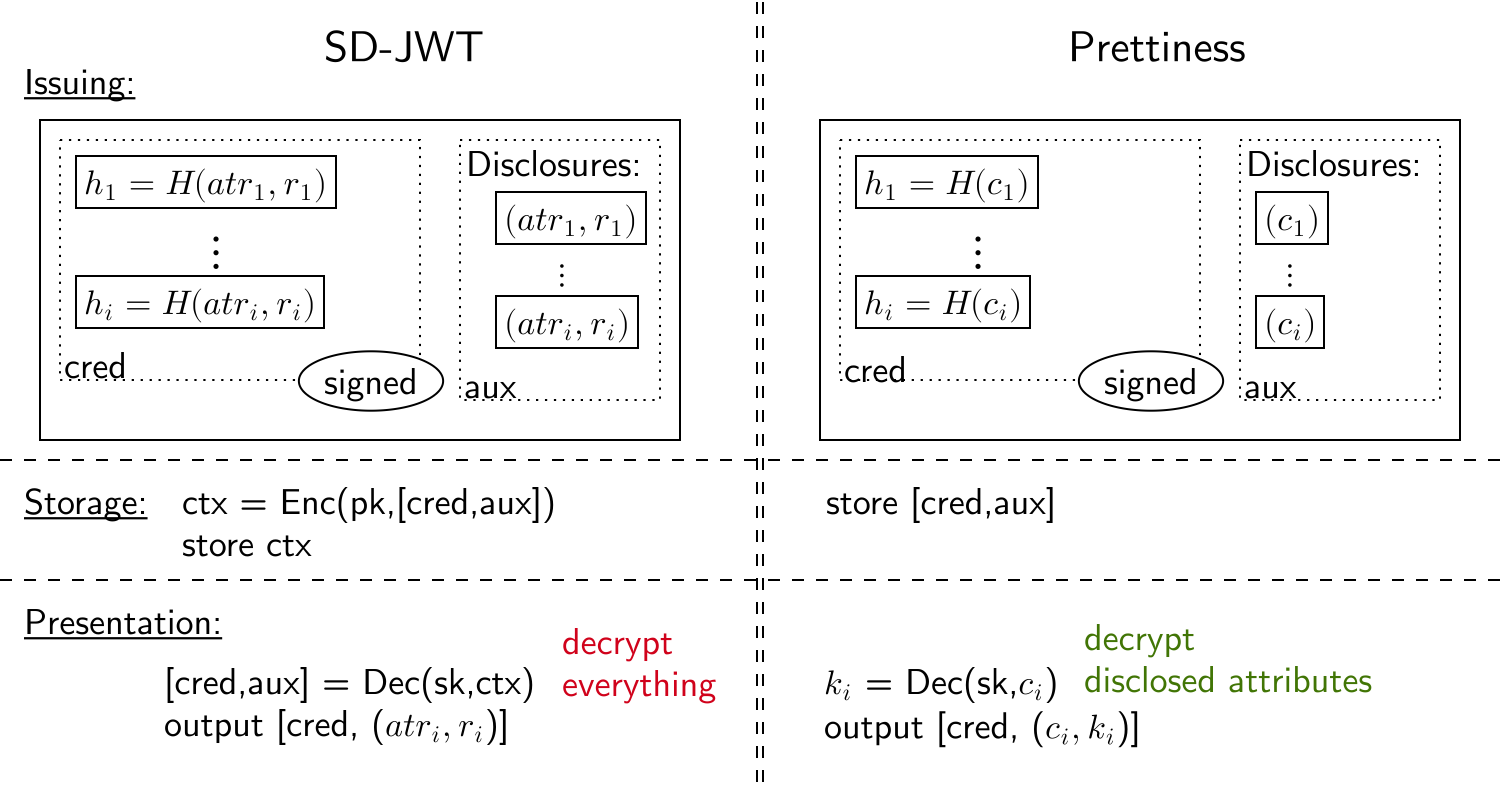}
    \caption{Comparison of Prettiness approach with Selective Disclosures for JWT}
    \label{fig:comparison}
\end{figure}

 
\begin{figure}[htbp]
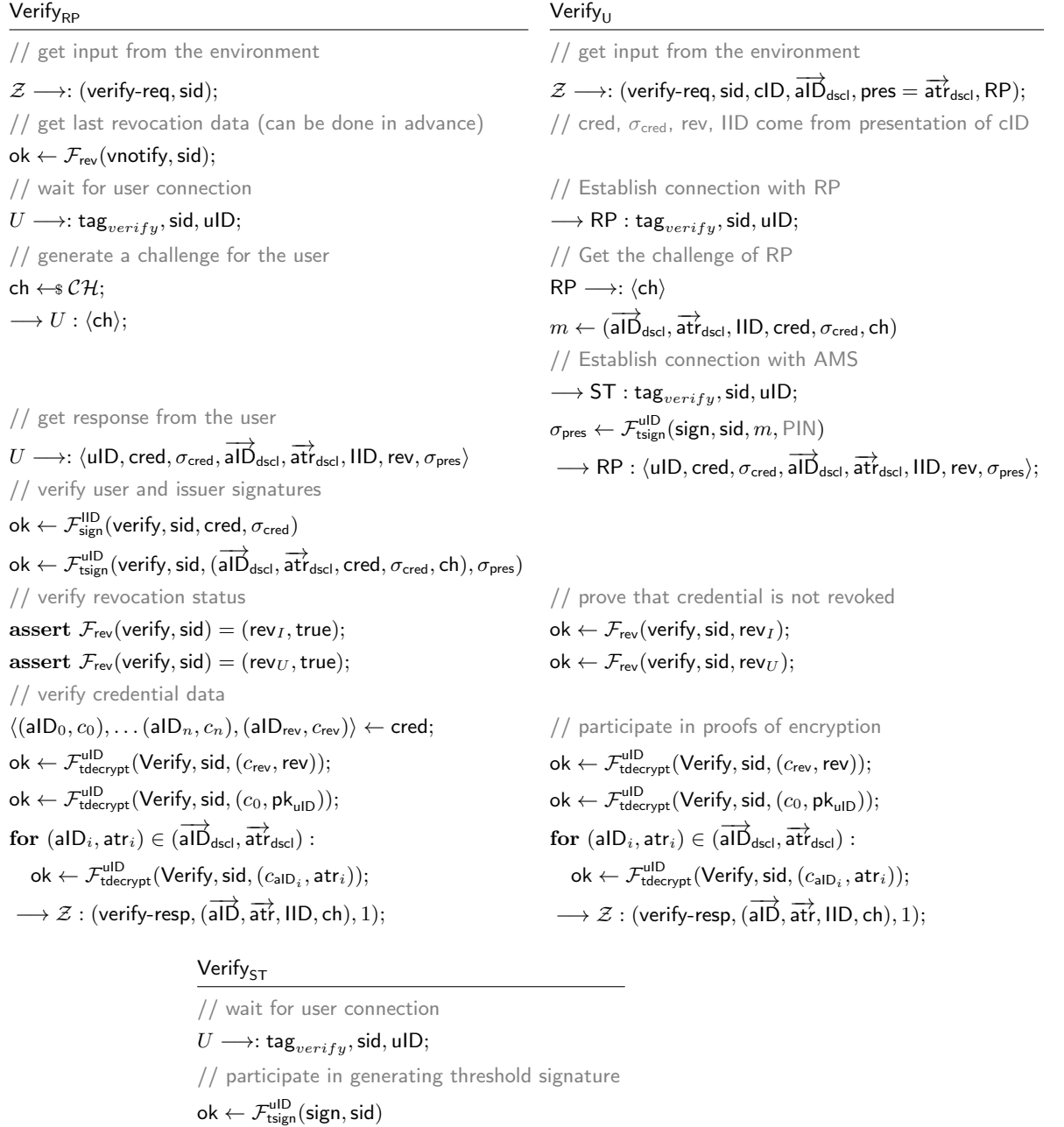

\centering
\begin{pchstack}
\begin{pcvstack}
\pseudocode[head={$\textsf{Verify}_{\mathsf{RP}}$}]{
\Comment{get input from the environment}\\
\alisa{\env \send: (\textsf{verify-req},\sid);}\white{\aidsd}\\
\Comment{get last revocation data (can be done in advance)}\\
\alisanotes{\ok \gets \funrev(\textsf{vnotify},\sid);}\\
\Comment{wait for user connection}\\
\alisa{U \send: \verifyreq, \sid, \uid;}\\
\Comment{generate a challenge for the user}\\
\liza{\ch \sample \mathcal{CH};} \\
\liza{\send U: \langle \ch \rangle;} \\
\\
\\
\Comment{get response from the user}\\
U \send : \langle \alisa{\uid}, \cred, \sigma_{\cred}, \alisa{\aidsd}, \alisa{\atrsd}, \iid,  \alisa{\rev}, \liza{\sigma_{\pres}} \rangle\\
\Comment{verify user and issuer signatures}\\
\ok \gets \funs^{\iid}(\textsf{verify},\sid,\cred,\sigma_{\cred}) \\
\liza{\ok \gets \funsign^{\uid}(\textsf{verify},\sid,\alisa{(\aidsd, \atrsd,\cred,\sigma_\cred,\ch)},\sigma_{\pres})} \\
\Comment{verify revocation status}\\
\pcassert \funrev(\textsf{verify},\sid) = (\alisa{\rev_I},\true); \\
\pcassert \funrev(\textsf{verify},\sid) = (\alisa{\rev_U},\true); \\
\Comment{verify credential data}\\
\langle (\aid_0, c_0), \dots (\aid_n, c_n), (\aid_{\rev}, c_{\rev}) \rangle \gets \cred; \\
\ok \gets \fundec^{\uid}(\textsf{Verify}, \sid, (c_{\rev},\rev));\\
\alisa{\ok \gets \fundec^{\uid}(\textsf{Verify}, \sid, (c_0,\pk_\uid));}\\
\pcfor (\aid_i,\atr_i) \in \alisa{(\aidsd,\atrsd)}: \\
\quad \ok \gets \fundec^{\uid}(\textsf{Verify}, \sid, (c_{\aid_i},\atr_i));\\
\send \env: (\textsf{verify-resp},(\aids, \atrs,\iid,\ch), 1);
}
\end{pcvstack}
\pchspace
\begin{pcvstack}
\pseudocode[head={$\textsf{Verify}_{\mathsf{U}}$}]{
\Comment{get input from the environment}\\
\alisa{\env \send: (\textsf{verify-req},\sid, \alisanotes{\cid}, \aidsd, \alisanotes{\pres=\atrsd}, \mathsf{RP});}\\
\Comment{$\cred$, $\sigma_{\cred}$, $\rev$, $\iid$ come from presentation of $\cid$}\\
\\
\Comment{Establish connection with RP}\\
\alisa{\send \mathsf{RP}: \verifyreq, \sid, \uid;}\\
\Comment{Get the challenge of RP} \\
\liza{\mathsf{RP} \send: \langle \ch \rangle} \\
\alisa{m \gets (\aidsd,\atrsd,\iid,\cred,\sigma_{\cred},\ch)} \\
\Comment{Establish connection with AMS}\\
\alisa{\send \store: \verifyreq, \sid, \uid;}\\
\liza{\sigma_{\pres} \gets \funsign^{\uid}(\textsf{sign},\sid,\alisa{m},\hidden{\pin})} \\
\send \mathsf{RP}: \langle \alisa{\uid}, \cred, \sigma_{\cred}, \alisa{\aidsd}, \alisa{\atrsd}, \iid, \alisa{\rev}, \liza{\sigma_{\pres}} \rangle;\\
\\ 
\\ \\
\Comment{prove that credential is not revoked}\\
\ok \gets \funrev(\textsf{verify},\sid,\alisa{\rev_I});\\
\ok \gets \funrev(\textsf{verify},\sid,\alisa{\rev_U});\\
\\
\Comment{participate in proofs of encryption}\\
\ok \gets \fundec^{\uid}(\textsf{Verify},\sid,(c_{\rev},\rev));\\
\alisa{\ok \gets \fundec^{\uid}(\textsf{Verify},\sid,(c_0,\pk_\uid))};\\
\pcfor (\aid_i,\atr_i) \in \alisa{(\aidsd,\atrsd)}: \\
\quad \ok \gets \fundec^{\uid}(\textsf{Verify},\sid,(c_{\aid_i},\atr_i));\\
\send \env: (\textsf{verify-resp},(\aids, \atrs,\iid,\ch), 1);
}
\end{pcvstack}
\end{pchstack}
\begin{pchstack}
\pseudocode[head={$\textsf{Verify}_{\mathsf{\store}}$}]{
\Comment{wait for user connection}\\
\alisa{U \send: \verifyreq, \sid, \uid;}\\
\Comment{participate in generating threshold signature}\\
\lizanew{\ok \gets \funsign^{\uid}(\mathsf{sign},\sid)} \\
}
\end{pchstack}
\caption{Verification of credential presentation. Here $\rev = (\rev_I,\rev_U)$}
\label{fig:prettiness-verify}
\end{figure}



\subsubsection{Logging.}
We assume that the server logs all request messages along with the time at which these requests have been made. The user can then get an overview of transactions performed to their data.  The list of logs and information that the user can extract from these is given in Table~\ref{tbl:log}.

In case of corrupted AMS, even if the user compares the log against some digest for integrity, the AMS may simply refuse to return the log to the user, so logging is useful assuming an \emph{honest-but-curious} server. Moreover, if the issuer is corrupted, then there are no guarantees that valid revoked $\cid$ has been recorded. \alisa{While accountability of AMS could potentially be enforced by requiring its signature on the digest, and the ciphertext $c_\cid$ of the issuer could contain information about the actual performed revocation (accompanied with a proof-of-encryption), these enhancements would require more details, and are out of scope of this paper.}

\begin{table}[ht!]
\begin{tabular}{l|l|l}
\textbf{Protocol} & \textbf{Logged record} & \textbf{What can be read by user} \\ \hline
Issue & $(\issuereq,id_{\cid},\uid,\aids,\iid)$ & issued $\cid$\\
Get $\cid$s & $(\getcidsreq,\uid)$ & fact of credential backup request \\
Revoke(user) & $(\revokereq,\uid,c_{\cid})$ & revoked $\cid \gets \fundec^{\uid}(\textsf{Decrypt},\sid,c_{\cid})$\\
Revoke(issuer) & $(\revokereq,\iid,c_{\cid})$ & revoker $\iid$ and $\cid \gets \fundec^{\uid}(\textsf{Decrypt},\sid,c_{\cid})$\\
Present & $(\presentreq,\uid,n,c_{\cid})$ & presented $\cid \gets \fundec^{\uid}(\textsf{Decrypt},\sid,c_{\cid})$\\
\end{tabular}
\caption{Logged records and their interpretation.}\label{tbl:log}
\end{table}

\subsubsection{\alisa{Cheating detection.}}

Revocation functionality $\funrevoke$ allows AMS to conceal the fact of revocation from a party $P$ (a user or a verifier). In this case, $P$ and $R$ together can detect cheating and prove it to the third party. However, some event needs to trigger such a cheating detection in practice. This depends on the impact that cheating may impose.

\paragraph{Cheating to a user.} If AMS silences some revocation from the user, it may result in a rejected legitimate verification. The user may consult the issuer to ask whether the issuer has revoked the credential. If it is indeed so, the user and the issuer together may prove the fact of cheating to a third party. If the credential has not been revoked, then the verifier has chosen not to accept the credential, but it is a legitimate decision of the verifier (in the scope of $\ifun$).

\paragraph{Cheating to a verifier.} If AMS silences some revocation from the verifier, it may result in an accepted illegitimate verification. The verifier and the issuer may resolve this problem, but the problem must be detected in first place. In practice, the verifier and the issuer could periodically run verification (with issuer acting as a presenter) on randomly selected tokens held by the issuers whom the verifier accepts, to detect cheating with some non-negligible probability.

\subsubsection{Expiration.}

Expired records need to be periodically removed (or archived). Let $\Delta$ be an overestimated time after which any attribute will "definitely expire". For both the credential and the revocation tables, the server removes an entry as soon as the time $\Delta$ passes after it has been added to the corresponding table. 

\subsection{\alisa{Security of the proposed protocol}}

\begin{theorem}\label{thm:system} Assuming a \emph{covertly} corrupted AMS (Definition~\ref{def:covert:formal}), the protocol set $\system$ of Sec.~\ref{sec:system} UC-realizes $\ifun$ in $\funs^{\iid}, \funsign^{\uid}, \fundec^{\uid}, \funrevoke$- hybrid model.
\end{theorem}

The full security proof is presented in \liza{supplementary material}. 
In this section, we present the intuition behind the simulation, considering the most interesting corruptions. 

\alisa{For the proof, we define a simulator $\simm$ mediating communication between the real adversary $\adv$ and the ideal functionality $\ifun$, trying to generate messages for $\adv$ and $\ifun$ in such a way that $\env$ gets convinced that it is communicating with the real protocol $\system$. The simulator plays roles of honest parties in $\system$, generating messages for corrupted parties, which are controlled by $\adv$. In addition, $\simm$ interacts with $\ifun$ to synchronise with its inputs and outputs. This shows that any attack of $\adv$ on $\system$ can be converted (by the simulator) to an analogous attack on $\ifun$. In particular, all leakages of $\system$ to the adversary are computable from the "intensional" leakage of $\ifun$ and the data that $\adv$ has provided for corrupted parties. We describe how $\simm$ behaves in each routine of $\ifun$.}

\underline{Initialisation:} $\simm$ defines the function $\present$ for $\ifun$, used to generate credential presentation. This function takes the attribute values $\atrs$ that correspond to $\cid$ and $\aids$ from the table $\mathsf{T}_a$ of issued attributes, as is done in the real protocol. $\simm$ defines $\ver$ that is used to verify credential presentation by checking that $\atrs$ have indeed been stored in $\mathsf{T}_a$ for given $\cid$ and $\aids$. $\simm$ also initializes all functionalities $\funsign^{\uid}$, $\fundec^{\uid}$ (a separate instance for each user of the system), $\funrevoke$ (one instance for the entire system) \liza{and $\funs^{\iid}$} (a separate instance for each issuer in the system).

\underline{Issuing:} $\simm$ ensures that the credential contents are only added to $\mathsf{T}_a$ if it has indeed been requested by the (honest) user and issued by the (honest) issuer. \alisa{On behalf of an honest issuer, the simulator treats credential request from a malicious user as successful only if it could indeed come up with some valid attribute values $\atrs$ for that user, and if generation of revocation token for that user has succeeded, so a credential will only be sent to the user if the issuer has possibility to revoke it later.} On behalf of an honest user, before accepting credential from a malicious issuer, simulator verifies that (1) it contains attribute from the valid set of the attributes, (2) that issued credential contains revocation token $\rev$ generated by the $\simm$ and (3) that credential contains valid signature. This ensures that successful issuing will be followed by successful presentation and revocation. 
\alisa{Finally, all credential data except the party identities and the attribute identifiers $\aids$ (which are explicitly leaked by $\ifun$) is encrypted, and can be simulated to a corrupted AMS.} 

\alisa{Whenever $\ifun$ adds a credential record $\cid$ to $\mathsf{T}_a$, the simulator adds a corresponding entry (with additional protocol-specific data, like signatures and revocation tokens) to its local table $\tiss$.}


\underline{Get all CID-s:} On behalf of honest AMS, $\simm$ simulates to a corrupted user all the $\cid$-s and credentials that have been issued to the corrupted user so far. \alisa{All such credentials have already been simulated to $\adv$ in the issuing phase, and recorded to $\tiss$.} On behalf of honest user, $\simm$ verifies the data provided by corrupted AMS against the digest representing the data that the AMS should have shown, \alisa{which ensures consistency with the output of $\ifun$.} 

\underline{Revocation:} On behalf of honest AMS, $\simm$ collaborates with the corrupted revoker, letting $\adv$ choose the token to revoke. It is not possible to illegitimately revoke a credential $\cid$ of an honest user without having rights on its revocation, which have previously been given by $\funrevoke$ only to the user and the issuer of $\cid$. On behalf of an honest revoker, unless the user of the revoked credential or the RP is corrupted (in these cases, $\simm$ gets $\cid$ from $\ifun$ and can use the same token $\rev$ that was issued for $\cid$), $\simm$ takes a random revocation token $\rev \gets \tokens$ (which will not coincide with any other token with probability $1/|\tokens|$), and simulates calls of $\funrevoke$ with that token. Adversary will not notice the difference \alisa{until the details of revocation gets leaked, which happens either upon reaching user corruption level $c_{\mathsf{U}} = 2$, or during presentation to a corrupted RP. In both cases, $\ifun$ will provide all relevant revocation data to $\simm$, so that it can add all relevant tokens $\rev$ to $\revoked$ of $\funrevoke$.} 

\alisa{Whenever $\ifun$ adds a credential record $\cid$ to $\mathsf{T}_r$, the simulator has added a corresponding revocation token $\rev$ to its local set $\revoked$ of revoked tokens, unless $\cid$ has not been leaked to $\simm$ yet. In the latter case, the user of $\cid$ is an honest party, and it will anyway not call verification of $\funrevoke$ with a revoked token.}

\underline{Presentation:} On behalf of honest AMS, $\simm$ collaborates with the corrupted user, letting $\adv$ choose the ciphertexts of their choice for decryption. 
On behalf of honest user, $\simm$ simulates presentation for a randomly chosen credential $\cid'$ from $\tiss$, which is consistent with $\mathsf{T}_a$ of $\ifun$. \alisa{If there is no such $\cid'$, presentation trivially fails, as there have been no credentials issued to that user.} At this point, 
$\adv$ only learns $|\aidsd|$, i.e. the number of threshold decryptions. $\fundec^{\uid}$ does not leak which ciphertexts exactly have been decrypted. 
\alisa{The simulator does not learn whether presentation of an honest user has eventually succeeded, but this fact does not need to be simulated to $\adv$.}

\underline{Verification:} If the user is honest and RP is corrupted, $\simm$ gets disclosed presentation data from $\ifun$, which is enough to construct all messages for RP, using additional data from $\tiss$. Obtained $\cid$ may be different from $\cid'$ previously chosen for the presentation phase, but the simulation has not depended on $\cid'$ anyway.

If the user is corrupted with level $c_{\mathsf{U}} \geq 2$ (with $c_{\mathsf{U}} = 1$, corruptions are passive), it may show a revoked or invalid credential to an honest RP. $\simm$ detects it using its internal tables for revocation, signing and decryption for that user. \alisa{If $c_{\mathsf{U}} = 3$, all these tables are valid for that user from the beginning.} Upon $c_{\mathsf{U}}=2$ corruption, $\ifun$ leaks internal records of that user to $\simm$, allowing it to update its tables accordingly. Hence, these tables are filled with data consistent with the internal tables of $\ifun$.

If both RP and $\store$ are corrupted, then the adversary expects to learn the sessions on which the presented credential $\cid$ has been revoked (if such a session exists). In this case, $\ifun$ sends a corresponding index of the revocation table to $\simm$, from which it can refer to the session id of that revocation.


\section{Efficiency}\label{sec:efficiency}

We have estimated how well Prettiness would perform in practice. First of all, we collected benchmarks for implementations of certain protocols UC-realising $\funrevoke$, $\fundec$, $\funsign$, and $\funs$. We run the benchmarks on Intel(R) Core(TM) i5-10210U CPU @ 1.60GHz, 4 cores, 2 threads per core, and 15GB RAM.

\subsection{Benchmarks for digest function.}

\alisa{We assume that the local database digest is computed using SHA256 hash function. We measured its computation times using \texttt{openssl speed sha256}, getting values below $7 \cdot 10^{-9}$ seconds per byte. While this small number is very small, the total time to compute the digest of the entire database grows with the number of user credenitals.}

\subsection{Benchmarks for $\funs$ and $\funsign$.}

 \alisa{We assume that $\funs$ is realised with an ordinary 2048-bit RSA signature. We measured the signature generation and verification times using \texttt{openssl speed rsa}, getting values below $0.001$ seconds. This small number has very little impact on the overall efficiency.}

The efficiency of a protocol UC-realizing $\funsign$ has been taken from~\cite{splitkey}, where it is based on RSA signature. \lizanew{Since protocol presented in~\cite{blind-splitkey} just adds blinding to the protocol presented in~\cite{splitkey},  the benchmarks will be similar.} Signing operation itself takes tens of milliseconds, which we for safety round up to $0.1$ seconds. \alisa{According to~\cite{splitkey}}, the most time overhead comes from the network delay to transmit the signature share from the client to the server, \alisa{which we do not include into the timing benchmark, instead giving a separate benchmark for communication.} 
The signature itself is an ordinary RSA signature, verified in a usual manner.

The communication has been estimated instantiating~\cite{splitkey} with 2048-bit RSA and 128-bit randomness. The communication consists of a message $\langle y,m,r,r' \rangle$ that the client sends to the server, and $\langle s,m \rangle$ that the server sends back. Here $y$ and $s$ are 2048-bit RSA signatures, and $r,r' \in \set{0,1}^\eta$. This gives $2 \cdot (2048 + 128 + |m|_{b}) = 4352 + 2 \cdot |m|_{b}$ bits, or $544 + 2\cdot|m|_B$ bytes, where $|m|_B$ is the message length in bytes.

\subsection{Benchmarks for $\fundec$.}

Protocols for $\fundec$ have been re-implemented in Go, \alisa{which gives better results compared to the Python implementation of~\cite{threshold-decryption}}, and provides a more precise estimation of sizes of different components. The underlying hash function is SHA-256. We got ca 282ms for encryption, 362ms for decryption, and 106B of total decryption communication.

\subsection{Benchmarks for $\funrevoke$.}

Efficiency of a protocol for $\funrevoke$ has been taken from~\cite{funrevoke}. The issuing and verification routines of this protocol do not depend on the total number of revoked tokens. However, the efficiency of notifications (and slightly of the revocations) also depends on the total number of revoked elements. We reproduce some benchmarks In Table~\ref{tbl:efficiency:funrevoke}. We will assume $m=1,000,000$ revoked tokens in this work, but we also reproduce the results for some smaller values of $m$ to demonstrate the efficiency in the beginning, when not so many credentials have been revoked yet.

\begin{table}[t!]
\begin{tabular}{|l|l|l| l|}
\hline
  & \textbf{Local comp. time (s)}	&	\textbf{Comm. (bytes)} & \textbf{Output size (bytes)} \\
\hline
SHA256(per byte) & 0.00000000682 & -- & 32 \\
\hline
$\funs(\mathsf{sign})$ & 0.001 & -- & 256 \\
$\funs(\mathsf{verify})$ & 0.001 & -- & -- \\
\hline
$\funsign(\mathsf{sign})$ & 0.1 & 544 + |\textsf{msg}| & 256 \\
$\funsign(\mathsf{verify})$ & 0.001& -- & -- \\
\hline
$\fundec(\mathsf{encrypt})$ & 0.282 & -- & 2618 + |\textsf{msg}| \\
$\fundec(\mathsf{decrypt})$ & 0.362 & 106 & -- \\
$\fundec(\mathsf{verify})$ & 0.00000131 & 32 & -- \\
\hline
\end{tabular}
\caption{Benchmarks for signing and threshold decryption protocols. Here the \emph{output size} is the signature size for signing routines, and ciphertext size for encryption routine.}\label{tbl:efficiency:bblocks}
\end{table}

\begin{table}[t!]
\begin{minipage}{0.62\textwidth}
\textbf{Local computation times (ms)}\\
\begin{tabular}{|l|l|l|l| l | l | }
\hline
m &		0	&	100	& 10,000 & 1,000,000 \\
\hline
$\funrevoke(\mathsf{issue})$ & 0.00447 & -- & -- & -- \\
$\funrevoke(\mathsf{revoke})$ & 0.108 & 16.5 & 23.8 & 743 \\
$\funrevoke(\mathsf{unotify})$ & 0.081 & 0.137  & 5.57 & 673 \\
$\funrevoke(\mathsf{vnotify})$ & 0.075 & 0.137 & 6.52 & 716 \\
$\funrevoke(\mathsf{verify})$ & $\leq $2300 & -- & -- & -- \\
\hline
\end{tabular}
\end{minipage}
\begin{minipage}{0.35\textwidth}
\textbf{Communication (bytes)}\\
\begin{tabular}{|l|l| l | l| }
\hline
0	&	100	&	10,000		 & 1,000,000 \\
\hline
48 & -- & -- & --  \\
160 & -- & -- &  --  \\
114 & 3,314 &  320,114 & 32,000,114 \\
114 & 3,314 &  320,114 & 32,000,114  \\
162 & -- & -- & --  \\
\hline
\end{tabular}
\end{minipage}
\caption{Benchmarks for different routines of revocation protocol. Notation "--" represent numbers that are the same as for $m=0$.}\label{tbl:efficiency:funrevoke}
\end{table}

\subsection{Estimating efficiency of the full Prettiness protocol}

\begin{table}[ht!]
\begin{tabular}{l|l|l|l|l|l|l|l}
 & \textbf{Issue} & \textbf{Get Creds} & \textbf{Revoke(I)} & \textbf{Revoke(U)} & \textbf{Present} & \textbf{DB update} & \textbf{Verify}\\ \hline
comm (B) & 90,716 & 3,906,474 & 3,355 & 6,688 & 7,590 & 32,000,114 & 122,170\\
time (s) & 8.114 & 0.127 & 1.027 & 1.487 & 5.71 & 0.716 & 4.702\\ \hline
\end{tabular}
\caption{Summarised efficiency of the full system for $N=100$ credentials per user, $n = 10$ attributes per credential, $m = 10^6$ revoked tokens.}\label{tbl:fullbench}
\end{table}

Benchmarks of the building block protocols can be found in Table~\ref{tbl:efficiency:bblocks} and Table~\ref{tbl:efficiency:funrevoke}. The times are measured for local computations only, without the network overhead, which is measured separately. Based on benchmarks of the building blocks, we estimated theoretical efficiency of the full system. In this estimate, we do \emph{not} include the overhead of establishing communication channel (TLS) between parties, as well as the actual sizes of network packets moving inside TLS connection.

Since efficiency depends a lot on the actual number of credentials and their sizes, we considered a particular use case with $N=100$ credentials per user, and $n = 10$ attributes per credential, and $m = 1,000,000$ total revocations. In practice, such an $m$ can be an overestimate if there would be not too many urgent revocations of active credentials (for which expiration time has not passed).

The summary of results is given in Table~\ref{tbl:fullbench}. \alisa{As the user and verifier notification requires a large download of the revocation database, which can be done in advance, we treat is as a separate routine \textbf{DB update} in our estimate, which has more or less the same efficiency for the user and the verifier.} More detailed benchmarks, explaining how exactly the times have been estimated based on the protocol descriptions, can be found in the supplementary material.

Let us discuss the bottlenecks of the full protocol efficiency.

\paragraph{Issuing.} The main overhead of issuing comes from letting the user decrypt the attributes to check their contents. Formally, we need this step to avoid selective failure attacks, where a further presentation of a malformed credential to RP surprisingly fails. 

\paragraph{Get all Credentials.} The communication overhead is caused by downloading the full backup from the server. Although the user could download only \emph{some} of the credentials, this would violate user privacy, giving to the server a hint of \emph{which} credential is going to be presented. We find this communication feasible if the number of credentials that have not been expired yet is not too large, as in the benchmarked $N=100$ case. 

\paragraph{DB Update.} In this case, communication overhead comes from the necessity to download the data for $m$ revoked tokens. \alisa{Assuming that RP updates the status of the revocation database periodically (e.g. once per hour), and it is not urgent, then a download of 32MB of data per update would not be an issue. It is more complicated for the user who may be located in a place with bad network connection. We propose that, as our protocol anyway does not hide from AMS the identity of the user of a revoked credential (to make logging easier), the user only needs to download $m_U$ of the $m$ entries of the revocation database, where $m_U$ is the number of revocations that concern the user $U$, which makes download size much smaller.}

Authenticating the revoker party allows to protect against denial-of-service attacks, where a revoker maliciously attempts to add to the revocation database too many (potentially non-existent) credentials. AMS could limit the number of allowed revocations for same user/issuer. E.g. a user is not allowed to revoke more credentials that have been issued \emph{to} it, and an issuer is not allowed to revoke more credentials that have been issued \emph{by} it. 


\section{Conclusion}\label{sec:conclusion}

This work proposes attribute management system, where the user is in control over their credentials. Credentials are stored 
on the supporting server side. To ensure secure storage of cryptographic material on the user's device, system relies on threshold cryptography (used for user's signatures and attribute decryption functionalities). Credentials can be revoked by both user and the issuer. 
We construct an ideal functionality that captures security guarantees provided by the system and prove that proposed protocol set securely realises this functionality. Finally, we provide an efficiency estimate of our system. Future work includes post-quantum version of the proposed system as well as possible optimisations.

%
%
%
\bibliographystyle{plain}
\bibliography{ref2}

\appendix

\section{Supporting ideal functionalities}\label{app:functionalities}

Figures~\ref{fig:funrevoke:1} and ~\ref{fig:funrevoke:2} define revocation functionality $\funrevoke$. The leakages of $\funrevoke$ for a particular revocation protocol of~\cite{funrevoke} are listed in Table~\ref{tbl:leaks:funrevoke}.

Figures~\ref{fig:fideal-sign-1},\ref{fig:fideal-sign-2},\ref{fig:fideal-sign-3} define blind server-supported signing functionality $\funsign$. Figures~\ref{fig:fideal-dec-1},\ref{fig:fideal-dec-2} define blind server-supported decryption $\fundec$. 


\begin{table}[t!]
\begin{tabular}{l|l|l|l}
\textbf{Routine} & \textbf{Leakage} & \textbf{Value} & \textbf{Comments} \\ \hline  \hline
\textbf{Issue} & $\leak{issue}{\emptyset}(U,R,x)$     & $U,R$ & Do not hide the party identities.  \\  \cline{2-4}
& $\leak{issue}{\set{R}}(U,R,x)$ & $x$ & The revoker learns $x$  \\  \cline{2-4}
& $\leak{issue}{\set{U}}(U,R,x)$ & $x$ & The user learns $x$  \\  \cline{2-4}
\hline
\textbf{Revoke} & $\leak{revoke}{\emptyset}(R,x)$ & $R$ & Do not hide the identity of the revoker $R$.  \\  \cline{3-4}
& & $(R \in \mathcal{R}_x)$ & Leak whether the request is legitimate  \\
& & $ \wedge (x \notin \revoked)$ & and $x$ has not been revoked yet. \\  \cline{2-4}
& \alisanotes{$\leak{revoke}{\set{S,V}}(R,x)$} & $\set{x} \cap \presentedv$ & The server learns $x$ shown to corrupted $V$. \\  \cline{2-4}
& $\leak{revoke}{\user{x}}(R,x)$ & $x$ & The user of $x$ learns $x$  \\  \hline
\hline
\textbf{Unotify} & $\leak{unotify}{\emptyset}(U,x)$ & $U$ & Do not hide the identity of the user $U$.  \\
\cline{3-4}
  & & $U \in \mathcal{U}_x$ & Leak the fact of illegitimate request. \\  \cline{2-4}
  \hline
\hline
\textbf{Vnotify} & $\leak{vnotify}{\emptyset}(V)$ & $V$ & Do not hide the identity of the verifier $V$. \\  \cline{2-4}
\hline
\hline
\textbf{Verify}& $\leak{verify}{\emptyset}(U,V,x)$ 
& $x \notin \revokedu$ & Leak whether the user has previously called \\
& & & a sufficiently fresh notification for $x$.  \\   \cline{2-4}
& $\leak{verify}{\set{V}}(U,V, x)$ & $x$ & The verifier learns $x$.\\ \cline{2-4}
& \alisanotes{$\leak{verify}{\set{S,V}}(U,V,x)$} & $\revsids{x}$ & Server and verifier learn when $x$ was revoked. \\
\hline
\end{tabular}
\caption{Leakages of $\funrevoke$. The set $\presentedv$ is the set of tokens $x$ such that $(\textsf{verify},\sid,x,V)$ has been input to $\funrevoke$.
The set $\revokedu$ is the set of tokens $x$ such that $\funrevoke$ has returned $(\textsf{unotify-resp-u},\sid,\true)$ on the \emph{latest} input $(\textsf{unotify},\sid,x)$ from the user $U$, which happened \emph{after} $\funrevoke$ sent its last $(\textsf{vnotify-resp-v},\sid,\ok)$ to $V$.
The set $\revsids{x}$ is the set of session identifiers $\sid$ such that $(\textsf{revoke},\sid,x)$ has been input to $\funrevoke$ before.
}\label{tbl:leaks:funrevoke}
\end{table}

\begin{figure}
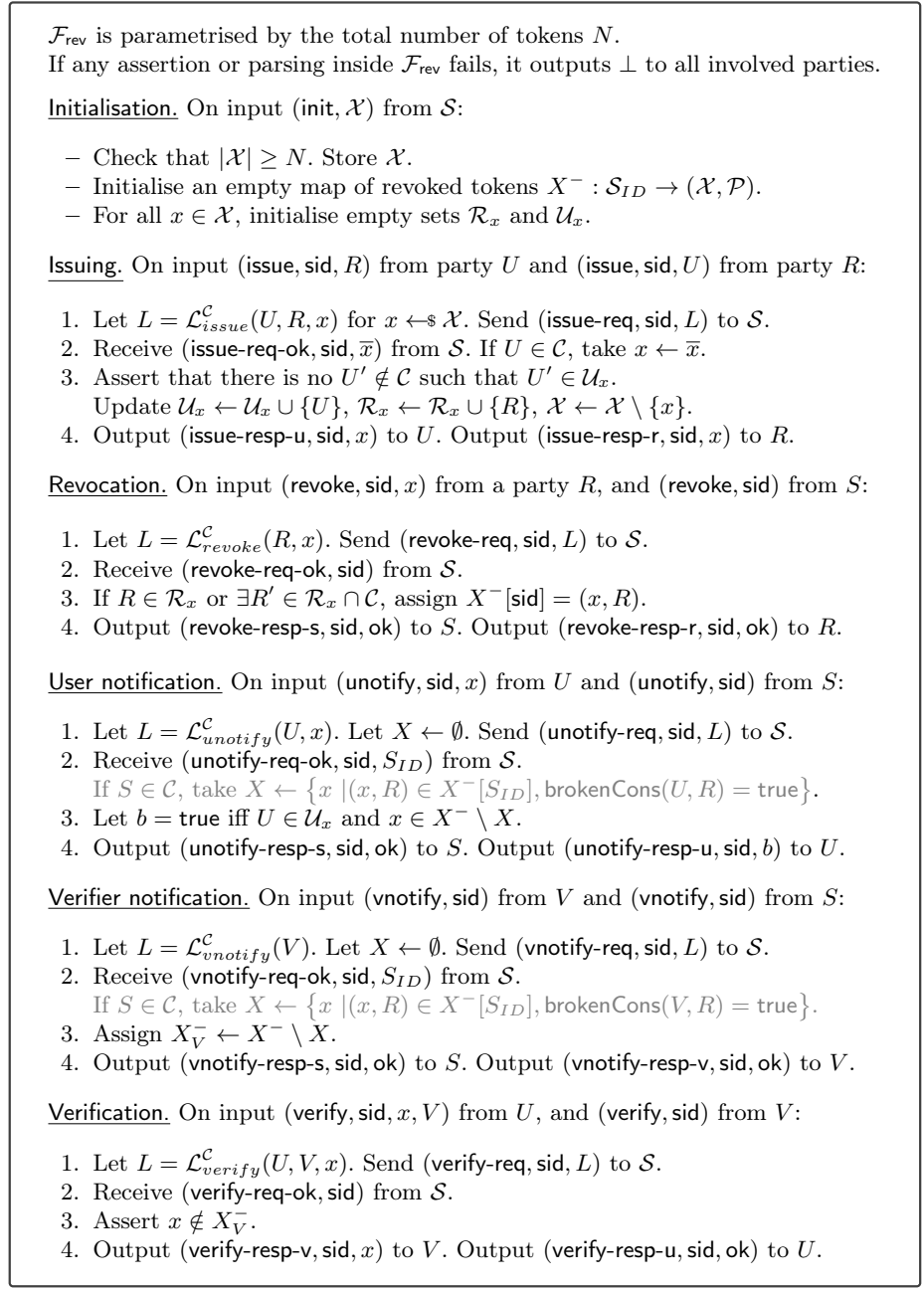

\begin{tcolorbox}[colback=white,arc=0.3mm, boxrule=0.3mm]
$\funrevoke$ is parametrised by the total number of tokens $N$.

If any assertion or parsing inside $\funrevoke$ fails, it outputs $\bot$ to all involved parties.

\vspace{0.2cm}

\underline{\textsf{Initialisation}.} On input $(\textsf{init}, \tokens)$ from $\sdv$:
\begin{itemize}
    \item Check that $|\tokens| \geq N$. Store $\tokens$.
    
    \item Initialise \alisanotes{an empty map of revoked tokens $\revoked : \sidset \to (\tokens,\parties)$}.
        
    \item For all $x \in \tokens$, initialise empty sets $\revoker{x}$ and $\user{x}$.
    
\end{itemize}

\underline{\textsf{Issuing}.} On input $(\textsf{issue}, \sid, \alisanotes{R})$ from party $U$ and $(\textsf{issue}, \sid, \alisanotes{U})$ from party $R$:
\begin{enumerate}

    \item Let $L = \leak{issue}{\corrupted}(U,R,x)$ for $x \sample \tokens$.
    Send $(\textsf{issue-req}, \sid, L)$ to $\sdv$.
    
    \item Receive $(\textsf{issue-req-ok}, \sid, \mal{x})$ from $\sdv$. If $U \in \corrupted$, take $x \gets \mal{x}$. 

    \item Assert that there is no $U' \notin \corrupted$ such that $U' \in \user{x}$.
    
    Update $\user{x} \gets \user{x} \cup \set{U}$, $\revoker{x} \gets \revoker{x} \cup \set{R}$, $\tokens \gets \tokens \setminus \set{x}$. 
    
    \item  Output $(\textsf{issue-resp-u}, \sid, x)$ to $U$. Output $(\textsf{issue-resp-r}, \sid, x)$ to $R$.
\end{enumerate}

\underline{\textsf{Revocation}.} On input $(\textsf{revoke}, \sid, x)$ from a party $R$, and $(\textsf{revoke}, \sid)$ from $S$:

\begin{enumerate}

    \item Let $L = \leak{revoke}{\corrupted}(R,x)$. Send $(\textsf{revoke-req}, \sid, L)$ to $\sdv$.
      
    \item Receive $(\textsf{revoke-req-ok}, \sid)$ from $\sdv$.

    \item If $R \in \revoker{x}$ or $\exists R' \in \revoker{x} \cap \corrupted$, \alisanotes{assign $\revoked[\sid] = (x,R)$.}
     
    \item Output $(\textsf{revoke-resp-s}, \sid, \ok)$ to $S$.  Output $(\textsf{revoke-resp-r}, \sid, \ok)$ to $R$.
\end{enumerate}

\underline{\textsf{User notification.}} On input $(\textsf{unotify}, \sid, x)$ from $U$ and $(\textsf{unotify}, \sid)$ from $S$:
\begin{enumerate}

    \item Let $L = \leak{unotify}{\corrupted}(U,x)$. Let $X \gets \emptyset$. Send $(\textsf{unotify-req}, \sid, L)$ to $\sdv$.
    
    \item Receive $(\textsf{unotify-req-ok}, \sid, \alisanotes{\sids})$ from $\sdv$.
    
    \gray{If $S \in \corrupted$, take $X \gets \set{x\ | (x,R) \in  \revoked[\sids], \brokencons(U,R) = \true}$}.
        
    \item Let $b = \true$ iff $U \in \user{x}$ and $x \in \revoked \setminus X$.    
       
    \item  Output $(\textsf{unotify-resp-s}, \sid, \ok)$ to $S$. Output $(\textsf{unotify-resp-u}, \sid, b)$ to $U$.   
    
\end{enumerate}

\underline{\textsf{Verifier notification.}} On input $(\textsf{vnotify}, \sid)$ from $V$ and $(\textsf{vnotify}, \sid)$ from $S$:
\begin{enumerate}

    \item Let $L = \leak{vnotify}{\corrupted}(V)$. Let $X \gets \emptyset$. Send $(\textsf{vnotify-req}, \sid, L)$ to $\sdv$.

    \item Receive $(\textsf{vnotify-req-ok}, \sid, \sids)$ from $\sdv$.   
    
    \gray{If $S \in \corrupted$, take $X \gets \set{x\ | (x,R) \in  \revoked[\sids], \brokencons(V,R) = \true}$.}

    \item Assign $\revoked_V \gets \revoked \setminus X$.
    
    \item Output $(\textsf{vnotify-resp-s}, \sid, \ok)$ to $S$. Output $(\textsf{vnotify-resp-v}, \sid, \ok)$ to $V$.

\end{enumerate}

\underline{\textsf{Verification.}} On input $(\textsf{verify}, \sid, x, \alisanotes{V})$ from $U$, and $(\textsf{verify}, \sid)$ from $V$:
\begin{enumerate}
    
    \item Let $L = \leak{verify}{\corrupted}(U,V,x)$. Send $(\textsf{verify-req}, \sid, L)$ to $\sdv$.
    
    \item Receive $(\textsf{verify-req-ok}, \sid)$ from $\sdv$.
        
    \item Assert $x \notin \revoked_V$.
    
    \item Output $(\textsf{verify-resp-v}, \sid, x)$ to $V$. Output $(\textsf{verify-resp-u}, \sid, \ok)$ to $U$.
\end{enumerate}
\end{tcolorbox}
\caption{Ideal functionality $\funrevoke$ (Part 1). For each subroutine $f$, the leakage function $\leak{f}{\corrupted}$ outputs the values leaked to the adversary assuming the set $\corrupted$ of corrupted parties, and the particular leakages are listed in Table~\ref{tbl:leaks:funrevoke}. Gray colour denotes commands which are never called assuming a covertly corrupted server.}
\label{fig:funrevoke:1}
\end{figure}

\begin{figure}[ht]
\begin{tcolorbox}[colback=white,arc=0.3mm, boxrule=0.3mm]


\underline{\textsf{Break consistency}.} On input $(\brokencons,(V,R))$ from $\sdv$,
    
if $V \notin \corrupted$,$R \notin \corrupted$, $S \in \corrupted$, set $\brokencons(V,R) = \true$.

\underline{\textsf{Detect cheating}.} Upon receiving $(\textsf{cheat-detect},\sid,\sid_P)$ from $P \in \set{U,V}$,\\ $(\textsf{cheat-detect},\sid,\sid_R)$ from $R$, and $(\textsf{cheat-detect},\sid)$ from $S$ and $J$:

    \begin{enumerate}
    \item If at least one involved party $P$, $R$, $S$, $J$ is corrupted:
        \begin{itemize}
        \item Send $(\textsf{dectect-cheat-req},\sid,\sid_P,\sid_R)$ to $\sdv$.
        \item Wait for $(\textsf{dectect-cheat-ok},\sid,\corrupted')$ from $\sdv$ with $\corrupted' \subseteq \corrupted$.
        \end{itemize}
        Otherwise, set $\corrupted' = \emptyset$.
        
    \item If $\brokencons(V,R) = \true$, output $(\mathsf{verdict},\sid,\verdict(\set{S} \cup \corrupted'))$ to $J$.
    
    Otherwise, output $(\mathsf{verdict},\sid,\verdict(\corrupted'))$ to $J$.
    \end{enumerate}
    

\underline{\textsf{\alisanotes{Unfairness}.}} Whenever a message of the form $(response,\sid,y)$ for $response \neq \mathsf{verdict}$ is to be output to an honest party $P \neq S$, send $(response,\sid,y)$ to $\sdv$. Wait for $(response\textsf{-ok},\sid)$ from $\sdv$ and send $(response,\sid,y)$ to $P$.

\end{tcolorbox}
\caption{Ideal functionality $\funrevoke$ (Part 2, adversary interfaces and accountability).}
\label{fig:funrevoke:2}
\end{figure}


\begin{figure*}[htbp]
\begin{tcolorbox}[colback=white,arc=0.3mm, boxrule=0.3mm,fontupper=\normalsize]

\underline{Key generation} \\
On message $(\textsf{KeyGen}, sid)$ from both client and server:
\begin{enumerate}
    \item Send $(\textsf{KeyGen},sid)$ to $\simm$ and wait for $(\textsf{Key},sid, \pk)$ from $\simm$.
    \item If $\pk \neq \bot$, create an empty table $\mathsf{T}$ for storing encrypted messages, and set decryption counter $\mathit{ctr}_\mathsf{dec}\gets 0$.
    \item\label{fideal:keygen:return} \lizanew{If the party $P \in \{ \textsf{client, server} \}$ is corrupted, wait for $\foutput{i}{y_i}$ from $\adv$ for $i \in \{ \textsf{c,s}\}$. Otherwise, take $y_i = \pk$.}
    \item Output $(\textsf{Key}, sid, \lizanew{y_{\textsf{c}}})$ to client and $(\textsf{Key}, sid, \lizanew{y_{\textsf{s}}})$ to server.
\end{enumerate}
\vspace{0.2cm}

\underline{Encryption} \\
On message $(\textsf{Encrypt},sid,m)$ from some party $P$:
\begin{enumerate}

    \item If $P$ is corrupted, send $(\textsf{Encrypt},sid,\lizanew{P},m)$ to $\simm$. Otherwise, send $(\textsf{Encrypt},sid)$ to $\simm$.
    \item Wait for $(\textsf{Encrypt},sid,c)$ from $\simm$ (the adversary chooses the ciphertext for the given plaintext).
    \item Store the pair $(m,c)$ in table $\mathsf{T}$.
    \item\label{fideal:encrypt:return} If the party $P$ is corrupted, wait for $\foutput{\lizanew{P}}{y}$ from $\adv$. Otherwise, take $y = c$.
    \item Send $(\textsf{Encrypted},sid,y)$ to $P$.

\end{enumerate}
\vspace{0.2cm}

\underline{Decryption procedure} $Decrypt(\dec, c)$ for a function $\dec : \ciphertxtset \to \plaintxtset$ and a ciphertext $c \in \ciphertxtset$:
        \begin{itemize}
        \item If there exists a pair $(m,c)$ in $\mathsf{T}$: (for correctness we need that $c$ would decrypt to $m$)
        \begin{itemize}
            \item If $m$ is not unique (there are two different $m$ for the same $c$), take $m' = \bot$.
            \item If $m$ is unique, take $m' = m$.
        \end{itemize}
        \item Otherwise, compute $m' := \dec(c)$, running it only a polynomial number of steps (the polynomial is a parameter of $\fundec$). Add $(m',c)$ to $\mathsf{T}$.
        \item Return $m'$.
        \end{itemize}
        \end{tcolorbox}
\caption{Ideal encryption functionality $\fundec$ from~\cite{threshold-decryption} (1/2)} \label{fig:fideal-dec-1}
\end{figure*}

\begin{figure*}[tbp]
\begin{tcolorbox}[colback=white,arc=0.3mm, boxrule=0.3mm,fontupper=\normalsize]

\underline{Decryption} (honest client) \\
On message $(\textsf{Decrypt},sid,c)$ from \lizanew{client} and $(\textsf{Decrypt},sid)$ from party \lizanew{server}:
\begin{enumerate}
\item Send $(\textsf{Decrypt-init}, sid)$ to $\simm$.
\item Upon receiving $(\textsf{Decrypt-init}, sid, \ver)$ from $\simm$, where $\ver : \ciphertxtset \to \set{\mathsf{true},\mathsf{false}}$, compute $b \gets \ver(c)$.
    \begin{itemize}
        \item If $b = \textsf{true}$, send $(\textsf{Decrypt-good-c}, sid)$ to $\simm$.
        \item If $b = \textsf{false}$ or $b = \bot$, send $(\textsf{Decrypt-bad-c}, sid)$ to $\simm$. Proceed to the point~\ref{fideal:decrypt-hc:return}) with \lizanew{$y_\textsf{c} = y_\textsf{c} = \bot$}.
    \end{itemize}
\item The final output \lizanew{$y_\textsf{c}$ for client} depends on whether the server is corrupt.

    If the server is \emph{honest}, the decryption can only be delayed:
    \begin{itemize}
        \item Upon receiving $(\textsf{Decrypt-complete}, sid, \dec)$ from $\simm$, take $y_\textsf{c} \gets Decrypt(\dec,c)$, and $y_\textsf{s} \gets \top$.
    \end{itemize}

    If the server is \emph{corrupted}, $\simm$ tells whether decryption succeeded. 
    \begin{itemize}
        \item Upon receiving $(\textsf{Decrypt-complete}, sid, \dec)$ from $\simm$, take $y_\textsf{c} \gets Decrypt(\dec,c)$.
        \item Upon receiving $(\textsf{Decrypt-fail}, sid)$ from $\simm$, take $y_\textsf{s} \gets \bot$. 
    \end{itemize}
    \item\label{fideal:decrypt-hc:return} If the server is corrupted, wait for $\foutput{2}{y'_\textsf{s}}$ from $\adv$ and take $y_\textsf{s} \gets y'_\textsf{s}$.
        
\item Send $(\textsf{Decrypted},sid,y_\textsf{c})$ to client and $(\textsf{Decrypted},sid,y_\textsf{s})$ to server.
\end{enumerate}
\vspace{0.2cm}

\underline{Decryption} (corrupted client) \\
On message $(\textsf{Decrypt},sid,c)$ from \lizanew{client} and  $(\textsf{Decrypt},sid)$ from \lizanew{server}:

\begin{enumerate}
    \item Set $\mathit{ctr}_\mathsf{dec} \gets \mathit{ctr}_\mathsf{dec}+1$ and send $(\textsf{Decrypt-init}, sid, c)$ to $\simm$.
    \item $\simm$ tells whether decryption succeeded:
    \begin{itemize}
        \item Upon receiving $(\textsf{Decrypt-complete}, sid)$ from $\simm$, take $y_\textsf{s} \gets \top$.
        \item Upon receiving $(\textsf{Decrypt-fail}, sid)$ from $\simm$, take $y_\textsf{s} \gets \bot$.
\end{itemize}    
    \item\label{fideal:decrypt-cc:return} Wait for $\foutput{1}{y_\textsf{c}}$ from $\adv$.
\item Send $(\textsf{Decrypted},sid,y_\textsf{c})$ to \lizanew{client} and $(\textsf{Decrypted},sid,y_\textsf{s})$ to \lizanew{server}.
\end{enumerate}

\underline{Silent decryption (corrupted client)}. On message $(\textsf{Decrypt-msg}, sid, \dec, c')$ from $\simm$:
\begin{enumerate}
    \item Compute $m = Decrypt(\dec,c')$.
    \item Set $\mathit{ctr}_\mathsf{dec}\gets\mathit{ctr}_\mathsf{dec} - 1$. If $\mathit{ctr}_\mathsf{dec}<0$, then stop.
    \item Send $(\textsf{Decrypted}, sid, m)$ to $\simm$.
\end{enumerate}

\underline{Verify decryption.} On input $(\textsf{Verify},\sid,(c,m))$ from the client and from the verifier:
\begin{enumerate}
    \item If there exists a pair $(m, c)$ in \alisanotes{$\mathsf{T}$}:
    \begin{itemize}
        \item If $m$ is not unique (there are two different $m$ for the same $c$), output $(\textsf{Verify-fail},\sid)$ to the verifier.
        \item If $m$ is unique, output $(\textsf{Verify-success},\sid)$ to the verifier.
    \end{itemize}
     \item Otherwise, output $(\textsf{Verify-fail},\sid)$ to the verifier.   
\end{enumerate}

\end{tcolorbox}
\caption{Ideal encryption functionality $\fundec$ from~\cite{threshold-decryption} (2/2), \alisa{extended with verification interface.}} \label{fig:fideal-dec-2}
\end{figure*}
 



\begin{figure}[htbp]
\begin{tcolorbox}[colback=white,arc=0.3mm, boxrule=0.3mm,fontupper=\normalsize]
$\underline{\textsf{Initialization}}$: On input $(\textsf{init})$ from $\sdv$, set $\mathsf{c}_\mathrm{S} \gets \mathsf{c}_\mathrm{C} \gets T \gets 0$; $b_\mathrm{lq} \gets b_\mathrm{OK} \gets b_\mathrm{sk} \gets 1$; \liza{$\ctr_s \gets 0$};$\ver\gets \bot$. Send $(\textsf{init})$ to $\sdv$.

$\underline{\textsf{Key Generation}}$: On input $(\textsf{keygen}, \sid, L,T_0, \pin)$ from client and $(\textsf{keygen},\sid,T_0)$ from server, if $\ver$ is already recorded or $\pin \notin \{0,\ldots,L-1\}$, ignore this query. Otherwise:
{
\begin{itemize}
    \item Send $(\textsf{keygen-init})$ to adversary $\mathcal{S}$. 
    \item Upon receiving $(\textsf{key}, \lizanew{\pk})$ from $\mathcal{S}$, where \lizanew{$\pk$} is public key, store $(\lizanew{\pk}, L, \pin, T_0)$ and send $(\textsf{key}, \lizanew{\pk})$ to both client and server.
\end{itemize} 
}
$\underline{\textsf{Corrupt server}}$: On input $(\textsf{corrupt-server})$ from $\mathcal{S}$: set $\mathsf{c}_\mathrm{S} \gets 1$ and send $(\textsf{corrupt-server})$ to server.

$\underline{\textsf{Corrupt client}}$: On input $(\textsf{corrupt-client},\ell)$ from $\sdv$, where $\ell\in\{1,2,3\}$:
{
\begin{itemize}
    \item If $\mathsf{c}_\mathrm{C}>\ell$, ignore this query.
    \item Set $b_\mathrm{lq}\gets\bot$. If \lizanew{$\pk=\bot$}, set $\mathsf{c}_\mathrm{C}\gets 3$. Otherwise, set $\mathsf{c}_\mathrm{C} \gets \ell$. Send $(\textsf{corrupt-client},\ell)$ to client.
\end{itemize} 
}

\underline{\textsf{Derived values:}}
$b_\mathrm{KG}\equiv[\pk\not=\bot]$.

\underline{\textsf{Trigger}}: If $\mathsf{c}_\textsf{S}=\mathsf{c}_\textsf{C}=1$, then set $\mathsf{c}_\textsf{C}\gets 2$.

\underline{\textsf{Trigger}}: If $\mathsf{c}_\mathrm{C}\geq 2$ and $b_\mathrm{KG}$ then send $(\textsf{corrupt-pin},\pin)$ to $\sdv$.

\underline{\textsf{Forwarding}}: If $\mathsf{c}_\mathrm{C}=3$ [resp. $\mathsf{c}_\textsf{S}=1$], then inputs from client [resp. server] are forwarded to $\sdv$; and $\sdv$ determines the outputs that client [resp. server] receives.
\end{tcolorbox}
\caption{Ideal functionality for blind server-supported signing $\funsign$ (1/3)}
\label{fig:fideal-sign-1}
\end{figure}

\begin{figure}[tbp]
\begin{tcolorbox}[colback=white,arc=0.3mm, boxrule=0.3mm,fontupper=\normalsize]
%
\underline{\textsf{Signing by client and server}}: On input $(\textsf{sign},\sid,M,\pin')$ from the client and $(\textsf{sign},\sid)$ from the server, ignore unless $b_\mathrm{KG}\wedge b_\mathrm{OK}$. Otherwise:
{
\begin{itemize}
    \item Call $\textsf{clone-check}(1)$ and $\textsf{process-pin}(\pin')$. If successful, compute $\sigma$ using $\textsf{request-sig}(\sid,M)$.
    \item If $\sigma\not=\bot$, output $(\textsf{sign-success},\sid, M, \sigma)$ to client, and $(\textsf{sign-success},\sid)$ to server. Otherwise, return $(\textsf{sign-fail},\sid)$ to client and server.
\end{itemize} 
}
\underline{\textsf{Signing by client and adversary}}: On input $(\textsf{sign},\sid,M,\pin')$ from the client and $(\textsf{sign-server},\sid)$ from $\sdv$, ignore unless $b_\mathrm{KG}\wedge [\mathsf{c}_\mathrm{S}=1]$. Otherwise:
{
\begin{itemize}
    \item Call $\textsf{process-pin}(\pin')$ and then compute $\sigma$ by querying $\textsf{request-sig}(\sid,M)$.
    \item If $\sigma\not=\bot$, output $(\textsf{sign-success},\sid, M, \sigma)$ to client. Otherwise, return $(\textsf{sign-fail},\sid)$ to both client and $\sdv$.
\end{itemize}  
}
\underline{\textsf{Signing by adversary and server}}: On input $(\textsf{sign-client},\sid,\pin')$ from $\sdv$ and $(\textsf{sign},\sid)$ from the server: ignore unless $b_\mathrm{KG}\wedge b_\mathrm{OK}\wedge [\mathsf{c}_\mathrm{C}>0]$. Otherwise go to step $\mathsf{c}_\mathrm{C}\in\{1,2,3\}$:
{\vspace{0cm}
\begin{enumerate}
\item Call $\textsf{clone-check}(0)$ and $\textsf{process-pin}(\pin')$. If successful, set $\mathsf{c}_\mathrm{C}\gets 2$; go to step 3.
\item Call $\textsf{clone-check}(0)$. If successful, go to step 3.
\item Send $(\textsf{sign-init},\sid)$ to $\sdv$ and upon receiving $(\textsf{sign-success},\sid)$, output $(\textsf{sign-success},\sid)$ to server. \liza{Set $\mathsf{ctr}_{s} \gets \mathsf{ctr}_{s} + 1$.} Otherwise, return $(\textsf{sign-fail})$ to both $\sdv$ and server.
\end{enumerate}
}
\underline{\textsf{Verification}}: On input $(\textsf{verify},\sid,M,\sigma)$ from a party $P$, {send $(\textsf{verify},\sid,M,\sigma)$ to $\sdv$. Upon receiving $(\textsf{verified},\sid,M,\sigma,\phi)$ from $\sdv$ with $\phi \in \{ 1, 0 \}$:
\begin{itemize}
    \item if $(M,\sigma,b')$ is recorded, then set $b = b'$
    \item \liza{if $\phi = 1$, $(M,\sigma,b)$ has not been recorded before, and  $b_\mathrm{sk}=1$, set $\ctr_s \gets \ctr_s - 1$:
    \begin{itemize}
        \item if $\ctr_s \geq 0$, record $(M,\sigma,b =\phi)$.
        \Comment{saving valid signature message pairs that have not been saved during "signing by adversary and server" queries}
        \item if $\ctr_s < 0$, set $b = 0$ and record $(M,\sigma,b)$
        \Comment{(one-more unforgeability) the adversary created $n+1$ valid signatures with $n$ signing queries}
    \end{itemize}}
    \item else, set $b = \phi$ and record $(M,\sigma,b)$.
    \item output $(\textsf{verified},\sid,M,\sigma,b)$.
\end{itemize}
}
\end{tcolorbox}
\caption{Ideal functionality for blind server-supported signing $\funsign$ (2/3)}
\label{fig:fideal-sign-2}
\end{figure}

\begin{figure}[htbp]
\begin{tcolorbox}[colback=white,arc=0.3mm, boxrule=0.3mm,fontupper=\normalsize]
\underline{\textsf{Supporting subroutine}} $\textsf{process-pin}(\pin')$: 
{\begin{itemize}
    \item If $\pin=\pin'$, set $T\gets0$ and give control back to the invoker.
    \item Otherwise, increment $T$. If $T\geq T_0$ then set $b_\mathrm{OK}\gets0$. Return $(\textsf{sign-fail},\sid)$ to the two parties (out of client, server, and $\sdv$) that initiated the signing.
\end{itemize} }

\underline{\textsf{Supporting subroutine}} $\textsf{clone-check}(d)$: 
{\begin{itemize}
    \item If $b_\mathrm{lq}=1-d$, set $b_\mathrm{OK}\gets0$ and return $(\textsf{sign-fail},\sid)$ to the two parties that initiated the signing.
    \item Otherwise, set $b_\mathrm{lq} \gets d$ and return to the invoker.
\end{itemize} }

\underline{\textsf{Supporting subroutine}} $\textsf{request-sig}(\sid,M)$:
\begin{itemize}
    \item {Send $(\textsf{sign-init},\sid)$ to $\mathcal{S}$.}
    \item Upon receiving $(\textsf{signature},\sid,\sigma)$ from $\mathcal{S}$, check whether $(M,\sigma,0)$ is already stored. If true, restart $\textsf{request-sig}(\sid,M)$.
    \item {Create record $(M,\sigma,1)$ and return $\sigma$.}
\end{itemize}
\end{tcolorbox}
\caption{Ideal functionality for blind server-supported signing $\funsign$ (3/3)}
\label{fig:fideal-sign-3}
\end{figure}

\input{fullproof.tex}

\end{document}

%% file: fullproof.tex
\newpage
\section{Detailed benchmarks of Prettiness protocols}

The detailed benchmarks for the full system can be found in Table~\ref{tbl:fullbench:comm:detailed} (communication) and Table~\ref{tbl:fullbench:time:detailed} (time). For each routine, we collected the total communication between all parties, including communication induced by used sub-functionalities. The total time overhead, summed up over all involved parties, is considered for the most time-consuming operations, which are (significantly) slower than the others. For clarity, we report time overheads of the all used building block functionalities even if these are negligible.

In benchmarks, we assume that the token randomness (hidden into the details of protocol implementing $\funrevoke$) is generated by the user in advance, and is stored in encrypted form as a part of credential, to ensure full backup. This ciphertext corresponds to the  value $c_{\rev_r}$ in the benchmarks.

\begin{table}
\begin{tabular}{|l|l|}\hline
\multicolumn{2}{|c|}{\textbf{Issue}}\\ \hline  \hline
$\funsign^{\uid}(\mathsf{sign}\text{ and }\mathsf{verify})$ & 7638\\
$\store \to U: I,\iid$ & 128\\
$\store \to I: U, id_{\cid},\uid, \aids, \sigma_{\req}, c_{\rev_r}$ & 3994\\
$U \to I: \rev_U$ & 32\\
$\funrevoke(\mathsf{issue})$ & 48\\
$I \to \store: \cred, \sigma_\cred$ & 38802\\
$\store \to U: \cred, \sigma_\cred $ & 38802\\
$(n + 2) \times \fundec^{\uid}(\mathsf{Dec})$ & 1272\\
\hline
 Total: & 90716\\ \hline
\end{tabular}
\\
\begin{tabular}{|l|l||l|l|}\hline
\multicolumn{2}{|c|}{\textbf{Get all Creds}} & \multicolumn{2}{|c|}{\textbf{Revocation by Issuer}}\\ \hline  \hline
$\funsign^{\uid}(\mathsf{sign}\text{ and }\mathsf{verify})$ & 674 & $I \to ST: m, \sigma_{\req}$ & 3195\\
$\store \to U: data$ & 3905800 & $\funrevoke(\mathsf{revoke})$ & 160\\
\hline
 Total: & 3906474 & Total: & 3355\\ \hline
\end{tabular}
\begin{tabular}{|l|l||l|l|}\hline
\multicolumn{2}{|c|}{\textbf{Revocation by User}} & \multicolumn{2}{|c|}{\textbf{Present}}\\ \hline  \hline
$\funsign^\uid(\mathsf{sign}\text{ and }\mathsf{verify})$ & 6422 & $\funsign^{\uid}(\mathsf{sign}\text{ and }\mathsf{verify})$ & 6424\\
$\fundec^\uid(\mathsf{Dec})$ & 106 & $(n+1) \times \fundec(\mathsf{Dec})$ & 1166\\
$\funrevoke(\mathsf{revoke})$ & 160 &  & \\
\hline
 Total: & 6688 & Total: & 7590\\ \hline
\end{tabular}
\begin{tabular}{|l|l||l|l|}\hline
\multicolumn{2}{|c|}{\textbf{DB update}} & \multicolumn{2}{|c|}{\textbf{Verify}}\\ \hline  \hline
$\funrevoke(\mathsf{(v/u)notify})$ & 32000114 & $RP \to U: \ch$ & 32\\
 &  & $\funsign^{\uid}(\mathsf{sign})$ & 80900\\
 &  & $U \to RP: data$ & 40530\\
 &  & $2 \times \funrevoke(\mathsf{verify})$ & 324\\
 &  & $(n+2) \times \fundec(\mathsf{Verify})$ & 384\\
\hline
 Total: & 32000114 & Total: & 122170\\ \hline
\end{tabular}
\caption{Communication (in bytes) for the full system for $N=100$ credentials per user, $n = 10$ attributes per credential, $m = 10^6$ revoked tokens.}\label{tbl:fullbench:comm:detailed}
\end{table}

\begin{table}
\begin{tabular}{|l|l|l|}\hline
\multicolumn{3}{|c|}{\textbf{Issue}}\\ \hline  \hline
U & Generate randomness of $\rev$ in advance, encrypt it using $\fundec^{\uid}(\mathsf{Enc})$ & 0.282\\
$U$, $\store$ & $\funsign^{\uid}(\mathsf{sign}\text{ and }\mathsf{verify})$ & 0.1\\
$I$ & $\funsign^{\uid}(\mathsf{verify})$ & 0.001\\
$I$ & $(n + 2) \times \fundec^{\uid}(\mathsf{Enc})$ & 3.384\\
$I$ & $\funs^{\iid}(\mathsf{sign})$ & 0.001\\
$U$, $I$ & $\funrevoke(\mathsf{issue})$ & 0.0\\
$\store$ & $\funs^{\iid}(\mathsf{verify})$ & 0.001\\
$U$ & $\funs^{\iid}(\mathsf{verify})$ & 0.001\\
$U$ & update $h_{\cred}$ & 0.0003\\
$U$, $\store$ & $(n + 2) \times \fundec^{\uid}(\mathsf{Dec})$ & 4.344\\
\hline
 Total: &  & 8.1143\\ \hline
\end{tabular}
\begin{tabular}{|l|l|l||l|l|l|}\hline
\multicolumn{3}{|c|}{\textbf{Get all Creds}} & \multicolumn{3}{|c|}{\textbf{Revocation by Issuer}}\\ \hline  \hline
$U$, $\store$ & $\funsign^{\uid}(\mathsf{sign}\text{ and }\mathsf{verify})$ & 0.1 & $I$ & $\fundec^\uid(\mathsf{Enc})$ & 0.282\\
$U$ & compute $h_{\cred}$ & 0.0266 & $I$ & $\funs^\iid(\mathsf{sign})$ & 0.001\\
 &  &  & $\store$ & $\funs^\iid(\mathsf{verify})$ & 0.001\\
 &  &  & $I$, $\store$ & $\funrevoke(\mathsf{revoke})$ & 0.743\\
\hline
 Total: &  & 0.1266 & Total: &  & 1.027\\ \hline
\end{tabular}
\begin{tabular}{|l|l|l||l|l|l|}\hline
\multicolumn{3}{|c|}{\textbf{Revocation by User}} & \multicolumn{3}{|c|}{\textbf{Present}}\\ \hline  \hline
$U$ & $\fundec^\uid(\mathsf{Enc})$ & 0.282 & $U$ & $\fundec^\uid(\mathsf{Enc})$ & 0.282\\
$U$, $\store$ & $\funsign^\uid(\mathsf{sign}\text{ and }\mathsf{verify})$ & 0.1 & $U$, $\store$ & $\funsign^{\uid}(\mathsf{sign}\text{ and }\mathsf{verify})$ & 0.1\\
$U$, $\store$ & $\fundec^\uid(\mathsf{Dec})$ & 0.362 & $U$, $\store$ & $(n+1) \times \fundec(\mathsf{Dec})$ & 3.982\\
$U$, $\store$ & $\funrevoke(\mathsf{revoke})$ & 0.743 & $U$, $\store$ & $2 \times \funrevoke(\mathsf{unotify})$ & 1.346\\
\hline
 Total: &  & 1.487 & Total: &  & 5.71\\ \hline
\end{tabular}
\begin{tabular}{|l|l|l||l|l|l|}\hline
\multicolumn{3}{|c|}{\textbf{DB update}} & \multicolumn{3}{|c|}{\textbf{Verify}}\\ \hline  \hline
U/RP & $\funrevoke(\mathsf{(v/u)notify})$ & 0.716 & $U$, $\store$ & $\funsign^{\uid}(\mathsf{sign})$ & 0.1\\
 &  &  & RP & $\funsign^{\uid}(\mathsf{verify})$ & 0.001\\
 &  &  & RP & $\funs^{\iid}(\mathsf{verify})$ & 0.001\\
 &  &  & $U$, RP & $2 \times \funrevoke(\mathsf{verify})$ & 4.6\\
 &  &  & $U$, RP & $(n+2) \times \fundec^{\uid}(\mathsf{Verify})$ & 0.0\\
\hline
 Total: &  & 0.716 & Total: &  & 4.702\\ \hline
\end{tabular}
\caption{Time (in seconds) for the full system for $N=100$ credentials per user, $n = 10$ attributes per credential, $m = 10^6$ revoked tokens.}\label{tbl:fullbench:time:detailed}
\end{table}

\section{UC proof for the full Prettiness system}
\label{app:mainproof}


To prove security of the proposed system, we define a simulator $\simm$ mediating communication between the $\adv$ and the ideal functionality $\ifun$, trying to convince $\adv$ that it is communicating with the real protocol $\system$. Description of the simulator is split into different phases where each phase has separate sub-description for the different corruption models. The involved party that is not listed among the "malicious" parties is treated as honest, and is being simulated by $\simm$. We only consider "interesting" corruptions, e.g. we do not describe simulation of a corrupted issuer alone if cheating is only possible in collaboration with a corrupted storage. We do not describe the case where all involved parties are corrupted, as there is nothing to be simulated in the real protocol, and their interactions with the sub-functionalities affect only corrupted parties. We also do not describe the case where all involved parties are honest, as the simulator would only need to simulate the notifications of different steps to $\adv$, and simulations of these are already covered by the cases where that party is honest.

The protocol $\system$ is using functionality $\funrevoke$ which is parametrized by leakage functions $\leak{f}{\corrupted}$ for a subroutine $f$ and the set of corrupted parties $\corrupted$. We write $L = \leak{f}{\corrupted}(x)$ to denote leakage $L$ that will be shown to the adversary, computed based on the internal state and the inputs $x$ of simulated $\funrevoke$. Throughout the proof, we write $\idfun(cmd,\sid,x)$ to denote running the subroutine $cmd$ of $\idfun$ with $\sid$ and input $x$.

For the purpose of this paper, we allow adversary to corrupt user \lizanew{either before key generation, reaching level $c_{\mathsf{U}} = 3$ or after the key generation phase, reaching level $c_{\mathsf{U}} = 1$ (and subsequently $c_{\mathsf{U}} = 2$). In the latter case, we prove that RP is protected from malicious user. } Adversary can increase its level of corruption while making silent requests to the functionality. Processing of each silent request begins with the call to authentication sub-routine and if it results in successful authentication, corruption level can increase.

\subsection{Initialisation}
In this phase, initialisation of all used sub-functionalities is simulated. $\simm$ also initializes $\ifun$ using the ideal adversary's interface. \alisanotes{We omit simulation of initialisation for a corrupted user since, as we mentioned before, for simplicity we will consider further corruption after the key generation. Other than that, initialisation affects only the AMS and the Issuer.}

\underline{Malicious AMS (storage)}:
\begin{enumerate}
    \item On input $(\textsf{keygen},\sid)$ from $\adv$ to $\funsign^{\uid}$, send $(\textsf{keygen-init},\sid)$ to $\adv$ on behalf of $\funsign^{\uid}$. Wait for $(\textsf{key},\sid, \ver)$ from $\adv$. Store $\ver$ as verification function for threshold signatures.
    
    \item On input $(\textsf{keygen},\sid)$ from $\adv$ to $\fundec^{\uid}$, send $(\textsf{keygen-init}, \sid)$ to $\adv$ on behalf of $\fundec^{\uid}$. Wait for $(\textsf{key},\sid,\pk)$ from $\adv$. Store $\pk$ as encryption key for the threshold decryption scheme.
    
    
    \item On input $(\textsf{setup}, \tokens, \sid)$ from $\adv$ to $\funrevoke$, check that $|\tokens| \geq N$ (here $N$ is a parameter of $\funrevoke$). Store $\tokens$ and initialise empty sets of revoked tokens $\revoked \gets \emptyset$ and corrupted parties $\corrupted \gets \emptyset$. Initialise empty 
    sets $\user{x}$ and $\revoker{x}$.

    \item Send $(\textsf{setup},\sid)$ to $\ifun$.
    \item Initialise internal table $\tenc^\uid$ to store all the plaintext-ciphertext pairs, $\tsign^\uid$ to store all the generated threshold signatures, $\tisign^\iid$ to store all the issuer signatures, $\tiss$ to store all the issued credentials. 
    Instantiate set of possible attribute values \alisanotes{$\mathcal{M}_{\aid}$ for an attribute identifier $\aid$.}
    \item \alisanotes{Define $\textsf{Present}(\lizanotes{\cred,\aidsd})$ as}
    \begin{enumerate}
        \item For all $\aid_i \in \aidsd$, find corresponding $\atr_i$ in $\cred$ and add them to the $\atrsd$ array.
        \item Set \alisa{$\pres =(\aidsd,\atrsd,\iid)$, where $\iid$ is extracted from $\cred$,} and return $\pres$.
    \end{enumerate} Define $\textsf{Ver}(\alisanotes{\cred},\pres)$ as
    \begin{enumerate}
        \item \alisa{Parse $\pres = (\aidsd,\atrsd,\iid)$.}
        \item For all $\aid_i \in \aidsd$ and $\atr_i \in \atrsd$, verify that there is exist $\aid'_j$ and $\atr'_j$ in $\cred$ such that $\aid_i = \aid'_j$ and $\atr_i = \atr'_j$. \alisa{Verify that $\cred$ contains $\iid$.} If these checks pass, set $b = 1$. Otherwise, set $b = 0$. Return $b$.
    \end{enumerate}
    \item On message $(\textsf{setup-ok})$ from $\ifun$, send $(\textsf{setup-ok},\sid,\textsf{Present}, \textsf{Ver})$ to $\ifun$.
\end{enumerate}

\underline{Malicious AMS (storage) and Issuer}:
\begin{enumerate}
    \item On input $(\textsf{keygen},\sid)$ from $\adv$ to $\funsign^{\uid}$, send $(\textsf{keygen-init},\sid)$ to $\adv$ on behalf of $\funsign^{\uid}$. Wait for $(\textsf{key},\sid, \ver)$ from $\adv$. Store $\ver$ as verification function for threshold signatures.
    \item On input $(\textsf{keygen},\sid)$ from $\adv$ to $\fundec^{\uid}$, send $(\textsf{keygen-init}, \sid)$ to $\adv$ on behalf of $\fundec^{\uid}$. Wait for $(\textsf{key},\sid, \pk)$ from $\adv$. Store $\pk$ as encryption key for threshold decryption scheme.
    
    
    \item On input $(\textsf{setup}, \tokens, \sid)$ from $\adv$ to $\funrevoke$, check that $|\tokens| \geq N$ (here $N$ is a parameter of $\funrevoke$). Store $\tokens$ and initialise empty sets of revoked tokens $\revoked \gets \emptyset$ and corrupted parties $\corrupted \gets \emptyset$. Initialise empty 
    sets $\user{x}$ and $\revoker{x}$.
    \item Send $(\textsf{setup},\sid)$ to $\ifun$.
    \item Initialise internal table $\tenc^\uid$ to store all the plaintext-ciphertext pairs, $\tsign^\uid$ to store all the generated threshold signatures, $\tisign^\iid$ to store all the issuer signatures, $\tiss$ to store all the issued credentials. 
    \item \alisanotes{Define $\textsf{Present}(\lizanotes{\cred,\aidsd})$ as}
    \begin{enumerate}
        \item For all $\aid_i \in \aidsd$, find corresponding $\atr_i$ in $\cred$ and add them to the $\atrsd$ array.
        \item Set \alisa{$\pres = (\aidsd,\atrsd,\iid)$, where $\iid$ is extracted from $\cred$,} and return $\pres$.
    \end{enumerate} Define $\textsf{Ver}(\alisanotes{\cred},\pres)$ as
    \begin{enumerate}
        \item Parse \alisa{$\pres = (\aidsd,\atrsd,\iid)$}.
        \item For all $\aid_i \in \aidsd$ and $\atr_i \in \atrsd$, verify that there exist $\aid'_j$ and $\atr'_j$ in $\alisanotes{\cred}$ such that $\aid_i = \aid'_j$ and $\atr_i = \atr'_j$. \alisa{Verify that $\cred$ contains $\iid$.} If these checks pass, set $b = 1$. Otherwise, set $b = 0$. Return $b$.
    \end{enumerate}
    \item On message $(\textsf{setup-ok})$ from $\ifun$, send $(\textsf{setup-ok},\sid,\textsf{Present}, \textsf{Ver})$ to $\ifun$.
\end{enumerate}

\lizanew{\underline{Malicious User:} if the adversary corrupts user before key generation, corruption reaches level $c_{\mathsf{U}} = 3$ and we cannot provide any security or privacy guarantees to the user in that case. In the analysis below, we do not assume full corruption of the device and analyse user corruption on level $c_{\mathsf{U}} = 1$ (that can increase to $c_{\mathsf{U}} = 2$ throughout the interaction with system).}

\subsection{\alisanotes{Adversary-controlled subroutine calls}}

In the real protocol, we are using functionalities with which the adversary may interact at any time, without input of any honest parties. Let us describe which cases there are, and how they are simulated. 

\begin{enumerate}
    \item The adversary may \textbf{encrypt} any message using $\fundec^{\uid}$.
    
    On message $(\textsf{Encrypt},\sid,m)$ to $\fundec^{\uid}$ from corrupted party interface, send $(\textsf{Encrypt},\sid,m)$ to $\adv$ and wait for $(\textsf{Encrypt},\sid,c)$. Store the pair $(m,c)$ in table $\tenc$. Send $(\textsf{Encrypted},\sid,c)$ to $\adv$.

    \item The adversary can \textbf{decrypt} a message using $\fundec^{\uid}$ only if both the user and the AMS are corrupted.
    
    On message $(\textsf{Decrypt},\sid,c)$ from corrupted user, and $(\textsf{Decrypt},\sid)$ from AMS, if there exists a pair $(m, c)$ in $\tenc$:
     \begin{itemize}
         \item If $m$ is not unique, set $m' = \bot$,
         \item If $m$ is unique, set $m' = \rev$.
     \end{itemize}
     Output decrypted $m'$ to the user.
    
    \item The adversary may \textbf{verify} any signature using $\funs^{\iid}$.
    
    Upon receiving $(\textsf{verify},\sid,msg,\sigma)$ from adversary, 
    \liza{if $(m, \sigma, 1)$ is recorded then set $f = 1$. Else, if the signer is not corrupted, and no entry $(m, \sigma', 1)$ for any $\sigma'$ is recorded, then set $f = 0$ and record the entry $(m, \sigma, 0)$. Else, if there is an entry $(m, \sigma, f')$ recorded, then set $f= f'$. Else, set $f= \phi$, and record the entry $(m, \sigma', \phi)$. Output $(\textsf{verified},\sid,msg,f)$ to adversary.}
    
    \item Similarly, the adversary may \textbf{verify} any signature using $\funsign^{\uid}$. 
       
    \item If the issuer is corrupted, the adversary may \textbf{sign} any message on their behalf using $\funs^{\iid}$.
    
    Upon receiving $(\textsf{sign},\sid,msg)$ to $\funs^{\iid}$ from corrupted signer interface, set $\sigma = \textsf{s}(msg)$ and verify that $\textsf{v}(\sigma,msg)=1$. If verification succeeds, output $(\textsf{signature},\sid,msg,\sigma)$ to $\adv$ and record entry $(msg,\sigma)$ to $\tisign$. Otherwise, output error message to $\adv$.

    \item \alisa{Similarly, the adversary may \textbf{sign} a message using $\funsign^{\uid}$.}  Differently from $\funs^{\iid}$, the adversary needs to corrupt \emph{both} the user and the AMS.

    \item If the AMS is corrupted, the adversary can make \textbf{requests to $\funrevoke$}, which allow to issue tokens to corrupted parties, get all revoked tokens of corrupted users, and revoke tokens by a corrupted party.  Only the last point may have an impact on an honest party and needs to be synchronised with $\ifun$.
    
    Let $\rev$ be the token that $\adv$ wants to silently revoke. $\simm$ checks whether $\rev$ has already been linked to some $\cid$ in the table  $\tiss$ of all issued credentials.
    \begin{itemize}
        
    \item If there is such $\cid$ (or even several $\cid$-s, which is possible for a corrupted user), then the adversary should have learned $\rev$ from a corrupted user, or a corrupted issuer. $\simm$ uses silent request interface of $\ifun$ to add $\cid$ to $\mathsf{T}_r$ of $\ifun$.
    
    \item \alisa{If there is no such $\cid$, then $\rev$ has not been issued yet (using $\ifun$). In this case, $\simm$ only adds $\rev$ to the set $\revoked$ of simulated $\funrevoke$.} 
    \Comment{As tokens are freshly sampled, this may happen only for a corrupted user, for whom addition to $\mathsf{T}_r$ (revocation of corrupted user's credential by a corrupted party) does not need to be simulated, and $\cid$ of a presented credential may even not exist if it does not concern any honest issuers.}
    \end{itemize}

    

    
    
\end{enumerate}

\subsection{\alisanotes{Malicious User Authentication}}\label{app:uauth}

A malicious user ($c_{\mathsf{U}} = 3$) is able to call routines of $\ifun$ at any time. However, \alisa{if the user has been honest in the beginning, the adversary does not have direct access to the user input interface of $\ifun$, and needs to use the special silent request interface, which succeeds based on the corruption level of the user.} In the real protocol, the adversary needs to get access to the secret key of $\funsign^{\uid}$, which authenticates the user, which means guessing the PIN, leading to corruption $c_{\mathsf{U}} = 2$. 
We define a separate process of the simulator, which deals with user authentication.

\begin{enumerate}


    \item \alisa{Wait for a connection request of the form $(\rtag_{routine},\sid,\uid)$ from $\adv$.}

    \item \alisa{Upon receiving input $(\textsf{sign-phone},\sid,msg,\pin')$ from $\adv$ to $\funsign^{\uid}$}, 
    
    send $(\textsf{invoke}, routine, \sid, \uid)$ to $\ifun$. Receive $(\textsf{auth},\sid)$ from $\ifun$ .
    
    
    \item Simulate threshold signature:
    \begin{enumerate}
        \item If $\pin' = \pin$, set $T \gets 0$, $c_{\mathsf{U}} \gets 2$, send $(\textsf{corrupt-party},\uid,\sid)$ to $\ifun$  and proceed with the next step \Comment{$\adv$ was able to guesses $\pin$ and hence increased its corruption level}. Otherwise, call $\textsf{clone-check}(0)$ and return $(\textsf{sign-fail})$ to $\adv$ and $(\textsf{auth-fail})$ to $\ifun$ \Comment{$\adv$ was not able to guesses $\pin$ and authentication attempt fails.}
        \item Call $\textsf{clone-check}(0)$, if successful proceed with the next step. \Comment{If adversary guesses the $\pin$ and knows clone-detection bitstring, it can successfully create at least one signature}
        \item \lizanew{Send $(\textsf{sign-init}, \sid)$ to $\adv$ through $\funsign^\uid$. Upon  receiving $(\textsf{sign-success}, \sid)$ from $\adv$, set $\ctr_s \gets \ctr_s + 1$. Otherwise, return $(\textsf{sign-fail})$.}
    \end{enumerate}

    \item \lizanew{Verify signature: after receiving $(\sigma_{\req},m)$ from $\adv$, send $(\textsf{verify},\sid,m,\sigma_{\req})$ to $\adv$. Upon receiving $(\textsf{verified},\sid, m,\sigma_{\req},\phi)$ from $\adv$ with $\phi \in \{ 1, 0 \}$:
    \begin{itemize}
        \item if $(m,\sigma_{\req},b')$ is recorded, then set $b = b'$
        \item if $\phi = 1$, $(m,\sigma_{\req},b)$ has not been recorded before, and  $b_\mathrm{sk}=1$, set $\ctr_s \gets \ctr_s - 1$:
        \begin{itemize}
        \item if $\ctr_s \geq 0$, record $(m,\sigma_{\req},b =\phi)$.
        \Comment{saving valid signature message pairs that have not been saved during "signing by adversary and server" queries}
        \item if $\ctr_s < 0$, set $b = 0$ and record $(m,\sigma_{\req},b)$
        \Comment{(one-more unforgeability) the adversary created $n+1$ valid signatures with $n$ signing queries}
    \end{itemize}
    \item else, set $b = \phi$ and record $(m,\sigma_{\req},b)$.
    \end{itemize}}

    \item 
    Send $(\textsf{auth-ok},\sid)$ to $\ifun$.  
    Assign $\req(\alisa{routine}, \alisa{\uid},\sid)=(msg,\sigma)$.

    \item As the result, whenever $\req(\alisa{routine},\alisa{\uid},\sid)$ is defined, $\ifun$ is waiting for a message $(\textsf{input}, \uid, routine,\sid,x)$ from the simulator to proceed with the routine call.
    \end{enumerate}

\subsection{Issuing}
Since there are three parties involved, the simulator needs to simulate behaviour of several parties at once. For clarity, we label each simulation step, which party is being simulated.

We do not consider the cases where AMS is \emph{honest}. This case would only be interesting in the case when we are \emph{protecting} the AMS, as in the other cases corrupting AMS only makes the adversary stronger. However, AMS has no inputs and outputs to $\zenv$, so there would be nothing to additionally simulate for an honest AMS.

\underline{Malicious AMS (storage)}:
\begin{enumerate}
    \item \usersim: On message $(\textsf{issue-req},\sid,\alisa{\uid, \iid, \aids})$ from $\ifun$, \alisa{establish $(\alisa{\issuereq}, \sid,\uid)$ connection with the server}. Generate $id_\cid \gets \pirkeys$, set $msg=(\issuereq,id_\cid,\uid,\aids,\iid)$, and wait for input $(\textsf{sign-server}, \sid)$ from $\adv$. \lizanew{Send $(\textsf{sign-init}, \sid)$ to $\adv$ on behalf of $\funsign^{\uid}$.} Upon receiving \lizanew{$(\textsf{signature}, \sid, \sigma)$} from $\adv$, check whether $(msg, \sigma,0)$ is already stored. In this case, restart signing. \lizanew{Otherwise, store $(msg, \sigma, 1)$ in table $\tsign^\uid$ and send $(\textsf{sign-success},\sid)$ to the $\adv$}. Otherwise, stop execution.

    \item \lizanew{\usersim: Send $msg=(\issuereq,id_\cid,\uid,\aids,\iid)$ and $\sigma$ to $\adv$.}
    

    \item \issuersim: Upon receiving \alisa{$(U',id_\cid',\uid',\aids',\sigma')$} from $\adv$ with $n$ entries in $\aids'$, verify the signature $\sigma'$ using $\funsign^{\uid'}$. This ensures that the message originates from the user $\uid'$.

    \item \issuersim: \alisa{Let $\cid \gets (id'_\cid,\textsf{summary}(\uid',\aids',\iid))$. Verify freshness of $\cid$ (and hence $id'_\cid$). This ensures that the corrupted AMS has not made a repetition  attack, so $\uid = \uid'$, $\aids' = \aids$ and $id'_\cid = id_\cid$.}

    \item \alisa{\usersim: Simulate $U$ generating a fresh token $\rev_U$ for itself using $\funrevoke$:}
    \begin{enumerate}
        \item Sample \alisa{$\rev_U \sample \tokens$}. Set $L = \leak{issue}{\set{\store}}((U,U),\rev_U) = (U,U)$,    and send $(\alisa{\textsf{issue-req}}, \sid, L)$ to $\adv$ on behalf of $\funrev$.
    
        \item On input $\alisa{(\textsf{issue-req-ok}, \sid)}$ from $\adv$, verify that there is no $U' \notin \corrupted$ s.t. $U' \in \user{\rev_U}$ (as $\rev_U$ is freshly sampled, it is unlikely for a sufficiently large $N$). Assign $\user{\rev_U} \gets \user{\rev_U} \cup \set{U}$, $\revoker{\rev_U} \gets \revoker{\rev_U} \cup \set{U}$.

        \item Update $\tokens \gets \tokens \setminus \set{\rev_U}$.
    \end{enumerate}

    \item \alisa{\usersim: Let $\adv$ send $(I',\iid)$ to $U$ on behalf of AMS. Simulate $U$ sending the revocation token $\rev_U$ to $I'$.}

    \item \alisa{\issuersim: Simulate $I$ receiving the revocation token $\rev'_U$ from $U'$.}
    
    \item \alisa{\usersim\ and \issuersim: If $U = U'$ and $I=I'$, jointly generate a fresh token $\rev_I$ using $\funrevoke$:}
    \begin{enumerate}
    \item Sample \alisa{$\rev_I \sample \tokens$}. Set $L = \leak{issue}{\set{\store}}(I,U,\rev_I) = (I,U)$, and send $(\alisa{\textsf{issue-req}}, \sid, L)$ to $\adv$ on behalf of $\funrev$.
    
    \item On input $(\alisa{\textsf{issue-req-ok}}, \sid)$ from $\adv$, verify that there is no $U' \notin \corrupted$  s.t. $U' \in \user{\rev_I}$ (as $\rev_I$ is freshly sampled, it is unlikely for a sufficiently large $N$). Assign $\user{\rev_I} \gets \user{\rev_I} \cup \set{U}$, $\revoker{\rev_I} \gets \revoker{\rev_I} \cup \set{I}$.

    \item Update $\tokens \gets \tokens \setminus \set{\rev_I}$. 
    \item \alisa{Output $(\textsf{issue-resp-u}, \sid, \ok)$ and $(\textsf{issue-resp-i}, \sid, \ok)$ to $\adv$ on behalf of $\funrevoke$. Now $\adv$ may additionally choose that $I$ or $U$ receive $\bot$ from $\funrevoke$. Proceed with simulating each of these parties only if it has not received $\bot$.}
    \end{enumerate}
    If $U' \neq U$ or $I \neq I'$, then the honest party interacts with the adversary impersonating $\uid$ or $\iid$. As the result, $\adv$ learns the issued $\rev_I$. An honest issuer receives a token $\rev'_I$.
    \Comment{We are \emph{not} hiding the token $\rev$ that will be assigned to the credential $\cid$ from a corrupted AMS. Hence, the adversary may know the link between $\cid$ and $\rev$ of issued credentials.}
    
    \item \alisa{\issuersim: Perform attribute retrieval and compose a credential with $\rev=(\rev'_I,\rev'_U)$:}
        \begin{enumerate}  
        \item Send $(\textsf{issue-ok},\sid,\alisanotes{\cid})$ to the $\ifun$ (start interaction with $\mathcal{Z}$).
        \item Wait for $(\textsf{issue-atr},\sid)$ from $\ifun$ ($I$ received attribute values from $\mathcal{Z}$).
        \item Send encryption request $(\textsf{Encrypt}, \sid)$ to $\adv$ through $\fundec^{\uid}$. Wait for response $(\textsf{Encrypt}, \sid, c_{\rev})$ from $\adv$ and store the pair $(\rev,c_{\rev})$ in table $\tenc^\uid$. 
        \item \liza{Set $\atr_0 \gets \pk_{\uid}$} \Comment{with $\pk_{\uid}$ coming from verification $\ver$ defined during the initialisation phase}
        \item Sample $n$ arbitrary attribute values $\atr_i \sample \mathcal{M}_{\aid_i}$. Send \liza{$n + 1$} encryption requests $(\textsf{Encrypt}, \sid)$ to $\adv$ on behalf of $\fundec^{\uid}$. Wait for responses $(\textsf{Encrypt}, \sid, c_i)$ from $\adv$ and store the pairs $(\atr_i,c_i)$ in table $\tenc^\uid$.

    \item Set $\cred = \langle \liza{(\aid_0,c_0)}, (\aid_1,c_1), \dots, (\aid_n,c_n), (\aid_{\rev}, c_{\rev}) \rangle$. 
    
    \item Generate signature on the credential $\cred$ as $\sigma_{\cred} = \textsf{s}(\cred)$ and verify that $\textsf{v}(\sigma_{\cred},\cred) = 1$ \alisanotes{(\textsf{s} and \textsf{v} are parameters of $\funs^{\iid}$)}. Save $(\sigma_{\cred},\cred)$ to table $\tisign^\iid$.
    
    \item Send $(\cred,\sigma_{\cred})$ to corrupt AMS server.
    \end{enumerate}
    

    \item \alisa{\usersim: Verify the received credential:}
        \begin{enumerate}
    \item Upon receiving $(\cred',\sigma'_{\cred})$ from corrupted AMS server, \alisanotes{check that $\cred'$ indeed contains $\aids$.}

    \item Verify issuer's signature $\sigma'_{\cred}$ on $\cred'$ \alisa{using $\funs^{\iid}$}:
    \begin{itemize}
        \item if $\textsf{v}(\sigma'_{\cred},\cred') = 1$ and no entry $(\sigma'_{\cred},\cred')$ is recorded in $\tisign^\iid$, output $(\textsf{issue-atr},\sid,\bot)$ to $\ifun$,
        \item otherwise, output is $\textsf{v}(\sigma'_{\cred},\cred')$.
    \end{itemize}
    
    \item Simulate decryption of $c'_{\rev}$ \alisanotes{(contained in $\cred'$)} by sending $(\textsf{decrypt-init},\sid)$ to $\adv$ on behalf of $\fundec^{\uid}$. 
    
    Upon receiving $(\textsf{decrypt-complete},\sid,\dec)$, if there exists a pair $(m,c'_{\rev})$ in $\tenc^\uid$:
    \begin{itemize}
        \item if $m$ is not unique, set $m' = \bot$,
        \item if $m$ is unique, set $m' = m$. 
    \end{itemize}
    Otherwise, compute $m' := \dec(c'_{\rev})$. Add $(m',c'_{\rev})$ to $\tenc^\uid$.
    
    \item Proceed if $m' = \alisa{(\rev_I,\rev_U)}$. \alisa{This implies $\rev_I' = \rev_I$ and $\rev_U' = \rev_U$.}

    \item \liza{Simulate decryption of $c'_{0}$ (contained in $\cred'$) by sending $(\textsf{decrypt-init},\sid)$ to $\adv$ on behalf of $\fundec^{\uid}$.} 
    
    \liza{Upon receiving $(\textsf{decrypt-complete},\sid,\dec)$, if there exists a pair $(m,c'_{0})$ in $\tenc^\uid$:
    \begin{itemize}
        \item if $m$ is not unique, set $m' = \bot$,
        \item if $m$ is unique, set $m' = m$. 
    \end{itemize}
    Otherwise, compute $m' := \dec(c'_{0})$. Add $(m',c'_{0})$ to $\tenc^\uid$.}
    
    \item \liza{Proceed if $m' = \pk_{\uid}$.}
    
    \item Simulate decryption of $n = |\aids|$ ciphertexts inside $\cred'$ by sending $(\textsf{decrypt-init},\sid)$ to $\adv$ on behalf of $\fundec^{\uid}$. Upon receiving $(\textsf{decrypt-complete},\sid,\dec)$, if there exists a pair $(\atr_i,c_i)$ in $\tenc^\uid$:
    \begin{itemize}
        \item If $\atr_i$ is not unique, set $m' = \bot$,
        \item If $\atr_i$ is unique, set $m' = \atr_i$.
    \end{itemize}
    Otherwise, compute $m' := \dec(c_i)$. Add $(m',c_i)$ to $\tenc^\uid$. 

     Let $\atr_i = m'$. Proceed if $\atr_i \in \mathcal{M}_{\aids_i}$. 
     
    
    
    

    \item \alisa{Compute $h_\cred \gets H(h_\cred,(\cid,\cred,\sigma_\cred))$.}

    \item Send message $(\textsf{atr-ok},\sid)$ to $\ifun$ and save $(\cid, \cred', \sigma'_{\cred}, \rev)$ to $\tiss$. \\
    \Comment{If both $U$ and $I$ finish successfully, then their views on the issued $\cid$, $\aids$, $\atrs$ match. This is what gets to $\tiss$ and the internal table $\mathsf{T}_a$ of $\ifun$. The credential data $\cred'$ and $\sigma'_\cred$ correspond to the user's view, as this is what adversary expect to see during presentation.}
    \end{enumerate}
\end{enumerate}

\underline{Malicious AMS (storage) and Issuer}:
Since the AMS does not provide additional protection for the user, we do not consider separately the case where only the issuer is corrupted, as the simulation would be the essentially same.
\begin{enumerate}
    \item On message $(\textsf{issue-req},\sid,\alisa{\uid, \iid, \aids})$ from $\ifun$, \alisa{establish $(\alisa{\issuereq},\sid,\uid)$ connection with the server}. Generate $id_\cid \gets \pirkeys$, set $msg=(\issuereq,id_\cid,\uid,\aids,\iid)$, and wait for input $(\textsf{sign-server}, \sid)$ from $\adv$. \lizanew{Send $(\textsf{sign-init}, \sid)$ to $\adv$ on behalf of $\funsign^{\uid}$.} Upon receiving \lizanew{$(\textsf{signature}, \sid, \sigma)$} from $\adv$, check whether $(msg, \sigma,0)$ is already stored. In this case, restart signing. \lizanew{Otherwise, store $(msg, \sigma, 1)$ in table $\tsign^\uid$ and send $(\textsf{sign-success},\sid)$ to the $\adv$}. Otherwise, stop execution.

    \item \lizanew{\usersim: Send $msg=(\issuereq,id_\cid,\uid,\aids,\iid)$ and $\sigma$ to $\adv$.}
    
    
    


    \item \alisa{Simulate $U$ generating a fresh token $\rev_U$ for itself using $\funrevoke$:}
    \begin{enumerate}
        \item Sample \alisa{$\rev_U \sample \tokens$}. Set $L = \leak{issue}{\set{\store}}((U,U),\rev_U) = (U,U)$,    and send $(\alisa{\textsf{issue-req}}, \sid, L)$ to $\adv$ on behalf of $\funrev$.
    
        \item On input $\alisa{(\textsf{issue-req-ok}, \sid)}$ from $\adv$, verify that there is no $U' \notin \corrupted$ s.t. $U' \in \user{\rev_U}$ (as $\rev_U$ is freshly sampled, it is unlikely for a sufficiently large $N$). Assign $\user{\rev_U} \gets \user{\rev_U} \cup \set{U}$, $\revoker{\rev_U} \gets \revoker{\rev_U} \cup \set{U}$.

        \item Update $\tokens \gets \tokens \setminus \set{\rev_U}$.
    \end{enumerate}

    \item \alisa{Let $\adv$ send $(I',\iid)$ to $U$ on behalf of AMS. Simulate $U$ sending the token $\rev_U$ to $I'$.}

    \item \alisa{Generate a fresh token $\rev_I$ using $\funrevoke$:}
    \begin{enumerate}
    \item Upon receiving $(\textsf{issue},\sid)$ from $\adv$ as the input of the issuer in $\funrevoke$, sample \alisa{$\rev_I \sample \tokens$}. Set $L = \leak{issue}{\set{I}}(I,U,\rev_I) = (I,U,\rev_I)$, and send $(\alisa{\textsf{issue-req}}, \sid, L)$ to $\adv$ on behalf of $\funrev$.
    
    \item On input $(\alisa{\textsf{issue-req-ok}}, \sid)$ from $\adv$, verify that there is no $U' \notin \corrupted$  s.t. $U' \in \user{\rev_I}$ (as $\rev_I$ is freshly sampled, it is unlikely for a sufficiently large $N$). Assign $\user{\rev_I} \gets \user{\rev_I} \cup \set{U}$, $\revoker{\rev_I} \gets \revoker{\rev_I} \cup \set{I}$.

    \item Update $\tokens \gets \tokens \setminus \set{\rev_I}$. 
    \item \alisa{Output $(\textsf{issue-resp-r}, \sid, \ok)$ to $\adv$ on behalf of $\funrevoke$.}
    \end{enumerate}

    \item \alisa{Generate $\cid \gets (id_\cid,\textsf{summary}(\uid,\aids,\iid))$}.
    
    \item Send $(\textsf{issue-ok},\sid,\alisanotes{\cid})$ to the $\ifun$.
    



    
    \item Upon receiving $(\textsf{input},\iid,\textsf{issue},\sid,\atrs')$ from $\ifun$ (input of a malicious issuer from $\mathcal{Z}$), \alisa{notify $\adv$ that $\iid$ received the input $\atrs'$}.
    
    \item Upon receiving $(\cred, \sigma_{\cred})$ from corrupted AMS, \alisanotes{check that $\cred$ indeed contains $\aids$.}

    \item Verify issuer's signature $\sigma_{\cred}$ on $\cred$ \alisa{using $\funs^{\iid}$ of $\iid$} as $b \gets \textsf{v}(\sigma_{\cred},\cred)$.\\
    Continue if $b=1$.


    \item Simulate decryption of $c_{\rev}$ \alisanotes{(contained in $\cred$)} by sending $(\textsf{decrypt-init},\sid)$ to $\adv$ on behalf of $\fundec^{\uid}$. Upon receiving $(\textsf{decrypt-complete},\sid,\dec)$, if there exists a pair $(m,c_{\rev})$ in $\tenc^\uid$:
    \begin{itemize}
        \item If $m$ is not unique, set $m' = \bot$,
        \item If $m$ is unique, set $m' = m$.
    \end{itemize}
    Otherwise, compute $m' := \dec(c_{\rev})$. Add $(m',c_{\rev})$ to $\tenc^\uid$.
    
        All cases are possible since credential has been generated by the adversary and ciphertexts may not have been queried though the encryption of $\fundec^\uid$.
        
    \item Proceed if $m' = \rev$ \alisa{for $\rev = (\rev_I,\rev_U)$}.

    \item \liza{Simulate decryption of $c_{0}$ (contained in $\cred$) by sending $(\textsf{decrypt-init},\sid)$ to $\adv$ on behalf of $\fundec^\uid$.} 
    
    \liza{Upon receiving $(\textsf{decrypt-complete},\sid,\dec)$, if there exists a pair $(m,c_{0})$ in $\tenc^\uid$:
    \begin{itemize}
        \item if $m$ is not unique, set $m' = \bot$,
        \item if $m$ is unique, set $m' = m$. 
    \end{itemize}
    Otherwise, compute $m' := \dec(c_{0})$. Add $(m',c_{0})$ to $\tenc^\uid$.}
    
    \item \liza{Proceed if $m' = \pk_{\uid}$.}

    \item Simulate decryption of $n=|\aids|$ ciphertexts from $\cred$ by sending $(\textsf{decrypt-init},\sid)$ to $\adv$ on behalf of $\fundec^\uid$.  Upon receiving $(\textsf{decrypt-complete},\sid,\dec)$, if there exists a pair $(\atr_i,c_i)$ in $\tenc^\uid$:
    \begin{itemize}
        \item If $\atr_i$ is not unique, set $m' = \bot$,
        \item If $\atr_i$ is unique, set $m' = \atr_i$.
    \end{itemize}
    Otherwise, compute $m' := \dec(c_i)$. Add $(m',c_i)$ to $\tenc^\uid$.
    
    All cases are possible since credential has been generated by the adversary and ciphertexts may not have been queried though the encryption of $\fundec^\uid$. 

     Let $\atr_i = m'$. Proceed if $\atr_i \in \mathcal{M}_{\aid_i}$.

    \item \alisanotes{Send $(\textsf{input},\iid,\textsf{issue},\sid,\atrs)$ to $\ifun$. This will be the input of malicious issuer in $\ifun$.}

    \item \alisa{Compute $h_\cred \gets H(h_\cred,(\cid,\cred,\sigma_\cred))$.}
    
    \item Wait for $(\textsf{issue-atr},\sid)$ from $\ifun$. \\ Send message $(\textsf{atr-ok},\sid)$ to $\ifun$ and save $(\cid, \cred, \sigma_{\cred}, \rev)$ to $\tiss$.\\ \Comment{Here we add to $\tiss$ the credential as viewed by the honest user.}
    
\end{enumerate}

\underline{Malicious AMS(Storage) and User}:
While AMS does provide additional protection for the issuer by authenticating the user, it still cannot prevent the user from malicious behaviour. Hence, we do not consider separately the case where only the user is corrupted, as the simulation would be essentially the same.
\begin{enumerate}

    \item \alisa{Simulate $I$ receiving $msg$ and $\sigma$ from $\adv$. Extract $\uid$, $\iid$, $\aids$ from $msg$. Verify $\sigma$ and $msg$ using $\funsign^\uid$.}
    
    \item \alisanotes{Send $(\textsf{input},\uid,\textsf{issue},\sid,(\aids, \alisa{\iid}))$ to $\ifun$.}

    \item \alisa{Receive $(\textsf{issue-req},\sid,\uid,\iid,\aids)$ from $\ifun$.}


   \item \alisa{Let $\cid \gets (id_\cid,\textsf{summary}(\uid,\aids,\iid))$. Verify freshness of $\cid$ (and hence $id_\cid$).}

    \item \alisa{Simulate $I$ receiving the revocation token $\rev_U$ from $\adv$.}
       
    \item \alisa{Generate a fresh token $\rev_I$ using $\funrevoke$:}
    \begin{enumerate}
    \item Upon receiving $(\textsf{issue},\sid)$ from $\adv$ as the input of the issuer in $\funrevoke$, sample \alisa{$\rev_I \sample \tokens$}. Set $L = \leak{issue}{\set{U}}(I,U,\rev_I) = (I,U,\rev_I)$, and send $(\alisa{\textsf{issue-req}}, \sid, L)$ to $\adv$ on behalf of $\funrev$.
    
    \item On input $(\alisa{\textsf{issue-req-ok}}, \sid, \mal{\rev_I})$ from $\adv$, take $\rev_I \gets \mal{\rev_I}$. Verify that there is no $U' \notin \corrupted$ s.t. $U' \in \user{\rev_I}$ (as $\rev_I$ is freshly sampled, it is unlikely for a sufficiently large $N$). Assign $\user{\rev_I} \gets \user{\rev_I} \cup \set{U}$, $\revoker{\rev_I} \gets \revoker{\rev_I} \cup \set{I}$.

    \item Update $\tokens \gets \tokens \setminus \set{\rev_I}$. 
    \item \alisa{Output $(\textsf{issue-resp-r}, \sid, \ok)$ to $\adv$ on behalf of $\funrevoke$.}
    \end{enumerate}
    
    \item Send $(\textsf{issue-ok},\sid,\alisanotes{\cid})$ to the $\ifun$ (start interaction with $\mathcal{Z}$).
    
    \item \alisanotes{Wait for $(\textsf{issue-atr},\sid,\atrs)$ from $\ifun$ ($I$ received attribute values from $\mathcal{Z}$).}
        
    
    \item Send encryption request $(\textsf{Encrypt}, \sid)$ to $\adv$ through $\fundec^\uid$. Wait for response $(\textsf{Encrypt}, \sid, c_{\rev})$ from $\adv$ and store the pair $(\rev,c_{\rev})$ in table $\tenc^\uid$. 

    \item \liza{Set $\atr_0 \gets \pk_{\uid}$} \Comment{with $\pk_{\uid}$ coming from verification $\ver$ defined during the initialisation phase}
    
    \item Simulate encryptions by honest issuer, sending \liza{$n=|\aids| + 1$} encryption requests $(\textsf{Encrypt}, \sid)$ to $\adv$ on behalf of $\fundec^\uid$. Wait for response $(\textsf{Encrypt}, \sid, c_i)$ from $\adv$ and store the pair $(\atr_i,c_i)$ in table $\tenc^\uid$. 

    
    \item Set $\cred = \langle \liza{(\aid_0,c_0)},(\aid_1,c_1), \dots, (\aid_n,c_n), (\aid_{\rev}, c_{\rev}) \rangle$. 
    
    \item Generate signature on the credential $\cred$ as $\sigma_{\cred} = \textsf{s}(\cred)$ and verify that $\textsf{v}(\sigma_{\cred},\cred) = 1$ \alisa{using $\funs^\iid$}. Save $(\sigma_{\cred},\cred)$ to table $\tisign^\iid$.
    
    \item Send $(\cred,\sigma_{\cred})$ to the user.    
    
    \item Send message $(\textsf{atr-ok},\sid)$ to $\ifun$ and save $(\cid, \cred, \sigma_{\cred}, \rev)$ to $\tiss$.\\ \Comment{Here we add to $\tiss$ the credential as viewed by the honest issuer.}
    
\end{enumerate}

\subsection{\alisanotes{Get all credential identifiers}}\label{app:fullproof:fetch}

\noindent\underline{Malicious AMS (storage):}


\begin{enumerate}
    \item On message $(\textsf{getcids-req},\sid,\alisa{\uid})$ from $\ifun$, \alisa{establish $(\alisa{\getcidsreq},\sid,\uid)$ connection with the server}. Set $msg=(\getcidsreq,\uid)$ and wait for input $(\textsf{sign-server}, \sid)$ from $\adv$ to $\funsign^{\uid}$. \lizanew{Send $(\textsf{sign-init}, \sid)$ to $\adv$ on behalf of $\funsign^{\uid}$.} Upon receiving \lizanew{$(\textsf{signature}, \sid, \sigma)$} from $\adv$, check whether $(msg, \sigma,0)$ is already stored. In this case, restart signing. \lizanew{Otherwise, store $(msg, \sigma, 1)$ in table $\tsign^\uid$ and send $(\textsf{sign-success},\sid)$ to the $\adv$}. Otherwise, stop execution.

    \item \lizanew{Send $msg=(\getcidsreq,\uid)$ and $\sigma$ to $\adv$.}

    \item \alisa{Wait for input $\cids,\creds$ from $\adv$ on behalf of corrupted server.} 

        \item \alisa{Let $\cids'$ be all the credential identifiers stored in $\tiss$ so far, and $\creds'$ the corresponding credentials. If $\cids = \cids$ and $\creds = \creds'$, send $(\textsf{getcids-ok},\sid)$ to $\ifun$. Note that the digest check by the user ensures that the check passes in the real protocol iff success is reported by $\ifun$ in the simulation.}

    \end{enumerate}

\noindent \underline{Malicious User}
\begin{enumerate}

\item Let $\req(\alisa{\textsf{getcids},\uid},\sid) = (msg,\sigma)$ result from a successful malicious user authentication as in Sec.~\ref{app:uauth}. Verify that $msg = (\getcidsreq, \uid)$.

\item Send $(\textsf{input},\uid,\textsf{getcids},\sid)$ to $\ifun$.

\item Wait for $(\textsf{getcids-req},\sid)$ from $\ifun$. \alisa{Let $\cids$ be all the credential identifiers stored in $\tiss$ so far, and $\creds$ the corresponding credentials. Send $\cids$ and $\creds$ to $\adv$ on behalf of AMS.} 

\item Send $(\textsf{getcids-ok},\sid)$ to $\ifun$. 

\end{enumerate}

\subsection{Revocation by issuer}\label{app:fullproof:revoke_issuer}
To simulate revocation by an honest revoker, we observe that $\ifun$ should leak $\cid$ if (1) it has been presented to some corrupted RP before, or (2) if it belongs to an honest user. Otherwise, 
\alisa{we take a random token $\rev \sample \tokens$, as in this case $\rev$ is not leaked to the adversary. The simulator still needs to simulate a successful revocation, even though the internal state of $\funrevoke$ will not change.}

\alisa{Later, upon full user corruption or presentation to a corrupted RP, it may turn out that a particular credential $\cid$ should have been revoked. 
The simulator will then silently add the corresponding token $\rev$ to the set $\revoked$.} 

\underline{Malicious AMS server (storage):} 
$\ifun$ does not leak $\cid$ to $\simm$ if both the user and the issuer are honest. In this case either $(\cid,\rev)$ has already been presented to some corrupted RP (and the simulator can reuse it), or $\adv$ has not seen $(\cid,\rev)$ pair before and we can run revocation for the randomly chosen revocation token.

\begin{enumerate}

    

    
    \item On message \alisanotes{$(\textsf{revoke-req},\sid,P=\iid,\alisa{\uid})$} from $\ifun$, \alisa{establish $(\alisa{\revokereq},\sid,\iid)$ connection with the server}. \alisa{Simulate encryption using $\fundec^{\uid}$ with a ciphertext $c_\cid$.} Generate a consent signature on behalf of $\iid$ using $\funs^\iid$.

    
    \item Come up with token $\rev$ to revoke as follows:
    \begin{itemize}
    
        \item If the actual $\cid$ has already been presented to some corrupted RP before, or it belongs to a corrupted user, then $\cid$ should have come from $\ifun$. \alisa{The token $\rev$ corresponds to entry from $\tiss$ with $\cid$.}

        \item Otherwise, \alisa{take random $\rev\sample\tokens$.} 
        \Comment{As $\rev$ is freshly sampled, collision with any issued token (in the past or in the future) is unlikely for a sufficiently large total number of tokens $N$}
    \end{itemize}


    \item On input $(\textsf{revoke},\sid)$ from $\adv$ to $\funrevoke$ on behalf of AMS:
    \begin{itemize}
        \item \alisa{if $\rev$ has been presented to a corrupted $RP$ before,} set $L = \leak{revoke}{\{\store, RP\}}(I,\alisa{\rev}) = \alisa{(I,b,\rev)}$;

         \item \alisa{if $\rev$ belongs to a corrupted user,} set $L = \leak{revoke}{\{\store, \user{\rev}\}}(I,\alisa{\rev}) = \alisa{(I,b,\rev)}$;
         
        \item otherwise, set $L = \leak{revoke}{\{\store\}}(I,\alisa{\rev}) = \alisa{(I,b)}$;
    \end{itemize}
    \alisa{where $b = (\rev \notin \revoked) \wedge (I \in \revoker{\rev}) = \true$, as in the real protocol an honest revoker only revokes a legitimately issued token that it has not revoked yet.}\\
    Send $(\alisa{\textsf{revoke-req}}, \sid, L)$ to $\adv$.

    \Comment{At this point, $\adv$ learns $\rev$ only if it corresponds to the true $\cid$ that is being revoked.}
    \item Wait for $(\alisa{\textsf{revoke-req-ok}}, \sid)$ from $\adv$.
    
    \item \alisa{Simulate successful check $I \in \revoker{\rev}$ without modifying the internal state of $\funrevoke$.} 
    
    \item Output $(\alisa{\textsf{revoke-resp-s}}, \sid, \ok)$ to $\adv$ on behalf of $\funrevoke$.

    \item Output 
    $(\textsf{revoke-ok},\sid)$ to $\ifun$.
    
    \item \alisa{If $\cid$ belongs to a corrupted user, there may be several $\cid'$-s that correspond to $\rev$. Look for all records $(\cid',\_,\_,\rev)$ in $\tiss$, and revoke all $\cid'$ on behalf $\uid$ (extracted from $\cid'$) using silent revocation interface.}
    
    \Comment{Since this case is only possible for $c_{\mathsf{U}} \geq 2$, all authentications of $\uid$ by $\ifun$ succeed.}
    
\end{enumerate}

\underline{Malicious Issuer:}
Since the malicious Issuer initiates revocation process, the adversary chooses the corresponding revocation token $\rev$. 
\begin{enumerate}


   \item \alisa{Wait for $\adv$ to establish a connection $(\alisa{\revokereq},\sid,\iid)$ with AMS on behalf of $I$.}

   \item Wait for a message \alisa{$msg = (\revokereq,\iid,c_\cid)$} and signature $\sigma_{\req}$ from $\adv$ on behalf of $I$. On behalf of AMS, verify the signature using $\funs^{\iid}$.
   \Comment{To ensure that this step succeeds only for $\iid \in \corrupted$ even in the case of replay attacks, we assume additional authentication of $\iid$.}

    \item On input $(\textsf{revoke},\sid,\alisa{x})$ from $\adv$ to $\funrevoke$ on behalf of $I$, set $L = \leak{revoke}{\{I\}}(I,\alisa{x}) = (I,b)$, \alisa{where $b = (x \notin \revoked) \wedge (I \in \revoker{x})$ is easy to simulate, as $\simm$ knows whether the corrupted $\iid$ is a revoker of $x$, and whether it has already revoked $x$.} 
    Send $(\alisa{\textsf{revoke-req}}, \sid, L)$ to $\adv$.
    
    \item Wait for $(\alisa{\textsf{revoke-req-ok}}, \sid)$ from $\adv$.

    \item If \alisa{$I \in \revoker{x}$ or $\exists R' \in \revoker{x} \cap \corrupted$}, add $x$ to $\revoked$.
    
    \item Output \alisa{$(\textsf{revoke-resp-r}, \sid, \ok)$} to $\adv$ on behalf of $\funrevoke$.

    \item Let $\cid$ be such that there is an entry $(\cid,\_,\_,(x,\_))$ or $(\cid,\_,\_,(\_,x))$ in $\tiss$. If there is none, set $\cid \gets \bot$.
    
    \item Send $(\textsf{input},\alisa{\iid},\textsf{revoke},\sid,\cid)$ to $\ifun$.

    \item \alisa{If $\cid = \bot$, then $\ifun$ outputs $(\textsf{revoke-resp},\sid,\bot)$ to $\iid$.}
    
    \alisa{Otherwise, a message $(\textsf{revoke-req},\sid,\ldots)$ comes from $\ifun$. Respond with $(\textsf{revoke-ok},\sid)$ to ensure that a record $(\uid,\cid,\iid)$ gets into $\mathsf{T}_r$.}

    \item \alisa{For a corrupted user, there may be several $\cid'$-s that correspond to $x$. Look for all records $(\cid',\_,\_,(x,\_))$ and $(\cid',\_,\_,(\_,x))$ in $\tiss$, and revoke all $\cid' \neq \cid$ on behalf $\uid$ (extracted from $\cid'$) using silent revocation interface.}
    
    \Comment{Since this case is only possible for $c_{\mathsf{U}} \geq 2$, all authentications of $\uid$ by $\ifun$ succeed.}
        
\end{enumerate}


\subsection{Revocation by user}


While the general idea behind token handling is the same, the simulation is more complicated than issuer revocation, \alisa{as it involves threshold decryption. However, it still does not leak decrypted $c_{\rev}$ (and hence $\cid$) to the adversary due to its blinding property.} 
Also, revocation requires a consent signature which needs to be generated using threshold signing $\funsign$.

\underline{Malicious AMS server (storage):}
\begin{enumerate}
    

    \item On message \alisanotes{$(\textsf{revoke-req},\sid,P=\uid)$} from $\ifun$, \alisa{establish $(\alisa{\revokereq},\sid,\uid)$ connection with the server}. \alisa{Simulate encryption using $\fundec^{\uid}$ with a ciphertext $c_\cid$.} 

    \item Set $msg = \alisa{(\revokereq,\uid,c_\cid)}$ and wait for input $(\textsf{sign-server}, \sid)$ from $\adv$. \lizanew{Send $(\textsf{sign-init}, \sid)$ to $\adv$ on behalf of $\funsign^{\uid}$.} Upon receiving \lizanew{$(\textsf{signature}, \sid, \sigma)$} from $\adv$, check whether $(msg, \sigma,0)$ is already stored. In this case, restart signing. \lizanew{Otherwise, store $(msg, \sigma, 1)$ in table $\tsign^\uid$ and send $(\textsf{sign-success},\sid)$ to the $\adv$}. Otherwise, stop execution.

    \item \lizanew{Send $msg = \alisa{(\revokereq,\uid,c_\cid)}$ and $\sigma$ to $\adv$.}

    \item Come up with token $\alisa{\rev}$ to revoke as follows:
    \begin{itemize}
    
        \item If the actual $\cid$ has already been presented to some corrupted RP before, then $\cid$ should have come from $\ifun$. \alisa{The token $\rev$ corresponds to entry from $\tiss$ with $\cid$.}

        \item Otherwise, \alisa{take random $\rev \sample \tokens$.} 
        \Comment{As $\rev$ is freshly sampled, collision with any issued token (in the past or in the future) is unlikely for a sufficiently large total number of tokens $N$}
    \end{itemize}

    \item Simulate decryption of $c_{\rev}$ by sending $(\textsf{decrypt-init},\sid)$ to $\adv$ on behalf of $\fundec^{\uid}$. 
    
    Upon receiving $(\textsf{decrypt-complete},\sid,\dec)$, if there exists a pair $(m,c_{\rev})$ in $\tenc^\uid$:
    \begin{itemize}
        \item if $m$ is not unique, set $m' = \bot$,
        \item if $m$ is unique, set $m' = m$. 
    \end{itemize}
    Otherwise, compute $m' := \dec(c_{\rev})$. Add $(m',c_{\rev})$ to $\tenc^\uid$.\\
    \Comment{As simulator has just computed $c_\rev$ itself, we have $m = m' = \rev$.}


    \item On input $(\textsf{revoke},\sid)$ from $\adv$ to $\funrevoke$ on behalf of AMS:
    \begin{itemize}
        \item \alisa{if $\rev$ has been presented to a corrupted $RP$ before,} set $L = \leak{revoke}{\{\store, RP\}}(U,\alisa{\rev}) = \alisa{(U,b,\rev)}$;
        
        \item otherwise, set $L = \leak{revoke}{\{\store\}}(U,\alisa{\rev}) = \alisa{(U,b)}$;
    \end{itemize}
    \alisa{where $b = (\rev \notin \revoked) \wedge (U \in \revoker{\rev}) = \true$, as an honest revoker only revokes a legitimately issued token that it has not revoked yet.}
    Send $(\alisa{\textsf{revoke-req}}, \sid, L)$ to $\adv$.

    \Comment{At this point, $\adv$ learns $\rev$ only if it corresponds to the true $\cid$ that is being revoked.}
    \item Wait for $(\alisa{\textsf{revoke-req-ok}}, \sid)$ from $\adv$.
    
    \item \alisa{Simulate successful check $U \in \revoker{\rev}$ without modifying the internal state of $\funrevoke$.} 
    
    \item Output $(\alisa{\textsf{revoke-resp-s}}, \sid, \ok)$ to $\adv$ on behalf of $\funrevoke$. 
    
    \item Output 
    $(\textsf{revoke-ok},\sid)$ to $\ifun$.

    \item \alisa{If $\cid$ belongs to a corrupted user, there may be several $\cid'$-s that correspond to $\rev$. Look for all records $(\cid',\_,\_,\rev)$ in $\tiss$, and revoke all $\cid'$ on behalf $\uid$ (extracted from $\cid'$) using silent revocation interface.}
    
    \Comment{Since this case is only possible for $c_{\mathsf{U}} \geq 2$, all authentications of $\uid$ by $\ifun$ succeed.}

\end{enumerate}

\underline{Malicious User}:

\begin{enumerate}

   \item Let $\req(\alisa{\textsf{revoke},\uid},\sid) = (msg,\sigma)$ result from a successful malicious user authentication as in Sec.~\ref{app:uauth}. \alisa{Verify that $msg=(\revokereq,\uid,\ldots)$.}



    


    \item Interact with the corrupted user in the threshold decryption protocol through $\fundec^{\uid}$ as follows:
    \begin{enumerate}    
    \item Wait for input $(\textsf{decrypt},\sid,c)$ from $\adv$ on behalf of corrupted $U$.
    \item Send $(\textsf{decrypt-init},\sid,c)$ to $\adv$.
    \item Upon receiving $(\textsf{decrypt-complete},\sid)$ from $\adv$, set $y_2 \gets \top$. Upon receiving $(\textsf{decrypt-fail},\sid)$ from adversary, set $y_2 \gets \bot$. \alisanotes{Here $y_2$ is the server output (success of fail).}
    \item Wait for $(\textsf{Output},1,y_1)$ from $\adv$ through $\fundec^{\uid}$ and send $(\textsf{Decrypted},\sid,y_1)$ to $\adv$. \alisanotes{Here $y_1$ is the user output chosen by $\adv$.}
    \end{enumerate}   


    
    

    
    \item On input $(\textsf{revoke},\sid,\alisa{x})$ from $\adv$ to $\funrevoke$ on behalf of $U$, set $L = \leak{revoke}{\{U\}}(U,\alisa{x}) = (U,b)$, \alisa{where $b = (x \notin \revoked) \wedge (I \in \revoker{x})$ is easy to simulate, as $\simm$ knows whether the corrupted user is a revoker of $x$, and whether it has already revoked $x$.} 
    Send $(\alisa{\textsf{revoke-req}}, \sid, L)$ to $\adv$.
    
    \item Wait for $(\alisa{\textsf{revoke-req-ok}}, \sid)$ from $\adv$.

    \item If \alisa{$U \in \revoker{x}$ or $\exists R' \in \revoker{x} \cap \corrupted$}, add $x$ to $\revoked$.
    
    \item Output \alisa{$(\textsf{revoke-resp-r}, \sid, \ok)$} to $\adv$ on behalf of $\funrevoke$.

    \item Let $\cid$ be such that there is an entry $(\cid,\_,\_,(x,\_))$ or  $(\cid,\_,\_,(\_,x))$ in $\tiss$. If there is none, set $\cid \gets \bot$.
    
    \item Send $(\textsf{input},\alisa{\uid},\textsf{revoke},\sid,\cid)$ to $\ifun$.

    \item \alisa{If $\cid' = \bot$, then $\ifun$ outputs $(\textsf{revoke-resp},\sid,\bot)$ to $\uid$.}
    
    \alisa{Otherwise, a message $(\textsf{revoke-req},\sid,\ldots)$ comes from $\ifun$. Respond with $(\textsf{revoke-ok},\sid)$ to ensure that a record $(\uid',\cid,\uid)$ gets into $\mathsf{T}_r$.}

    \item \alisa{For a corrupted user, there may be several $\cid'$-s that correspond to $x$. Look for all records $(\cid',\_,\_,(x,\_))$ and $(\cid',\_,\_,(\_,x))$ in $\tiss$, and revoke all $\cid' \neq \cid$ on behalf $\uid$ (extracted from $\cid'$) using silent revocation interface.}
    
    \Comment{Since this case is only possible for $c_{\mathsf{U}} \geq 2$, all authentications of $\uid$ by $\ifun$ succeed.}
    
\end{enumerate}

\subsection{Present}
\underline{Malicious AMS (storage) (and Issuer):}

Here, we are simulating communication that is expected by the adversary through the $\funsign^{\uid}, \fundec^{\uid}, \funrev$ functionalities with a randomly chosen credential, 
as this step does not leak any information about the $\cid$ to $\adv$.
 

\begin{enumerate}

    
    
    \item On message $(\textsf{present-req},\sid,l_{\atr}=|\aidsd|)$ from $\ifun$, \alisa{establish $(\alisa{\presentreq},\sid,\uid)$ connection with the server}. \alisa{Simulate encryption using $\fundec^{\uid}$ with a ciphertext $c_\cid$.} Set \alisa{$msg=(\presentreq,\uid,n,c_\cid)$} and wait for input $(\textsf{sign-server}, \sid)$ from $\adv$. \lizanew{Send $(\textsf{sign-init}, \sid)$ to $\adv$ on behalf of $\funsign$.} Upon receiving \lizanew{$(\textsf{signature}, \sid, \sigma)$} from $\adv$, check whether $(msg, \sigma,0)$ is already stored. In this case, restart signing. \lizanew{Otherwise, store $(msg, \sigma, 1)$ in table $\tsign^\uid$ and send $(\textsf{sign-success},\sid)$ to the $\adv$}. Otherwise, stop execution.

    \item \lizanew{Send $msg=(\presentreq,\uid,n,c_\cid)$ and $\sigma$ to $\adv$.}

    \item Simulate retrieval of credential data.
    \begin{itemize}

    \item  Select any entry $(\cid, \cred, \sigma_{\cred},\rev)$ from $\tiss$, where $\cid$ contains $\uid$.
    
    \Comment{Since presentation request has come from $\ifun$, a suitable credential must exist.}

    \item Parse $\cred = \langle \liza{(\aid_0, c_0)}, (\aid_1, c_1), \dots, (\aid_n,c_n), (\aid_{\rev}, c_{\rev}) \rangle$.


    \end{itemize}

    \item Simulate decryption of $c_{\rev}$ by sending $(\textsf{decrypt-init},\sid)$ to adversary on behalf of $\fundec^{\uid}$. Upon receiving $(\textsf{decrypt-complete},\sid,\dec)$, look for a pair $(\rev,c_{\rev})$ in $\tenc^\uid$ . If $\rev$ is not unique, set $\rev \gets \bot$.
    \Comment{$\rev \neq \bot$ always exists, since $c_{\rev}$ was taken from $\tiss$, where it has been previously computed using $\fundec^{\uid}$}
    
    \item Simulate revocation notification \alisa{for $x \in \set{\rev_I,\rev_U}$ where $(\rev_I,\rev_U) \gets \rev$}.
    \begin{itemize}
    \item Upon receiving $(\textsf{unotify},\sid)$ from the $\adv$ as input to $\funrev$. Set $L = \leak{unotify}{\{S,I\}}(U,\alisa{x}) = (U,\alisa{b'})$, \alisa{where $b' = (U \in \user{x}) = \true$, as the honest user will only make this query for $x$ issued to it.} 
    Send $(\alisa{\textsf{unotify-req}},\sid,L)$ to $\adv$. 
    
    \item Wait for $\alisa{(\textsf{unotify-req-ok},\sid)}$ from $\adv$. 
    
    \item Let $b = \true$ iff $x \notin \revoked$ and $U \in \user{x}$.
    
    \item Send $(\alisa{\textsf{unotify-resp-s}},\sid,\ok)$ to $\adv$ on behalf of $\funrevoke$.
    \end{itemize}

    
    

   

    \item Simulate decryption of $l_{\atr}$ attributes from $\cred$ by sending $(\textsf{decrypt-init},\sid)$ to adversary on behalf of $\fundec^{\uid}$. Upon receiving $(\textsf{decrypt-complete},\sid,\dec)$, look for a pair $(\atr_i,c_i)$ in $\tenc^\uid$ If $\atr_i$ is not unique, set $\atr_i = \bot$.
    
    \Comment{Here $\atr_i \neq \bot$ always exists since $\cred$ has been taken from $\tiss$, where it has been previously computed using $\fundec^{\uid}$.}
    
    \Comment{Since our decryption protocol is blind, adversary will not notice that we are decrypting some random attribute, but we do $l_{\atr}$ decryption queries as adversary expects to see this amount of queries.}
    
        

    \item Send $(\textsf{present-req},\sid,\ok)$ to $\ifun$.
\end{enumerate}

\underline{Malicious User}: 
In this case, we are simulating presentation of attributes requested by the malicious user. Malicious user is able to proceed with the presentation request only if the user is fully corrupted ($c_{\mathsf{U}} = 2$ \alisa{or $c_{\mathsf{U}} = 3$}).
\begin{enumerate}

    \item Let $\req(\alisa{\textsf{present},\uid}, \sid) = (msg,\sigma)$ result from a successful malicious user authentication as in Sec.~\ref{app:uauth}. Verify that $msg = (\presentreq,\uid,n,\ldots)$.
    
    \item Let $n$ be the reported length of $\aidsd$ which is a part of $msg$ that the server received. Interact with the corrupted user in \liza{$n+2$} decryptions using the threshold decryption protocol through $\fundec^{\uid}$ as follows:
    \begin{enumerate}    
    \item Wait for input $(\textsf{decrypt},\sid,c_i)$ from $\adv$ on behalf of corrupted $U$.
    \item Send $(\textsf{decrypt-init},\sid,c_i)$ to $\adv$.
    \item Upon receiving $(\textsf{decrypt-complete},\sid)$ from $\adv$, set $y_2 \gets \top$. Upon receiving $(\textsf{decrypt-fail},\sid)$ from adversary, set $y_2 \gets \bot$. \alisanotes{Here $y_2$ is the server output (success of fail).}
    \item Wait for $(\textsf{Output},1,y_1)$ from $\adv$ through $\fundec^{\uid}$ and send $(\textsf{Decrypted},\sid,y_1)$ to $\adv$. \alisanotes{Here $y_1$ is the user output chosen by $\adv$.}
    \end{enumerate}

    \item Interact with the corrupted user in two notifications as follows:
    \begin{itemize}
    \item Upon receiving $(\textsf{unotify},\sid,\alisa{x})$ from $\adv$ to $\funrevoke$ on behalf of $U$, set $L = \leak{unotify}{\set{U}}(U,\alisa{x}) = (U,\alisa{b'})$ \alisa{where $b' = \true$ iff $U \in \user{x}$}. Send $\alisa{(\textsf{unotify-req},\sid,L)}$ to the $\adv$. 
    
    \item Wait for \alisa{$(\textsf{unotify-req-ok},\sid)$} from $\adv$.
    
    \item Let $b = \true$ iff $x \notin \revoked$ and $U \in \user{x}$.
    
    \item Output $\alisa{(\textsf{unotify-resp-u}, \sid, b)}$ to $\adv$ on behalf of $\funrevoke$.
    \end{itemize}

    \item \alisanotes{Let the adversary choose the output. 
    Based on adversary data, we in general cannot come up with the "actual" inputs $\cid$ and $\aidsd$ for a valid presentation. 
    Hence, send $(\textsf{inputs},\uid,\textsf{present},\sid,\bot)$ to $\ifun$, causing it to fail, and choose the output $\pres$ for corrupted $U$. The outcome of $\ifun$ does not matter, as it has no impact neither on the inner state of $\ifun$ nor $\mathcal{Z}$.}

\end{enumerate}

\subsection{Verification}
Revocation functionality $\funrevoke$ assumes that the verifier party $V$ has called it at least once using the subroutine $\textsf{vnotify}$ of $\funrevoke$. It needs to remember the set $\revoked_V$ of the tokens revoked at that point. Similarly to $\funrevoke$, the simulator defines $\revoked_V$ as the current state of $\revoked$ that it maintains locally.

\underline{Malicious AMS:}
\lizanew{Since the RP is honest, the verification is fully simulated by the simulator. Therefore, we can simulate generating signature for a random challenge and random credential. Due to blindness property of the signature scheme, AMS server will not distinguish the executions.}
\begin{enumerate}

    \item \lizanew{Upon receiving $(\textsf{verify-req},\sid)$ from $\ifun$, \alisa{sample $\ch \gets \mathcal{CH}$ and} send $(\textsf{verify-req-ok},\sid,\ch)$ to $\ifun$  to confirm proceeding.}

    \item \lizanew{Take any entry from $\tiss$ with $\cid$ \alisa{for which a successful presentation has been constructed before}. Take corresponding $\cred$, $\sigma_{\cred}$, $\rev$, $\alisa{\iid}$, \alisa{and construct $\pres=(\aids,\atrs,\iid)$, which $\adv$ will not see anyway (including the number of disclosed attributes).}}
    


    \Comment{If we have reached this point, such a credential does exist.}


   \item \alisa{Establish $(\alisa{\verifyreq},\sid,\uid)$ connection with the server on behalf of user}.
   \item \lizanew{Set $\alisa{msg=(\pres,\ch',\cred,\sigma_\cred)}$ and wait for input $(\textsf{sign-server}, \sid)$ from adversary. Send $(\textsf{sign-init}, \sid)$ to adversary on behalf of $\funsign$. Upon receiving $(\textsf{signature}, \sid, \sigma_{\pres})$ from adversary, check whether $(msg, \sigma_{\pres}, 0)$ is already stored. In this case, restart signing. Otherwise, store $(msg, \sigma_{\pres}, 1)$ in table $\tsign^\uid$.}
   
    \item Simulate sending \alisa{$\langle \uid, \cred, \sigma_{{\cred}}, \pres, \rev_I, \rev_U, \sigma_{\pres} \rangle$} from the user to RP.

    \item \lizanew{Simulate verifying presentation by RP}.

    \item Send $(\textsf{verify-req-ok}, \sid)$ to $\ifun$.
\end{enumerate}

\underline{Malicious RP (and AMS):} 

\begin{enumerate}
   

    \item On message $(\textsf{verify-req}, \sid, \alisa{\uid}, \cid,\pres,\aidsd)$ from $\ifun$, parse \alisa{$(\aidsd,\atrsd,\iid) \gets \pres$}. Find entry from $\tiss$ with $\cid$ and take corresponding $\cred$, $\sigma_{\cred}$, $\rev$.

    \item \alisa{Establish a connection $(\alisa{\verifyreq}, \sid,\uid)$ with RP.}

    \item \liza{Receive $\ch$ that $\adv$ sends on behalf of RP}.

    \item For all $\atr_i$ in $\atrsd$, update the entry $(\atr_i,c_i)$ in $\tenc^\uid$, where $c_i$ are taken from $\cred$. \Comment{This is needed since $c_i$ have been encryptions of potentially garbage values $\atr'_i$ during the issuing phase. It is safe to update $\atr'_i$ that have not been used in any simulations yet, and if $\atr'_i$ has already been simulated to $\adv$, then $\atr_i = \atr'_i$, as the only place where $\atr_i$ are opened are the verification routine (here) or the adaptive user corruption (described in App.~\ref{app:fullproof:corrupt}), where $\atr_i$ are taken from $\mathsf{T}_a$ of $\ifun$.}


   \item \alisa{Establish $(\alisa{\verifyreq},\sid,\uid)$ connection with the server on behalf of user}.
   
   \item \lizanew{Set \alisa{$msg=(\pres,\ch,\cred,\sigma_\cred)$} and construct presentation signature:}
   \begin{itemize}
       \item if storage is corrupted, wait for input $(\textsf{sign-server}, \sid)$ from adversary and proceed. 
       \item if storage is not corrupted, proceed.
   \end{itemize}

   Send $(\textsf{sign-init}, \sid)$ to adversary on behalf of $\funsign$. Upon receiving $(\textsf{signature}, \sid, \sigma_{\pres})$ from adversary, check whether $(msg, \sigma_{\pres}, 0)$ is already stored. In this case, restart signing. Otherwise, store $(msg, \sigma_{\pres}, 1)$ in table $\tsign^\uid$.
   
   \item Simulate sending $\alisa{\langle \uid, \cred, \sigma_{{\cred}},\pres,\rev,\sigma_\pres \rangle}$ to corrupted RP.

    \item \alisa{Parse $(\rev_I,\rev_U) \gets \rev$.} 
    
    \item\label{step:verify:rp:begin} On message $(\textsf{verify},\sid)$ from $\adv$ to $\funrevoke$ on behalf of RP, \alisa{define leakage $L$ as follows:}
    \begin{itemize}
    \item if storage is not corrupted, take \alisa{$L = \leak{verify}{\set{RP}}(U,RP,x) = (x,b')$};
    \item if storage is corrupted, take \alisa{$L = \leak{verify}{\set{\store,RP}}(U,RP,x) = (x,b',\revsids{x})$};
    \end{itemize}
    \alisa{where $x = \rev_I$ and $b' \in \revokedu$ is true iff $U$ has received a successful notification for $x$ before $RP$.} \alisa{Since an honest $U$ does not proceed to verification without a successful prior notification, $b' = \true$.} \alisa{The set $\revsids{x}$ should contain all session identifiers for which $x$ has been revoked so far, which corresponds to $\sids$ received from $\ifun$ for a corrupted $\store$. For all $\sid' \in \sids$ that ended up successfully, the simulator also adds $\rev_P$ to $\revoked$, where $P$ is the revoker party of the session $\sid'$.}

    Send $(\alisa{\textsf{verify-req}}, \sid, L)$ on behalf of $\funrevoke$ to $\adv$.
    
    \item Wait for input $(\alisa{\textsf{verify-req-ok}}, \sid)$ to $\funrevoke$ from $\adv$.
    
    
    \item \alisa{Assert $x \notin \revoked_{RP}$.} 
    
    \item\label{step:verify:rp:end} Output \alisa{$(\textsf{verify-resp-v}, \sid, x)$}  
    on behalf of $\funrevoke$ to $\adv$.
    
    \item \alisa{Repeat steps (\ref{step:verify:rp:begin})-(\ref{step:verify:rp:end}) with $x = \rev_U$.}

    \item On message $(\textsf{Verify},\sid,(c_{\rev},\rev))$ from $\adv$ to $\fundec^\uid$ on behalf of $RP$, if there exists a pair $(\rev,c_{\rev})$ in $\tenc^\uid$ and $\rev$ is unique, output $(\textsf{Verify-success},\sid)$ to $\adv$. Otherwise, output $(\textsf{Verify-fail},\sid)$ to $\adv$.

    \item For $i = \liza{0},\ldots,|\atrs|$, on message $(\textsf{Verify},\sid,(c_i,\atr_i))$ from $\adv$ to $\fundec^\uid$ on behalf of RP, if there exists a pair $(\atr_i, c_i)$ in $\tenc^\uid$ and $\atr_i$ is unique, output $(\textsf{Verify-success},\sid)$ to $\adv$. Otherwise, output $(\textsf{Verify-fail},\sid)$ to $\adv$.

    \item Send $(\textsf{verify-req-ok}, \sid)$ to $\ifun$.
\end{enumerate}

\underline{Malicious User:} It is important that revocations of honest revokers are covered by $\funrevoke$. It is possible that, for some $\rev$ included into $\funrevoke$ (revocations by a corrupted revoker), the corresponding credential $\cid$ is \emph{not} included into $\mathsf{T}_r$ of $\ifun$. In this case, instead of letting $\ifun$ rely on the check $(\_,\cid,\_) \notin \mathsf{T}_r$, the simulator will explicitly make $\ifun$ fail using the adversary's interface. 
\begin{enumerate}


   \item \alisa{Wait for $\adv$ to establish a connection $(\alisa{\verifyreq},\sid,\uid)$ with RP on behalf of $U$.}
   

    \item \alisa{On behalf of the RP, generate a random challenge $\ch \sample \mathcal{CH}$.} \liza{Send $\ch$ to $\adv$.} 

   \item \alisa{\alisa{Wait for $\adv$ to establish $(\alisa{\verifyreq},\sid,\uid)$ connection with the AMS. On behalf of AMS, participate in generating a blind signature using $\funsign^\uid$.}}
   
    \item Let $\adv$ choose the message $\pres'$ sent to RP on behalf of $U$. 
    
    \liza{Parse $\langle \uid, \cred, \sigma_{{\cred}},\pres=(\aidsd,\atrsd,\iid),\rev,\sigma_\pres \rangle \gets \pres'$.}

    \item \alisa{Look for an entry $(\cid,\cred,\sigma_\cred,\rev)$ in $\tiss$. If there is no such entry:}
    \begin{itemize}
    \item If $\iid \in \corrupted$, use silent issuing interface to generate an entry with a suitable data $\cred$, provided with a fresh identifier $\cid$. This creates a record $\cred \in \mathsf{T}_a$ in $\ifun$, with $(\_,\cid,\_) \notin \mathsf{T}_r$. If $\rev \in \revoked$, use silent revocation interface of $\ifun$ to add $(\mathsf{U},\cid,\mathsf{U})$ to $\mathsf{T}_r$.

    \Comment{In this case, the adversary may still cause verification to fail due to presentation signature check or failures by $\funrevoke$, but such a failure will be handled separately.}
    
    \item Otherwise, take $\cid = \bot$.
    
    \Comment{In this case, then the verification will fail in both $\ifun$ and the real protocol. This is because, if the issuer is honest, the adversary is only able to come up with valid issuer signatures that are stored in $\tiss$.}
    \end{itemize}
    
    \item \alisanotes{Send $(\textsf{input}, \uid, \textsf{verify-req},\sid,  \pres, \alisa{\textsf{RP}})$ to $\ifun$  to confirm proceeding.}
    
    \item Wait for $(\textsf{verify-req}, \sid)$ from $\ifun$.  

    \item \alisanotes{On behalf of RP, perform checks that $\aidsd$ are contained in $\cred$, and $|\aidsd| = |\atrs|$. Verify the signatures $\sigma_{\cred}$ and \liza{$\sigma_{\pres}$} using $\funs^\iid$ and $\funsign^\uid$.}
    
    
    \item\label{step:verify:u:begin}On message $(\textsf{verify},\sid,\alisa{x})$ from $\adv$ to $\funrevoke$ on behalf of $U$, set 
    \alisa{$L = \leak{verify}{\set{U}}(U,RP,x) = b'$, where $b' \in \revokedu$ is true iff $U$ has received a successful notification before $RP$ (as $U$ is corrupted, both events are known to the simulator)}. Send $(\alisa{\textsf{verify-req}}, \sid, L)$ on behalf of $\funrevoke$ to $\adv$.
    
    \item Wait for input \alisa{$(\textsf{verify-req-ok}, \sid)$} from $\adv$. 

    \item \alisa{Assert $x \notin \revoked_{RP}$} 
    
    \item\label{step:verify:u:end} Output \alisa{$(\textsf{verify-resp-u}, \sid, \ok)$} 
    on behalf of $\funrevoke$ to $\adv$.

    \item \alisa{Parse $(\rev_I,\rev_U) = \rev$. On behalf of RP, check $x = \rev_I$. Repeat steps (\ref{step:verify:u:begin})-(\ref{step:verify:u:end}) and check $x = \rev_U$.}

     \item On message $(\textsf{Verify},\sid,(c_{\rev},\rev))$ from $\adv$ to $\fundec^{\uid}$ on behalf of $U$, if there exists a pair $(\rev,c_{\rev})$ in $\tenc^\uid$ and $\rev$ is unique, output $(\textsf{Verify-success},\sid)$ to RP. Otherwise, output $(\textsf{Verify-fail},\sid)$ to RP.

    \item For $i = \liza{0},\ldots,|\atrs|$, on message $(\textsf{Verify},\sid,(c_i,\atr_i))$ from $\adv$ to $\fundec^{\uid}$ on behalf of $U$, if there exists a pair $(\atr_i,c_i)$ in $\tenc^\uid$ and $\atr_i$ is unique, output $(\textsf{Verify-success},\sid)$ to RP. Otherwise, output $(\textsf{Verify-fail},\sid)$ to RP.    

    \item If all checks by RP succeed, send $(\textsf{verify-req-ok}, \sid)$ to $\ifun$.
    
\end{enumerate}

\subsection{Corruption}\label{app:fullproof:corrupt}
\alisa{In the real protocol, the adversary is corrupting \emph{parties}. Whenever a party receives a special corruption message, it is considered corrupted in every subfunctionality ($\fundec^{\uid}$, $\funs^{\iid}$, $\funrevoke$) used by the protocol. The simulator needs to propagate this corruption to $\ifun$.}On input $(\textsf{corrupt})$ from $\adv$ to a party $P$, send $(\textsf{corrupt-party},P)$ to $\ifun$.

\alisa{For the functionality $\funsign$, the corruptions are handled differently, as the adversary may explicitly send messages $(\textsf{corrupt-server})$ and $(\textsf{corrupt-user},\ell)$ to $\funsign$. For all $\uid$, the simulator constantly observers the internal state of $\funsign^\uid$, acting as follows:}
\begin{itemize}
\item \alisa{Whenever $c_{\mathsf{S}}=1$ is set in $\funsign$, send $(\textsf{corrupt-party},\store)$ to $\ifun$.}
\item \lizanew{If setup has not been performed, \alisa{only $c_{\mathsf{U}} = 3$ can be set in $\funsign$} (full corruption). This implies that adversary has full control over the user's device and there are no privacy guarantees for user any more.} \alisa{Send $(\textsf{corrupt-party},\uid)$ to $\ifun$, so that it will set $c_{\mathsf{U}} = 3$ as well.}
\item \lizanew{If the setup has been completed} 
\alisa{and $c_{\mathsf{U}} = 1$ is set in $\funsign$, send $(\textsf{corrupt-party},\uid)$ to $\ifun$, so that it will increase $c_{\mathsf{U}} = 1$ as well.}
This happens when adversary has managed to learn phone’s memory sometime between threshold signing sessions and has copied encrypted memory.
\item \alisa{If the setup has been completed}, 
\alisa{and $c_{\mathsf{U}}=2$ is set in $\funsign$, send $(\textsf{corrupt-party},\uid)$ to $\ifun$, so that it will increase $c_{\mathsf{U}} = 2$ as well.} This happens when $\adv$ guesses PIN during the authentication, having managed to gain access to the secret key storage. \lizanew{In practice, this means that adversary can be able to initiate multiple interactions with AMS server (by creating correct authentication signature) before the legitimate user tries to initiate signing transaction. At the legitimate user's signing attempt, clone detection mechanism will trigger and the user will be blocked, meaning that neither the legitimate user nor the adversary will be able to initiate any further request.} 

\end{itemize}

Whenever corruption level $c_{\mathsf{U}} = 2$ is reached, $\ifun$ sends a special message $(\textsf{full-leakage},\sid,\uid,(\mathsf{T}^{\uid}_a, \mathsf{T}^{\uid}_r))$, which contains projections of tables to the user $\uid$. \alisa{As discussed above, the practical leakage will be limited due to clone detection mechanism, but for safety we are assuming the worst-case here, as we do not know what exactly the adversary will manage to decrypt before being detected.} Update the internal tables as follows:
\begin{itemize}
\item For all $(\uid,\cid,\aids,\atr) \in \mathsf{T}_a$, find an entry in $\tiss$ with the corresponding $\cid$ and $\cred = \langle (\aid_0, c_0),  (\aid_1, c_1), \dots, (\aid_n,c_n), (\aid_{\rev}, c_{\rev}) \rangle$ . For every pair $(\aid_i,c_i)$, update the entry $(\atr_i,c_i)$ in $\tenc^\uid$ for the given $c_i$. If some $(\atr_i,c_i)$ has already been opened to $\adv$ before (during presentation), then it is the same $\atr_i$, taken from $\mathsf{T}_a$.

\item For all $(\uid,\cid,P) \in \mathsf{T}_r$, find an entry  in $\tiss$ with the $\cid$ and the corresponding revocation token $\rev \alisa{=(\rev_I,\rev_U)}$. Add $\rev_P$ to the set $\revoked$ of the simulated $\funrevoke$ (it is possible that $\rev_P$ is already there). This ensures that $\cid$ will be treated as revoked during subsequent presentations.




\end{itemize}

\subsection{Indistinguishability of the real and the simulated protocol.}

\paragraph{Issuing}
In the case of honest issuer, simulator indistinguishably simulates issuance of credential 

\[ \cred = \langle (\aid_0,c_0), (\aid_1,c_1), \dots, (\aid_n,c_n), (\aid_{\rev},c_{\rev}) \rangle \] by \alisa{including the public key $\pk_{\uid}$ of the user $\alisa{\uid}$ into $c_0$}, taking attribute values $\atr_i$ from the pre-defined set $\mathcal{M}_{\aid_i}$ (that is the same as in the real protocol) \alisa{for $i = 1..n$}, sampling revocation token from $\tokens$, assembling information together and creating signature. Indistinguishability of the signing and encryption processes is inherited from the underlying used constructions. Number of required attributes $n$ and the attribute identifiers $\aids$ are received from the $\ifun$. \alisanotes{The values  $\rev$ are allowed to repeat for a corrupted user.}

In the case of malicious issuer, adversary itself generates credential for the user. Before saving issued credential to the internal tables, simulator verifies that credential contains correct revocation token $\rev \in \tokens$ and correct attribute values $\atr_i \in \mathcal{M}_{\aid_i}$. Since $\rev$ is freshly sampled, we have that $\rev \notin \revoked$ and $\cid$ is unique with overwhelming probability.

In the case of malicious user, the issuer generates the credential for the user identity $\uid$ that is eligible to receive the credential, even if the request is sent on behalf of a malicious party. If $\uid$ is honest, the issuer will only generate a credential if $\uid$ has indeed requested it, which is proven by the fresh signature of $\uid$. If AMS is corrupted, it is possible that the user $\uid$ will not receive the credential, but the adversary does not learn the attribute values, since they are encrypted by the public key of $\uid$.

The simulator saves the credential that passes all checks to the internal table $\tiss$. A credential with $\cid$ containing $\uid$ and $\iid$ is only added to $\tiss$ if the following conditions are satisfied:
\begin{itemize}
\item The credential has indeed been requested by the user $\uid$ (if $\uid$ is honest).
\item The credential has indeed been issued by the issuer $\iid$ (if $\iid$ is honest).
\end{itemize}
Here unfairness w.r.t. user is allowed, i.e. the user will not necessarily receive the credenital due to network delay.

\paragraph{Presentation}

If the user is honest, presentation process is all simulated and adversary does not get to see credential presentation $\pres$. This allows simulator to simulate sub-protocols of the presentation phase with randomly selected credential. 
Usage of the blind threshold decryption ensures that the adversary does not learn which attribute is being decrypted. Thus, in this case simulated protocol is indistinguishable from the real protocol.


In the case of malicious user, adversary is able to initiate presentation process only if it has guessed $\pin$ from the underlying threshold signature scheme. If the $\pin$ check is successful, simulator simulates all the answers to the sub-protocol queries using its internal tables for the decryption and revocation using ciphertexts $c_{\atr_i}$ and $c_{\rev}$ supplied by adversary \alisa{(note that adversary may decrypt \emph{any} ciphertexts of $\uid$ at this point, and this is why there are no privacy guarantees for the user when achieving $c_{\mathsf{U}} = 2$).} The adversary expects to see the correct notification from $\funrevoke$. This is indeed so, as the simulated $\revoked$ in this case is consistent with the actually revoked tokens (either by the user or some issuer), both honest and corrupted. In the end of presentation, the output is produced by the simulator, and there is no effect on $\ifun$.


\paragraph{Verification}  
If the RP is malicious, it sees the presentation that is output of the ideal functionality, correctly constructed from the attributes requested by the user. 
We need to additionally ensure that the simulated revocation token is consistent with the previous revocations. During the previous protocol phases, whenever credential $\cid$ of an honest user has been revoked, the following has happened:
\begin{itemize}
\item In the cases when $\ifun$ leaked $\cid$ to the simulator, the corresponding $\rev$ was added to $\revoked$ 

\item Otherwise, adding $\rev$ to $\revoked$ was postponed (a random $\rev$ has been revoked, which coincides with at least one real token with probability \alisa{$\frac{n}{N}$ for $n$ issued credentials}). During the verification, the simulator has added all relevant tokens $\rev$ to $\revoked$ as soon as it received the corresponding information from $\ifun$, before simulating verification using $\funrevoke$.
\end{itemize}

If user is malicious and RP is honest, the simulated challenge $\ch$ is unique with probability $m/|\mathcal{CH}|$ for $m$ verification sessions. The simulator performs checks of credential presentation as an honest verifier. The simulated protocol accepts the credential only if
\begin{itemize}
    \item credential signature $\sigma_{\cred}$ verifies,
    \item revocation token is not among the revoked tokens: $\rev \notin \revoked$,
    \item disclosed attributes $\atr_i$ correspond to ciphertexts from credential $\cred$.
    \item \alisa{the user could respond to the fresh challenge, proving ownership of $\pk_{\uid}$}
\end{itemize}
As malicious verifications are only allowed in the case $c_{\mathsf{U}}=2$ and \alisa{$c_{\mathsf{U}}=3$}, it holds that either (1) the user has been fully corrupted from the beginning or (2) the attacker guessed PIN of the user. In the case (1), the simulator has all the data of this user included into internal tables, so all previous checks can be done against them. In the case (2), some data of this user stored in the simulator tables (up to the first corruption) is garbage. However, the true data is provided by $\ifun$ upon corruption, and the simulator fills its local tables accordingly, ensuring consistency of simulated tables with the tables of $\ifun$.

\paragraph{Revocation by issuer}
In case of malicious issuer, adversary knows the connection between the credential identifier and the revocation token and expects correct credential to be revoked. Simulator uses its internal table $\tiss$ to ensure that correct credential is marked as being revoked.

If AMS is corrupted, and $\cid$ has already been presented to a corrupted RP, or it belongs to a corrupted user, then $\rev_I$ should be leaked to the adversary by $\funrevoke$. In both cases, $\cid$ is leaked by $\ifun$, and the simulator takes the corresponding $\rev_I$ from $\tiss$.

If AMS is corrupted, but the credential has not been presented to a corrupted RP yet, and it belongs to an honest user, then the simulator does not receive $\cid$ from $\ifun$. However, as $\funrevoke$ does not leak $\rev_I$ neither, using random $\rev_I$ is indistinguishable from using the actual revocation token $\rev_I$ from $\tiss$ \alisa{(unless $\rev_I$ coincides by chance with some legitimate token, which is $\frac{n}{N}$ for $n$ issued credentials)}.


\paragraph{User revocation}
If the user is corrupted, adversary is able to initiate revocation process only if it has guessed $\pin$ from the underlying threshold signature scheme. If the $\pin$ check is successful, then either the user has been corrupted from the beginning \alisa{$(c_{\mathsf{U}} = 3)$}, or we have $c_{\mathsf{U}} = 2$, and in both cases the simulator has filled its internal tables according to $\mathsf{T}_a$ and $\mathsf{T}_r$ of $\ifun$, so it has all necessary data for the simulation. The simulator simulates all the answers to the sub-protocol queries using its internal tables for the decryption and revocation. 

If AMS is corrupted, and $\cid$ has already been presented to a corrupted RP, then $\rev_U$ should be leaked to the adversary by $\funrevoke$. In this case, $\cid$ is leaked by $\ifun$, and the simulator takes the corresponding $\rev_U$ from $\tiss$.

If AMS is corrupted, but the credential has not been presented to a corrupted RP yet, then the simulator does not receive $\cid$ from $\ifun$. However, as $\funrevoke$ does not leak $\rev_U$ neither, and decryption by means of $\fundec^{\uid}$ is blind, using a random $\rev_U$ is indistinguishable from using the actual revocation token $\rev_U$ from $\tiss$ \alisa{(unless $\rev_U$ coincides by chance with some legitimate token, which is $\frac{n}{N}$ for $n$ issued credentials)}.

%% file: ref2.bib
@article{covert,
  title={Security against covert adversaries: Efficient protocols for realistic adversaries},
  author={Aumann, Yonatan and Lindell, Yehuda},
  journal={Journal of Cryptology},
  volume={23},
  number={2},
  pages={281--343},
  year={2010},
  publisher={Springer}
}

@misc{funrevoke,
      author = {Alisa Pankova and Jelizaveta Vakarjuk},
      title = {Credential Revocation Assisted by a Covertly Corrupted Server},
      howpublished = {Cryptology {ePrint} Archive, Paper 2025/1854},
      year = {2025},
      url = {https://eprint.iacr.org/2025/1854}
}

@inproceedings{cacs,
  title={CACS: A Cloud Privacy-Preserving Attribute Management System},
  author={Kalu, Aivo and Kus, Burak Can and Laud, Peeter and Leung, Kin Long and Snetkov, Nikita and Vakarjuk, Jelizaveta},
  booktitle={Proceedings of the 18th International Conference on Availability, Reliability and Security},
  pages={1--9},
  year={2023}
}

@inproceedings{reqs,
  title={Cryptographic requirements of verifiable credentials for digital identification documents},
  author={Richter, Maximilian and Bertram, Magdalena and Seidensticker, Jasper and Margraf, Marian},
  booktitle={2023 IEEE 47th Annual Computers, Software, and Applications Conference (COMPSAC)},
  pages={1663--1668},
  year={2023},
  organization={IEEE}
}

@article{frigo2024anonymous,
  title={Anonymous credentials from ECDSA},
  author={Frigo, Matteo and others},
  journal={Cryptology ePrint Archive},
  year={2024}
}

@article{ramic2024selective,
  title={Selective disclosure in digital credentials: A review},
  author={Rami{\'c}, {\v{S}}eila Be{\'c}irovi{\'c} and Cogo, Ehlimana and Prazina, Irfan and Cogo, Emir and Turkanovi{\'c}, Muhamed and Mulahasanovi{\'c}, Razija Tur{\v{c}}inhod{\v{z}}i{\'c} and Mrdovi{\'c}, Sa{\v{s}}a},
  journal={ICT Express},
  year={2024},
  publisher={Elsevier}
}

@inproceedings{clonewars,
  title={How to win the clonewars: efficient periodic n-times anonymous authentication},
  author={Camenisch, Jan and Hohenberger, Susan and Kohlweiss, Markulf and Lysyanskaya, Anna and Meyerovich, Mira},
  booktitle={Proceedings of the 13th ACM conference on Computer and communications security},
  pages={201--210},
  year={2006}
}

@inproceedings{splitkey,
  title={Server-supported RSA signatures for mobile devices},
  author={Buldas, Ahto and Kalu, Aivo and Laud, Peeter and Oruaas, Mart},
  booktitle={Computer Security--ESORICS 2017: 22nd European Symposium on Research in Computer Security, Oslo, Norway, September 11-15, 2017, Proceedings, Part I 22},
  pages={315--333},
  year={2017},
  organization={Springer}
}

@String{Computing = "Computing" }

@String{Computer = "{IEEE} Computer" }

@String{Springer = "Springer-Verlag" }

@ArtifactSoftware{R,
    title = {R: A Language and Environment for Statistical Computing},
    author = {{R Core Team}},
    organization = {R Foundation for Statistical Computing},
    address = {Vienna, Austria},
    year = {2019},
    url = {https://www.R-project.org/},
}

@misc{sd-jwt,
    series =    {Request for Comments},
    number =    9901,
    howpublished =  {RFC 9901},
    publisher = {RFC Editor},
    doi =       {10.17487/RFC9901},
    url =       {https://www.rfc-editor.org/info/rfc9901},
    author =    {Daniel Fett and Kristina Yasuda and Brian Campbell},
    title =     {{Selective Disclosure for JSON Web Tokens}},
    pagetotal = 88,
    year =      2025,
    month =     nov,
}

@inproceedings{idemix,
  author       = {Jan Camenisch and
                  Anna Lysyanskaya},
  editor       = {Birgit Pfitzmann},
  title        = {An Efficient System for Non-transferable Anonymous Credentials with
                  Optional Anonymity Revocation},
  booktitle    = {Advances in Cryptology - {EUROCRYPT} 2001, International Conference
                  on the Theory and Application of Cryptographic Techniques, Innsbruck,
                  Austria, May 6-10, 2001, Proceeding},
  series       = {Lecture Notes in Computer Science},
  volume       = {2045},
  pages        = {93--118},
  publisher    = {Springer},
  year         = {2001},
  url          = {https://doi.org/10.1007/3-540-44987-6\_7},
  bibsource    = {dblp computer science bibliography, https://dblp.org}
}

@misc{u-prove,
author = {Paquin, Christian and Zaverucha, Greg},
title = {U-Prove Cryptographic Specification V1.1 (Revision 3)},
year = {2013},
month = {December},
publisher = {Microsoft},
url = {https://www.microsoft.com/en-us/research/publication/u-prove-cryptographic-specification-v1-1-revision-3/},
note = {Released under the Open Specification Promise (http://www.microsoft.com/openspecifications/en/us/programs/osp/default.aspx)},
}

@inproceedings{UC,
  author       = {Ran Canetti},
  title        = {{Universally Composable Security: {A} New Paradigm for Cryptographic
                  Protocols}},
  booktitle    = {42nd Annual Symposium on Foundations of Computer Science, {FOCS} 2001,
                  14-17 October 2001, Las Vegas, Nevada, {USA}},
  pages        = {136--145},
  publisher    = {{IEEE} Computer Society},
  year         = {2001},
  url          = {https://doi.org/10.1109/SFCS.2001.959888},
  bibsource    = {dblp computer science bibliography, https://dblp.org}
}

@inproceedings{delegatable,
  author       = {Jan Camenisch and
                  Manu Drijvers and
                  Maria Dubovitskaya},
  editor       = {Bhavani Thuraisingham and
                  David Evans and
                  Tal Malkin and
                  Dongyan Xu},
  title        = {{Practical UC-Secure Delegatable Credentials with Attributes and Their Application to Blockchain}},
  booktitle    = {Proceedings of the 2017 {ACM} {SIGSAC} Conference on Computer and
                  Communications Security, {CCS} 2017, Dallas, TX, USA, October 30 -
                  November 03, 2017},
  pages        = {683--699},
  publisher    = {{ACM}},
  year         = {2017},
  url          = {https://doi.org/10.1145/3133956.3134025},
  bibsource    = {dblp computer science bibliography, https://dblp.org}
}

@inproceedings {bind,
author = {Julia Hesse and Nitin Singh and Alessandro Sorniotti},
title = {How to Bind Anonymous Credentials to Humans},
booktitle = {32nd USENIX Security Symposium (USENIX Security 23)},
year = {2023},
isbn = {978-1-939133-37-3},
address = {Anaheim, CA},
pages = {3047--3064},
url = {https://www.usenix.org/conference/usenixsecurity23/presentation/hesse},
publisher = {USENIX Association},
month = aug
}

@misc{wallet-arf,
title = {{The European Digital Identity Wallet Architecture and Reference Framework V2.7.3 }},
year = {2025},
url = {https://eudi.dev/2.7.3/architecture-and-reference-framework-main/},
}

@inproceedings{threshold-decryption,
  title={Privacy-preserving server-supported decryption},
  author={Laud, Peeter and Pankova, Alisa and Vakarjuk, Jelizaveta},
  booktitle={2025 IEEE 38th Computer Security Foundations Symposium (CSF)},
  pages={127--142},
  year={2025},
  organization={IEEE}
}

@inproceedings{atr-with-pairings,
  author       = {Jan Camenisch and
                  Anna Lysyanskaya},
  editor       = {Matthew K. Franklin},
  title        = {{Signature Schemes and Anonymous Credentials from Bilinear Maps}},
  booktitle    = {Advances in Cryptology - {CRYPTO} 2004, 24th Annual International
                  CryptologyConference, Santa Barbara, California, USA, August 15-19,
                  2004, Proceedings},
  series       = {Lecture Notes in Computer Science},
  volume       = {3152},
  pages        = {56--72},
  publisher    = {Springer},
  year         = {2004},
  url          = {https://doi.org/10.1007/978-3-540-28628-8\_4},
  bibsource    = {dblp computer science bibliography, https://dblp.org}
}

@inproceedings{atr-with-pairings-2,
  author       = {Sietse Ringers and
                  Eric R. Verheul and
                  Jaap{-}Henk Hoepman},
  editor       = {Aggelos Kiayias},
  title        = {{An Efficient Self-blindable Attribute-Based Credential Scheme}},
  booktitle    = {Financial Cryptography and Data Security - 21st International Conference,
                  {FC} 2017, Sliema, Malta, April 3-7, 2017, Revised Selected Papers},
  series       = {Lecture Notes in Computer Science},
  volume       = {10322},
  pages        = {3--20},
  publisher    = {Springer},
  year         = {2017},
  url          = {https://doi.org/10.1007/978-3-319-70972-7\_1},
  bibsource    = {dblp computer science bibliography, https://dblp.org}
}

@inproceedings{atr-with-pairings-3,
  author       = {Amira Barki and
                  Solenn Brunet and
                  Nicolas Desmoulins and
                  Jacques Traor{\'{e}}},
  editor       = {Roberto Avanzi and
                  Howard M. Heys},
  title        = {Improved Algebraic MACs and Practical Keyed-Verification Anonymous
                  Credentials},
  booktitle    = {Selected Areas in Cryptography - {SAC} 2016 - 23rd International Conference,
                  St. John's, NL, Canada, August 10-12, 2016, Revised Selected Papers},
  series       = {Lecture Notes in Computer Science},
  volume       = {10532},
  pages        = {360--380},
  publisher    = {Springer},
  year         = {2016},
  url          = {https://doi.org/10.1007/978-3-319-69453-5\_20},
  bibsource    = {dblp computer science bibliography, https://dblp.org}
}

@inproceedings{towards-cloud,
  author       = {Stephan Krenn and
                  Thomas Lor{\"{u}}nser and
                  Anja Salzer and
                  Christoph Striecks},
  editor       = {Srdjan Capkun and
                  Sherman S. M. Chow},
  title        = {{Towards Attribute-Based Credentials in the Cloud}},
  booktitle    = {Cryptology and Network Security - 16th International Conference, {CANS}
                  2017, Hong Kong, China, November 30 - December 2, 2017, Revised Selected
                  Papers},
  series       = {Lecture Notes in Computer Science},
  volume       = {11261},
  pages        = {179--202},
  publisher    = {Springer},
  year         = {2017},
  url          = {https://doi.org/10.1007/978-3-030-02641-7\_9},
  bibsource    = {dblp computer science bibliography, https://dblp.org}
}

@inproceedings{breaking-fixing,
  author       = {Ulrich Hab{\"{o}}ck and
                  Stephan Krenn},
  editor       = {Yi Mu and
                  Robert H. Deng and
                  Xinyi Huang},
  title        = {Breaking and Fixing Anonymous Credentials for the Cloud},
  booktitle    = {Cryptology and Network Security - 18th International Conference, {CANS}
                  2019, Fuzhou, China, October 25-27, 2019, Proceedings},
  series       = {Lecture Notes in Computer Science},
  volume       = {11829},
  pages        = {249--269},
  publisher    = {Springer},
  year         = {2019},
  url          = {https://doi.org/10.1007/978-3-030-31578-8\_14},
  bibsource    = {dblp computer science bibliography, https://dblp.org}
}

@misc{regulation,
      title = {{Regulation (EU) 2024/1183 of the European Parliament and of the Council of 11 April 2024 amending Regulation (EU) No 910/2014 as regards establishing the European Digital Identity Framework}},
      year = {2024},
      url = {https://eur-lex.europa.eu/eli/reg/2024/1183/oj}
}

@techreport{mDL,
address = {Geneva, CH},
type = {Standard},
title = {{Personal identification — ISO-compliant driving licence — Part 5: Mobile driving licence (mDL) application}},
shorttitle = {{ISO}/{IEC} 18013-5:2021},
url = {https://www.iso.org/standard/69084.html},
language = {en},
number = {ISO/IEC 18013-5:2021},
institution = {International Organization for Standardization},
author = {{ISO Central Secretary}},
year = {2021}
}

@techreport{mDL-whitepaper,
type = {Standard},
title = {{The Mobile Driver’s License (mDL) and Ecosystem}},
url = {https://www.securetechalliance.org/wp-content/uploads/Mobile-Drivers-License-WP-FINAL-Update-March-2020-4.pdf},
language = {en},
institution = {Secure Technology Alliance},
author = {{Secure Technology Alliance Identity Council}},
year = {2020}
}

@article{cryptographers,
  title={Cryptographers’ Feedback on the EU Digital Identity’s ARF},
  author={Baum, Carsten and Blazy, Olivier and Camenisch, Jan and Hoepman, Jaap-Henk and Lee, Eysa and Lehmann, Anja and Lysyanskaya, Anna and Mayrhofer, Ren{\'e} and Montgomery, Hart and Nguyen, Ngoc Khanh and others},
  journal={Tech. Rep.},
  year={2024}
}

@misc{blind-splitkey,
      author = {Nikita Snetkov and Mihkel Jaas Karu and Jelizaveta Vakarjuk and Alisa Pankova},
      title = {Coppercloud: Blind Server-Supported {RSA} Signatures},
      howpublished = {Cryptology {ePrint} Archive, Paper 2025/1824},
      year = {2025},
      url = {https://eprint.iacr.org/2025/1824}
}

@inproceedings{UC-cert,
  author       = {Ran Canetti},
  title        = {Universally Composable Signature, Certification, and Authentication},
  booktitle    = {17th {IEEE} Computer Security Foundations Workshop, {(CSFW-17} 2004),
                  28-30 June 2004, Pacific Grove, CA, {USA}},
  pages        = {219},
  publisher    = {{IEEE} Computer Society},
  year         = {2004},
  url          = {https://doi.ieeecomputersociety.org/10.1109/CSFW.2004.24},
  doi          = {10.1109/CSFW.2004.24},
  bibsource    = {dblp computer science bibliography, https://dblp.org}
}

@inproceedings{canetti-pke,
  author       = {Ran Canetti and
                  Shai Halevi and
                  Jonathan Katz},
  editor       = {Joe Kilian},
  title        = {Adaptively-Secure, Non-interactive Public-Key Encryption},
  booktitle    = {Theory of Cryptography, Second Theory of Cryptography Conference,
                  {TCC} 2005, Cambridge, MA, USA, February 10-12, 2005, Proceedings},
  series       = {Lecture Notes in Computer Science},
  volume       = {3378},
  pages        = {150--168},
  publisher    = {Springer},
  year         = {2005},
  url          = {https://doi.org/10.1007/978-3-540-30576-7\_9},
  doi          = {10.1007/978-3-540-30576-7\_9},
  bibsource    = {dblp computer science bibliography, https://dblp.org}
}
